\documentclass[11pt,letterpaper,english,preprintnumbers,amsmath,amssymb,superscriptaddress,nofootinbib,prl]{revtex4}

\usepackage{mathpazo}
\usepackage[T1]{fontenc}
\usepackage[latin9]{inputenc}
\usepackage[letterpaper]{geometry}
\usepackage{verbatim}
\usepackage{amsmath}
\usepackage{amsthm}
\usepackage{amssymb}
\usepackage{graphicx}
\usepackage{setspace}

\makeatletter

\pdfpageheight\paperheight
\pdfpagewidth\paperwidth

\providecommand{\tabularnewline}{\\}

\@ifundefined{textcolor}{}
{%
 \definecolor{BLACK}{gray}{0}
 \definecolor{WHITE}{gray}{1}
 \definecolor{RED}{rgb}{1,0,0}
 \definecolor{GREEN}{rgb}{0,1,0}
 \definecolor{BLUE}{rgb}{0,0,1}
 \definecolor{CYAN}{cmyk}{1,0,0,0}
 \definecolor{MAGENTA}{cmyk}{0,1,0,0}
 \definecolor{YELLOW}{cmyk}{0,0,1,0}
}
  \theoremstyle{remark}
  \newtheorem*{rem*}{\protect\remarkname}
  
  \theoremstyle{Proposition}
 \newtheorem{prop}{Proposition}

  \theoremstyle{remark}
  \ifx\thechapter\undefined
    \newtheorem{rem}{\protect\remarkname}
  \else
    \newtheorem{rem}{\protect\remarkname}[chapter]
  \fi
  \theoremstyle{plain}
  \newtheorem*{thm*}{\protect\theoremname}
  \theoremstyle{plain}
  \ifx\thechapter\undefined
    \newtheorem{lem}{\protect\lemmaname}
  \else
    
  \fi
  \theoremstyle{plain}
  \ifx\thechapter\undefined
    \newtheorem{conjecture}{\protect\conjecturename}
  \else
    \newtheorem{conjecture}{\protect\conjecturename}[chapter]
  \fi

\renewcommand{\Re}{\operatorname{Re}}
\renewcommand{\Im}{\operatorname{Im}}

\makeatother

\usepackage{babel}
  \providecommand{\conjecturename}{Conjecture}
  \providecommand{\lemmaname}{Lemma}
  \providecommand{\remarkname}{Remark}
  \providecommand{\theoremname}{Theorem}

\begin{document}

\title{Eigenpairs of Toeplitz and disordered Toeplitz matrices with a Fisher-Hartwig
symbol}

\author{Ramis Movassagh }
\email{q.eigenman@gmail.com}
\selectlanguage{english}%
\address{IBM Watson Research Center, Yorktown Heights, NY 10598, USA}

\author{Leo P. Kadanoff }
\affiliation{The James Franck Institute, University of Chicago, Chicago Illinois, 60637}
\affiliation{The Perimeter Institute, Waterloo, Ontario, Canada N2L 2Y5\\
(LPK is deceased as of October 26, 2015)}

\date{\today}
\begin{abstract}
\begin{singlespace}
Toeplitz matrices have entries that are constant along diagonals.
They model directed transport, are at the heart of correlation function
calculations of the two-dimensional Ising model, and have applications
in quantum information science. We derive their eigenvalues and eigenvectors
when the symbol is singular Fisher-Hartwig. We then add diagonal disorder
and study the resulting eigenpairs. We find that there is a ``bulk''
behavior that is well captured by second order perturbation theory
of non-Hermitian matrices. The non-perturbative behavior is classified into two classes: Runaways type I leave the complex-valued
spectrum and become completely real because of eigenvalue attraction.
Runaways type II leave the bulk and move very rapidly in response
to perturbations. These have high condition numbers and can be predicted.
Localization of the eigenvectors are then quantified using entropies and inverse participation ratios . Eigenvectors corresponding
to Runaways type II are most localized (i.e., super-exponential),
whereas Runaways type I are less localized than the unperturbed counterparts
and have most of their probability mass in the interior with algebraic
decays. The results are corroborated by applying free probability theory
and various other supporting numerical studies.\end{singlespace}

\end{abstract}
\maketitle
\tableofcontents{}

\begin{singlespace}
\section{\label{sec:Toeplitz}Toeplitz and Toeplitz-Like Matrices }
\end{singlespace}
\subsection{Applications of Toeplitz matrices}
Since their eigenvalues are real, traditionally in sciences one
usually studies models that are Hermitian. In recent years, however,
non-Hermitian models have emerged that capture the nature of various
problems in fluid and plasma physics \cite[ref. therein]{TrefethenEmbree2005},
in biology \cite[review]{Nelson2012} and non-Hermitian integrated
photonics \cite{Ramezani2011}.

A canonical class of non-Hermitian matrices arising in applications
are Toeplitz matrices. A Toeplitz matrix is any square matrix whose entries are constants along diagonals (Eq. \eqref{eq:ToeplitzMatrix}). Since the value of their entries depend on their distance from the diagonal, Toeplitz matrices model
directed transport or propagation of information with strengths that
vary with distance. There are numerous examples of applications of
Toeplitz matrices in condensed matter physics, entanglement theory,
electrical engineering and chemical physics \cite{FishHart1,Korepin2004,McCoyWu1973,GrayReview,Ivanov2013,keating2004}.
The spin-spin correlation function of the two dimensional Ising model
on a square lattice can be written as a Toeplitz determinant \cite{kadanoff1966,McCoyWu1973,Montroll1963}.
``Non-Hermitian quantum mechanics,'' was a phrase coined by Hatano-Nelson
\cite{HatanoNelson1997} and later developed by \cite{FeinbergZee1997,BrezinZee1998,Brouwer1997},
for an effective theory that describes the pinning of vortices in
superconductors of the second type \cite{HatanoNelson1997}. Moreover,
in pure mathematics, Toeplitz operators are useful for proving index
theorems in the framework of non-commutative geometry \cite{connes1990,hurder1991}.

The spectral properties of Toeplitz matrices have fascinated mathematicians
and numerical linear algebraists \cite{TrefethenEmbree2005,BoettcherGrudsky2005}.
Because of their non-hermiticity the eigenpairs can show rich sensitivity
to perturbations. This has inspired development of new mathematics,
most notably pseudo-spectra theory \cite{BottEmbreeSokolov,TrefethenEmbree2005}.
The surprising new features of non-symmetric matrices, compared to
Hermitian matrices, are quite counter-intuitive. Some new mathematical
features are presented in this work as well. 

Toeplitz matrices are generated by a complex function called the symbol
(see the following section). The functional form of the symbol, such
as its singularities, has direct physical implications. For many applications,
the asymptotic behavior of the determinant of the Toeplitz matrix
is of fundamental importance. For smooth symbols Szeg\"o's theorem
gives the asymptotic determinant \cite{Szego1915}; however, if the
symbol contains singularities such as jumps or zeros, the asymptotic
determinant is given by the Fisher-Hartwig theory \cite{FishHart1}.
The latter is not yet fully understood and has many surprising new
features. These matrices have applications in physics \cite{Forrester2004}
and the original conjectures led to various mathematical developments ultimately leading to the proof of Fisher-Hartwig conjectures
\cite{Widom2,Widom7,DeiftItsKrasovsky2012,basor1991fisher,Widom1,BoettcherSilbermann1998,Ehrhart1997}.

Toeplitz matrices have various applications in physics. In quantum information theory, the entanglement entropy
quantifies the entanglement content of a state and for translationally
invariant (i.e., Toeplitz) systems the entanglement properties of
the model are encoded in the symbol \cite{jens_areal}. It turns out
that, in order to understand the correlation and entanglement of such
systems, one needs to grasp the spectral properties such as the determinant.
The idea first appeared in the context of the XX model \cite{Korepin2004}.
The discontinuities of the symbol define the Fermi surface and the
Fisher-Hartwig analysis provides the tools needed for understanding
the asymptotic behavior of the determinant \cite[p.7 and appendix]{jens_areal}.
For more general isotropic models, the number of jumps in the symbol
was argued to be related to the pre-factor in the entanglement scaling
in the conformal charge of the underlying conformal field theory \cite{keating2004}.
When there is no Fermi surface, there is no jump in the symbol and
the system is gapped and non-critical and the system obeys an area
law. However, the discontinuity in the symbol makes the system critical
and the entropy becomes logarithmically divergent \cite{jens_areal,keating2004,Korepin2004,Korepin2005}.

A natural question to consider is what happens to the rich properties
of the singular Toeplitz matrix in the presence of perturbations?
For example, the tight binding model \cite{AshcroftMermin1976} is
an example of a (symmetric) Toeplitz matrix with the symbol $a\left(z\right)=z+1/z$.
Anderson showed that the eigenstates become localized when one adds
random onsite potentials, i.e., a random diagonal matrix \cite{Anderson58}.
M. Kac introduced Toeplitz-like matrices in the context of lattice
vibrations where masses on a one dimensional chain of harmonic oscillators
are random \cite{MKac1968}. In mathematics, structured perturbation
of Toeplitz matrices were studied \cite{BottEmbreeSokolov,TrefethenEmbree2005}. 

In this work we study Toeplitz matrices with singular Fisher-Hartwig
symbols. We first derive the asymptotic form of the eigenvalues and
eigenvectors from Wiener-Hopf factorization. The derivations herein
provide improvements and corrections to our earlier work \cite{Dai2009}.
In particular, the left eigenvectors are shown to have norms that
can be much greater than unity despite standard eigenvectors being
normalized. We then analyze the eigenpairs in the presence of ``onsite''
disorder, i.e., adding a diagonal random matrix. In doing so we draw from mathematical tools of numerical
linear algebra, analysis, non-Hermitian perturbation theory and condensed
matter physics.

\subsection{\label{sub:The-Toeplitz-structure}The Toeplitz structure}
Following \cite{DeiftItsKrasovsky2012} we introduce an $n\times n$
 Toeplitz matrix $T_{n}$ as a matrix with coefficients $\left(T_{n}\right)_{jk}=t_{j-k}$,
$0\le j,k\le n-1$, for some given sequence $\left\{ t_{\ell}\right\} _{\ell\in\mathbb{Z}}$.
And $D_{n}\left(T_n\right)$ denotes the determinant of a Toeplitz matrix
$T_n$. Let $\mathtt{s}=\left\{ z\in\mathbb{C}:\mbox{ }\left|z\right|=1\right\} $
be the unit circle. A {\it symbol} $a\left(\mathtt{s}\right)$ is an integrable function
on the unit circle with Fourier coefficients 
\begin{eqnarray}
t_{\ell} & \equiv & \frac{1}{2\pi}\int_{-\pi}^{\pi}e^{-i\ell p}a\left(e^{-ip}\right)dp,\qquad\ell\in\mathbb{Z}.\label{eq:a_l}\\
 & = & \frac{1}{2\pi i}\oint\frac{a\left(z\right)}{z^{\ell+1}}dz\nonumber 
\end{eqnarray}
The associated Toeplitz matrix $T_n\left(a\right)$ and Toeplitz determinant
$D_{n}\left(a\right)$ are 
\begin{eqnarray*}
T_{n}\left(a\right) & = & \left\{ t_{j-k}\right\} _{0\le j,k\le n-1}\\
D_{n}\left(a\right) & = & \det T_{n}\left(a\right).
\end{eqnarray*}
Incidentally Hankel matrices are of the form $\left\{ t_{j+k}\right\} _{0\le j,k\le n-1}.$
Toeplitz and Hankel matrices are finite sections of closely related
Toeplitz and Hankel {\it operators}, where $0\le n<\infty$ \cite[Section 1]{DeiftItsKrasovsky2012}.
The extension of a Toeplitz matrix to the case where $-\infty<n<\infty$
is called the Laurent operator.

The matrix form of $T_{n}$ is
\begin{equation}
T_{n}\left(a\right)=\left\{ t_{j-k}\right\} _{0\le j,k\le n-1}\mathbf{\rightarrow}\left[\begin{array}{ccccc}
t_{0} & t_{-1} & t_{-2} & \cdots & t_{-\left(n-1\right)}\\
t_{1} & t_{0} & t_{-1} & \ddots & \vdots\\
t_{2} & t_{1} & t_{0} & \ddots & t_{-2}\\
\vdots & \ddots & \ddots & \ddots & t_{-1}\\
t_{\left(n-1\right)} & \cdots & t_{2} & t_{1} & t_{0}
\end{array}\right]\label{eq:ToeplitzMatrix}
\end{equation}
The symbol, defined above, of a Toeplitz matrix or Toeplitz operator or Laurent operator
is the generating function
\[
a\left(z\right)=\sum_{k=-\left(n-1\right)}^{n-1}t_{k}z^{k}\quad.
\]

A circulant matrix is a finite dimensional analogue of a Laurent operator,
in which the entries of the Toeplitz matrix wrap around periodically,
i.e., $t_{i}=t_{-(n-i)}$. We define $\mbox{spec}\left(M\right)$
to be the spectrum (i.e., collection of eigenvalues) of $M$ and $\nu\left(\lambda,a\right)$ to be the \textit{winding
number} of $a\left(\mathtt{s}\right)$ about the point $\lambda$.
In this work we are only concerned with matrices, but for the sake of
concreteness we summarize what is known about the spectral properties
in the table below, which is taken from the book by Trefethen and Embree \cite[Theorem 7.1]{TrefethenEmbree2005}. In this table,  $\mathtt{s}_{N}=\left\{ z\in\mathbb{C}:\mbox{ }z^{N}=1\right\} $ is
the subset of $\mathtt{s}$ that corresponds to the roots of unity.
\begin{center}
\begin{tabular}{|l|}
\hline 
\textbf{Spectra of Toeplitz and Laurent Operators}\tabularnewline
\hline 
\hline 
\textit{Let $T$ be a circulant matrix or Laurent or Toeplitz operator
with continuous symbol $a$.}\tabularnewline
\hline 
$\qquad$(i) If $T$ is a \textit{circulant matrix}, then $\mbox{spec}\left(T\right)=a\left(\mathtt{s}_{N}\right)$ \tabularnewline
\hline 
$\qquad$(ii) If $T$ is a \textit{Laurent operator}, then $\mbox{spec}\left(T\right)=a\left(\mathtt{s}\right)$ \tabularnewline
\hline 
$\qquad$(iii) If $T$ is a \textit{Toeplitz operator} with symbol
continuous on $\mathtt{s}$, then $\mbox{spec}\left(T\right)$ is
equal to $a\left(\mathtt{s}\right)$ \tabularnewline
together with all the points enclosed by this curve with nonzero winding
numbers \tabularnewline
\hline 
\end{tabular}
\par\end{center}

To better appreciate (iii), let $T$ have a continuous symbol $a(z)$
and let $\lambda\in\mathbb{C}$ be any number with $\nu\left(\lambda,a\right)<0$.
Then $\lambda\in\mbox{spec}\left(T\right)$ and is actually an eigenvalue
of $T$ with an eigenvector $|\psi\rangle=\left\{ \psi_{j}\right\} $
whose amplitude decreases as $j\rightarrow\infty$; if $a$ is a rational
function then the decrease is exponential. These are called boundary
eigenvectors or boundary eigenmodes. Similarly, if $\nu\left(\lambda,a\right)>0$,
$T$ does not have boundary eigenmodes, but its transpose $T^{\top}$
does; $T^{\top}$ is also a Toeplitz operator with the symbol $a\left(z^{-1}\right)$.
This implies that $\lambda\in\mbox{spec}\left(T\right)$. For $T_{n}$,
i.e., $n\times n$ Toeplitz matrix with the same symbol curve, $\nu\left(a,\lambda\right)<0$,
$\nu\left(a,\lambda\right)>0$ correspond to eigenmodes attached to
the left and right boundaries respectively (see \cite[Chapter 7]{TrefethenEmbree2005}
for a detailed discussion). 

In earlier applications, such as correlation function calculations
of two dimensional Ising model, the asymptotic Toeplitz determinant
was the main object of study. In the case that $a\left(z\right)$
is smooth, this determinant is given by Szeg{\"o}'s theorem. If the
symbol contains singularities such as zeros or jumps, it is given
by Fisher-Hartwig theory \cite{basor1991fisher,Ehrhart1997,DeiftItsKrasovsky2012}.

Below we only consider the $n\times n$ Toeplitz matrix and
drop the subscript $n$ on $T$ for notational simplicity.

\subsection{\label{sec:Previous}Fisher-Hartwig symbols}
The story of Toeplitz matrices is intertwined with the problem of
two dimensional Ising model (See \cite{DeiftItsKrasovsky2012} for
an overview). Fisher and Hartwig \cite{FishHart1} introduced a class
of singular symbols for Toeplitz determinants. The symbols of Fisher-Hartwig
have the following general form \cite[Section 6]{DeiftItsKrasovsky2012} 
\begin{eqnarray}
a\left(z\right) & = & e^{V\left(z\right)}z^{\sum_{j=0}^{m}\beta_{j}}\prod_{j=0}^{m}\left|z-\bar{z}_{j}\right|^{2\alpha_{j}}g_{z_{j},\beta_{j}}\left(z\right)\bar{z}_{j}^{-\beta_{j}};\label{eq:HartwigFisherGeneral}\\
 &  & z=e^{-ip},\quad0\le p\le2\pi\nonumber 
\end{eqnarray}
for some $m=0,1,2,\cdots$, where 
\begin{eqnarray*}
\bar{z}_{j} & = & e^{-ip_{j}},\;j=0,1,\dots,m\quad0=p_{0}<p_{1}<\cdots<p_{m}<2\pi,\\
g_{z_{j}\beta_{j}\left(z\right)} & \equiv & g_{\beta_{j}}\left(z\right)=\left\{ \begin{array}{c}
e^{i\pi\beta_{j}},\qquad0\le\arg z<p_{j},\\
e^{-i\pi\beta_{j}},\qquad p_{j}\le\arg z<2\pi
\end{array}\right.\\
\Re\alpha_{j}>-\frac{1}{2}, &  & \beta_{j}\in\mathbb{C},\qquad j=0,1,2,\dots,m,
\end{eqnarray*}
 and $V\left(e^{-ip}\right)$ is a sufficiently smooth function on
the unit circle. The condition on $\Re\alpha_{j}$ ensures integrability. 

This symbol is of fundamental importance for analysis of Toeplitz
matrices in general. Previously \cite{LPK2009,Dai2009} investigated
the particular Fisher-Hartwig symbol (see Fig. \ref{fig:ImageSymbol})
\begin{equation}
a_{\alpha,\beta}\left(z\right)=\left(-1\right)^{\alpha+\beta}\left(\frac{z-1}{z}\right)^{2\alpha} z^\beta=\left(2-z-\frac{1}{z}\right)^{\alpha}\left(-z\right)^{\beta}\quad.\label{eq:LeoSymbol}
\end{equation}

On the unit circle $z=e^{-ip}$, the factor $\left(-z\right)^{\beta}$
has a jump discontinuity and $\left(2-z-\frac{1}{z}\right)^{\alpha}=\left(2-2\cos p\right)^{\alpha}$
is a function that may have a zero, a pole, or a discontinuity of
oscillating type. 

The elements of the Toeplitz matrix are related to the Fourier transformation
of the symbol by 
\begin{eqnarray}
t_{j,k} & = & t_{j-k}=\int_{\mathtt{s}}\frac{dz}{2\pi i}\frac{a\left(z\right)}{z^{j-k+1}}=\left(-1\right)^{\beta+\alpha}\int_{\mathtt{s}}\frac{dz}{2\pi i}\frac{\left(z-1\right)^{2\alpha}}{z^{j-k-\beta+\alpha+1}}.\label{eq:LeoToeplitz}
\end{eqnarray}

We demand integrability, i.e., $\Re\left(\alpha\right)>-1/2$. After
integration of Eq. \eqref{eq:LeoToeplitz} and ignoring an overall constant
multiple of $\left(-1\right)^{\alpha+\beta}\exp\left[i\pi\left(\alpha+\beta\right)\right]$;
i.e., choosing the trivial root, we have 
\begin{eqnarray}
t_{j-k} & = & \left(-1\right)^{j-k}\frac{\mbox{ }\Gamma\left(2\alpha+1\right)}{\Gamma\left(\alpha+\beta+1-j+k\right)\Gamma\left(\alpha-\beta+1+j-k\right)},\label{eq:T_leo_Matrix}
\end{eqnarray}
where we used the generalized binomial theorem for $\left(z-1\right)^{2\alpha}$
to perform the integral. 

Below we shall use the following properties of the Gamma functions
\begin{eqnarray*}
\Gamma\left(n+1\right) & = & n\Gamma\left(n\right)\\
\Gamma\left(-x\right) & = & \frac{-\pi}{x\Gamma\left(x\right)\sin\left(\pi x\right)}\\
\Gamma\left(\epsilon-n\right) & = & \left(-1\right)^{n-1}\frac{\Gamma\left(-\epsilon\right)\Gamma\left(1+\epsilon\right)}{\Gamma\left(n+1-\epsilon\right)}\\
\lim_{n\rightarrow\infty}\frac{\Gamma\left(n+\alpha\right)}{\Gamma\left(n\right)n^{\alpha}} & = & 1.
\end{eqnarray*}

Because of the Toeplitz structure, it is sufficient to specify the
first row and first column of $T$ to fully specify the matrix. Using
these identities, Eq. \eqref{eq:LeoToeplitz} becomes
\begin{eqnarray*}
t_{-k} & = & \frac{\Gamma\left(2\alpha+1\right)\sin\left(\pi\epsilon_{r}\right)}{\pi}\frac{\Gamma\left(k+1-\epsilon_{r}\right)}{\Gamma\left(k+\epsilon_{c}\right)},\quad\epsilon_{r}\equiv\alpha-\beta+1\\
t_{j} & = & \frac{\Gamma\left(2\alpha+1\right)\sin\left(\pi\epsilon_{c}\right)}{\pi}\frac{\Gamma\left(j+1-\epsilon_{c}\right)}{\Gamma\left(j+\epsilon_{r}\right)},\quad\epsilon_{c}\equiv\alpha+\beta+1.
\end{eqnarray*}

Let $r\equiv j-k$, then the limits of Eq. \eqref{eq:T_leo_Matrix},
ignoring terms of order $\mathcal{O}\left(1/r^{2\left(\alpha+1\right)}\right)$
and higher, are
\begin{eqnarray*}
\lim_{r\gg1}t\left(r\right) & = & \frac{\Gamma\left(2\alpha+1\right)}{\pi r^{2\alpha+1}}\sin\left[\pi\left(\alpha+\beta\right)\right]\left\{ 1+\mathcal{O}\left(\frac{1}{r}\right)\right\} \\
\lim_{r\ll-1}t\left(r\right) & = & \frac{\Gamma\left(2\alpha+1\right)}{\pi\left|r\right|^{2\alpha+1}}\sin\left[\pi\left(\alpha-\beta\right)\right]\left\{ 1+\mathcal{O}\left(\frac{1}{r}\right)\right\} .
\end{eqnarray*}

For $\alpha=1/3$  and $\beta=-1/2$ the foregoing equations show that the entries of $T$ decay algebraically away from the diagonal with super-diagonals being negative and sub-diagonals positive.

\subsection{\label{sec:summary}Outline and summary of the main results}
The theory of Toeplitz matrices is well developed \cite{TrefethenEmbree2005,Dai2009}.
When the symbol is singular and the Toeplitz matrix is finite, often
analytical results are lacking. In Section \ref{sec:Eigenvalues} we derive the eigenpairs of $T_n$ with the symbol \eqref{eq:LeoSymbol}. Namely, we analytically derive formulas for the eigenpairs of the finite Toeplitz matrix with a singular Fisher-Hartwig symbol, which extends and improves the previous work \cite{Dai2009}.

The randomly perturbed, non-symmetric, Toeplitz matrix is rarely considered. Suppose we add
diagonal disorder to the Toeplitz matrix; Mark Kac called such matrices
\textit{Toeplitz-like} \cite{MKac1968}. We consider 
\[
T\left(\sigma\right)\equiv T+\sigma V\label{eq:TSigma}
\]
where $T$ is as above, $V$ is a diagonal random matrix and $\sigma$
is some real parameter that quantifies the strength of the perturbation.
In Section \ref{sec:Eigenvalues} we focus on $\sigma=0$ and in Section \ref{sec:Eigenvectors} 
we extend our work to $0\le\sigma<\left\Vert V\right\Vert /\left\Vert T\right\Vert $, where $\sigma V$ is seen as a perturbation of $T$.
 The eigenvalues of the perturbed Toeplitz matrix are classified
into three categories: 
\begin{enumerate}
\item The Bulk eigenpairs (Subsection \ref{sub:Perturbation-theory-of}): The eigenpairs are well approximated by the second
order perturbation theory of non-Hermitian matrices, which we calculate analytically.
\item Runaways type I  (Subsection \ref{sub:Non-Perturbative-TypeI}): First class of nonperturbative eigenpairs. The eigenvalues
that are initially near the real line, become exactly real in response
to small perturbations. 
\item Runaways type II  (Subsection \ref{sub:Runaways-Type-II}): Second class of nonperturbative eigenpairs. The
eigenvalues are all ill-conditioned and move substantially in the complex plane in response to
small perturbations. There are a number of related conjectures that
we list in Subsection \ref{sub:Runaways-Type-II} .
\end{enumerate}
The corresponding eigenvectors also fall into the same three classes.
In Section \ref{sec:EigenvectorsDis}, we show a correspondence between
the eigenvalues and eigenvectors. We denote the $j^{\mbox{th}}$ component of the $\ell^{\mbox{th}}$ eigenvectors by $\psi^{\ell}_j$. We call the $j=0$ component of any eigenvector its boundary and $0\ll j\ll n$ the interior. We summarize our findings in the following table:
\begin{center}
\begin{tabular}{|c|l|l|}
\hline 
 & Eigenvalues & Eigenvectors\tabularnewline
\hline 
\hline 
$\sigma=0$ & For $n$ large: Approximately the image of the symbol & $\psi_{j}^{\ell}\propto\exp\left(ip^{\ell}j\right)$; $\quad p^{\ell}=2\pi\ell/n+i(2\alpha+1)\ln n/n$\tabularnewline
\hline 
 & Bulk: Second order perturbation theory & Exponential decay: maximum at the boundary\tabularnewline
\cline{2-3} 
$\sigma>0$ & Runaways type I: Attraction of complex conjugates & Algebraic decay:  maximum in the interior\tabularnewline
\cline{2-3} 
 & Runaways type II: Large condition numbers & Super-exponential decay: maximum at the boundary\tabularnewline
\hline 
\end{tabular}
\par\end{center}
\begin{rem}
In our numerical work below, for the sake of concreteness, we make
the choice of $\alpha=1/3$ and $\beta=-1/2$ so the various plots and
arguments are comparable. This particular choice of $\alpha$ and $\beta$ is explained in the following section and was previously used \cite{Dai2009}. 
\end{rem}
\begin{figure}
\begin{centering}
\includegraphics[scale=0.3]{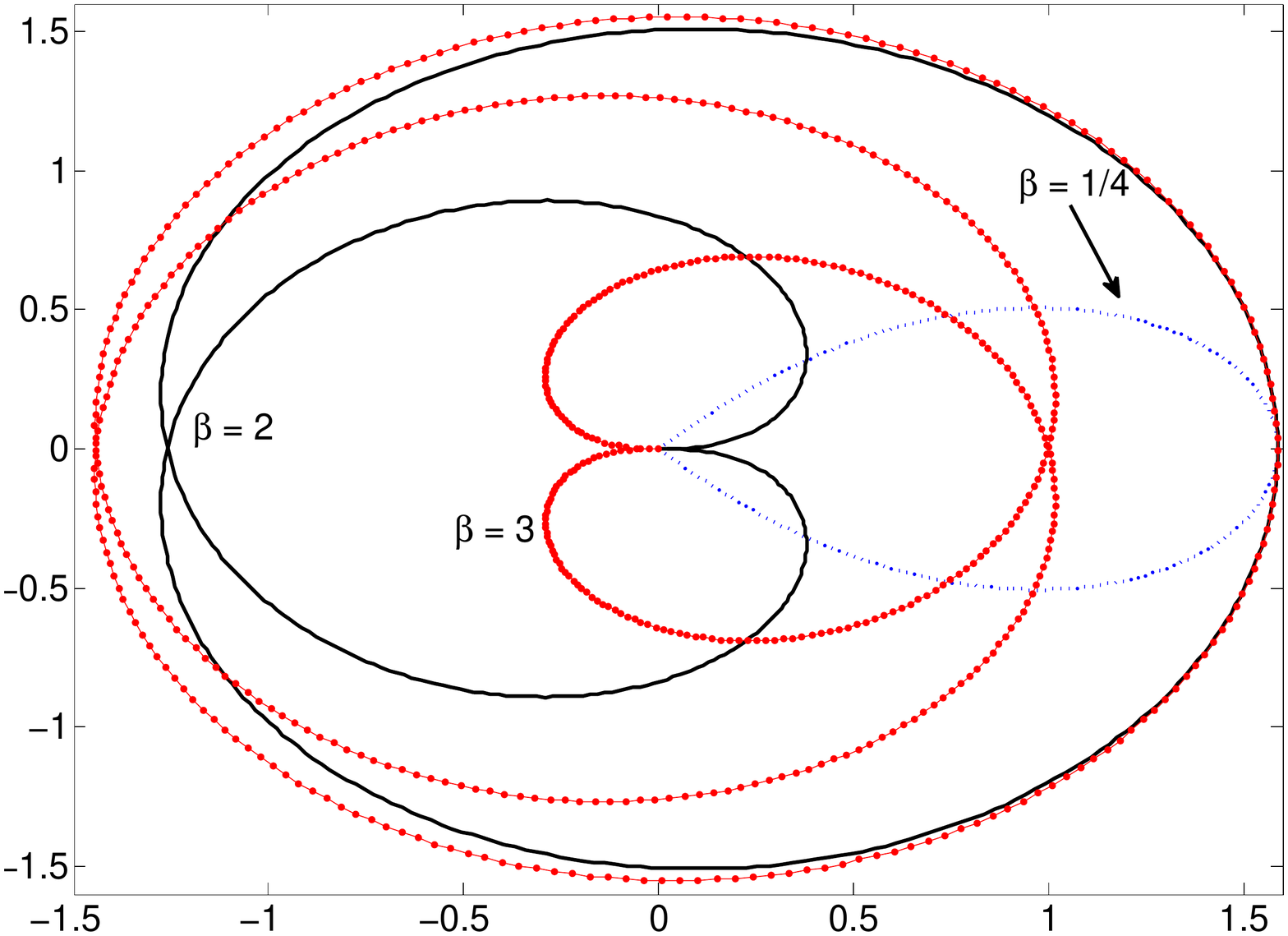}\includegraphics[scale=0.3]{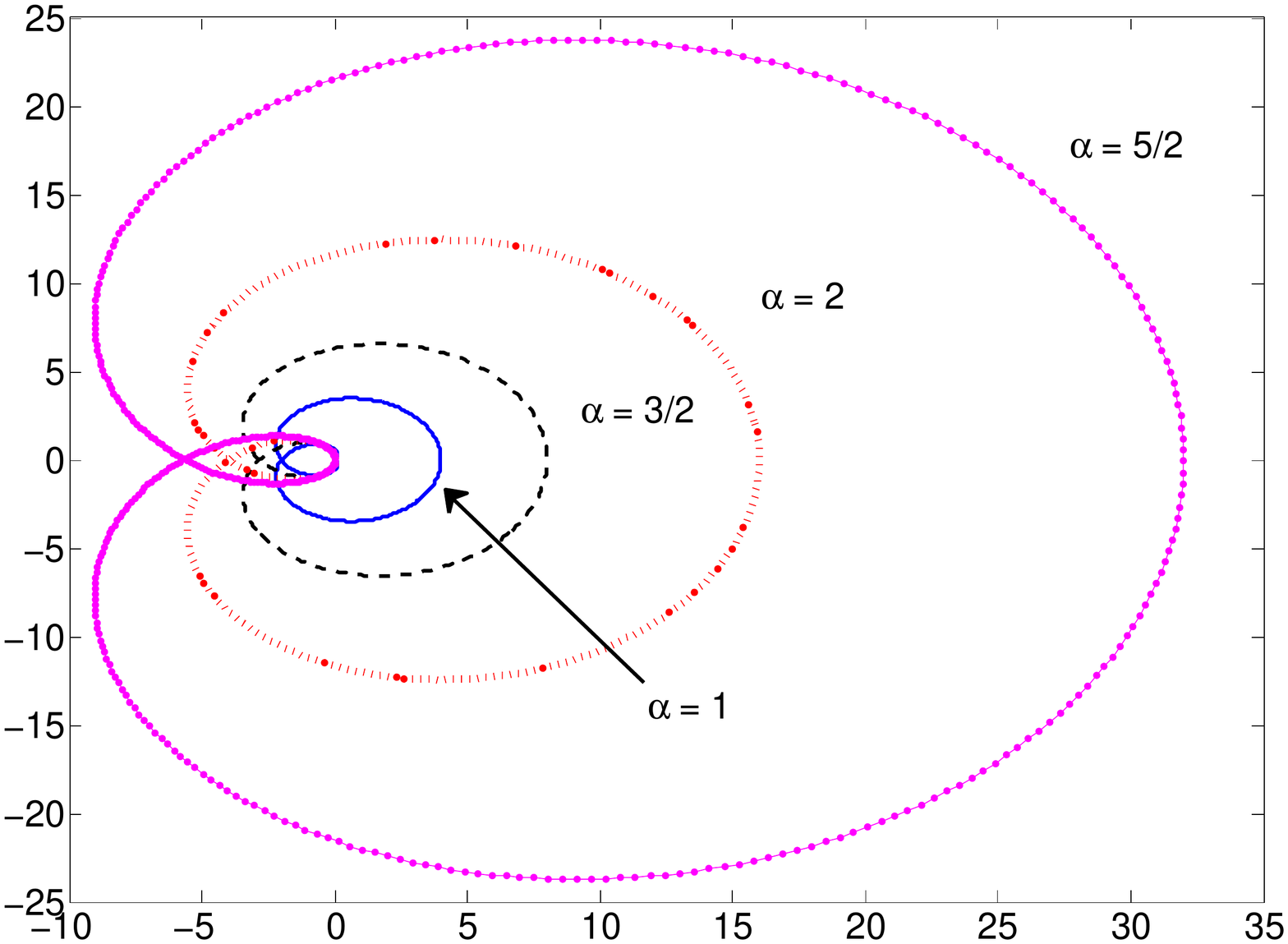}
\par\end{centering}
\centering{}\caption{\label{fig:ImageSymbol}Image of the symbol $a\left(\mathtt{s}\right)$  given by Eq. \eqref{eq:LeoSymbol}
for various $\alpha,\beta$. Left: $\alpha=1/3$  and $\beta$ is varied. Right: $\beta=2$ and $\alpha$ is varied.}
\end{figure}
\begin{rem}
In the plots, the vertical is the imaginary and the horizontal is
the real axis respectively unless stated otherwise. All of the simulations
and plots were done in MATLAB.
\end{rem}

\section{\label{sec:Eigenvalues}No Disorder $\sigma=0$}
\subsection{Eigenvalues}
The spectrum of a Laurent operator whose entries $t_{j-i}$ are defined
for  $-\infty<j-i<+\infty$ is exactly the image of the symbol $a\left(e^{-ip}\right)$,
where $p\in\left[-\pi,\pi\right)$. This is easily seen if one looks
at the Fourier representation of $T$ (recall it is translationally
invariant). The problem is more complicated for the semi-infinite
and even more difficult for the finite sections where $-n\le j-i\le n$. 

There is a large literature on finite sections of Toeplitz operators
as $n\rightarrow\infty$. In particular, singular values converge
to their infinite dimensional counterparts but the eigenvalues may
not \cite[p. 61]{TrefethenEmbree2005}. However, in the limit $n\rightarrow\infty$,
for many classes of symbols, the spectra approach
the image of the symbol on the unit circle \cite{Widom7}. For example symbols containing
a single jump discontinuity belong to this class \cite{Widom7}. In
\cite{Dai2009} it was shown that the eigenvalues of $T$ are distributed
according to $a\left(\exp\left(-ip^{\ell}\right)\right)$, where the
real parts of $p^{\ell}$ are uniformly distributed on the interval
$\left[0,2\pi\right)$. In Fig \ref{fig:ImageSymbol}, we show the
qualitative dependence of the image of the symbol (Eq. \eqref{eq:LeoSymbol})
 on $\alpha$ and $\beta$.

A challenge in studying the eigenvalues of a general Toeplitz matrix
is the non-Hermiticity. First, the eigenvalues are in general complex
and a priori one does not have a natural way of ordering and labeling
them. Second the eigenvectors cannot be taken to be an orthonormal
set and one has to carefully analyze the left eigenvectors as well.
The latter can have arbitrary norms rendering ill-conditioned and nonperturbative behavior
as we will show. 

For the $n\times n$ Toeplitz matrix in Eq. \eqref{eq:T_leo_Matrix},
for large $n$, it was shown that the "momenta", eigenvalues and eigenvectors respectively are \cite{Dai2009,Lee2007} 
\begin{eqnarray}
p^{\ell} & = & \frac{2\pi\ell}{n}+i\left(2\alpha+1\right)\frac{\ln n}{n}+O\left(1/n\right),\label{eq:p_l}\\
E^{\ell} & = & a\left(\exp\left(-ip^{\ell}\right)\right)+o\left(1/n\right)\label{eq:E_l}\\
\psi_{j}^{\ell} & \propto & \exp(ip^{\ell}j)\label{eq:psi_l},
\end{eqnarray}
where $j$ refers to the $j^{\mbox{th}}$ component of the eigenvector. In
these works the semi-infinite Toeplitz matrix  (i.e., $j=0,\cdots,\infty$)
was used to analytically derive Eqs. (\ref{eq:p_l}-\ref{eq:psi_l}). It was then argued that in
the finite case and for sufficiently large $n$ in the regime  $0\ll\ell\ll n$,
the eigenpairs are well approximated by the semi-infinite results.

Note that in Eq. \eqref{eq:p_l}, $p^{\ell}$ has a small imaginary
part which for $0\ll\frac{j}{n}\ll1$ produces an exponential decay
of the wave-function $\psi_j^\ell$ as $j/n$ increases. Hence the eigenfunctions
are localized and have their maxima near $j=0$. It is also important
to notice that $p^{\ell}$'s are roughly equally spaced (see Fig. \ref{fig:Diff_of_p})
\begin{figure}
\centering{}\includegraphics[scale=0.3]{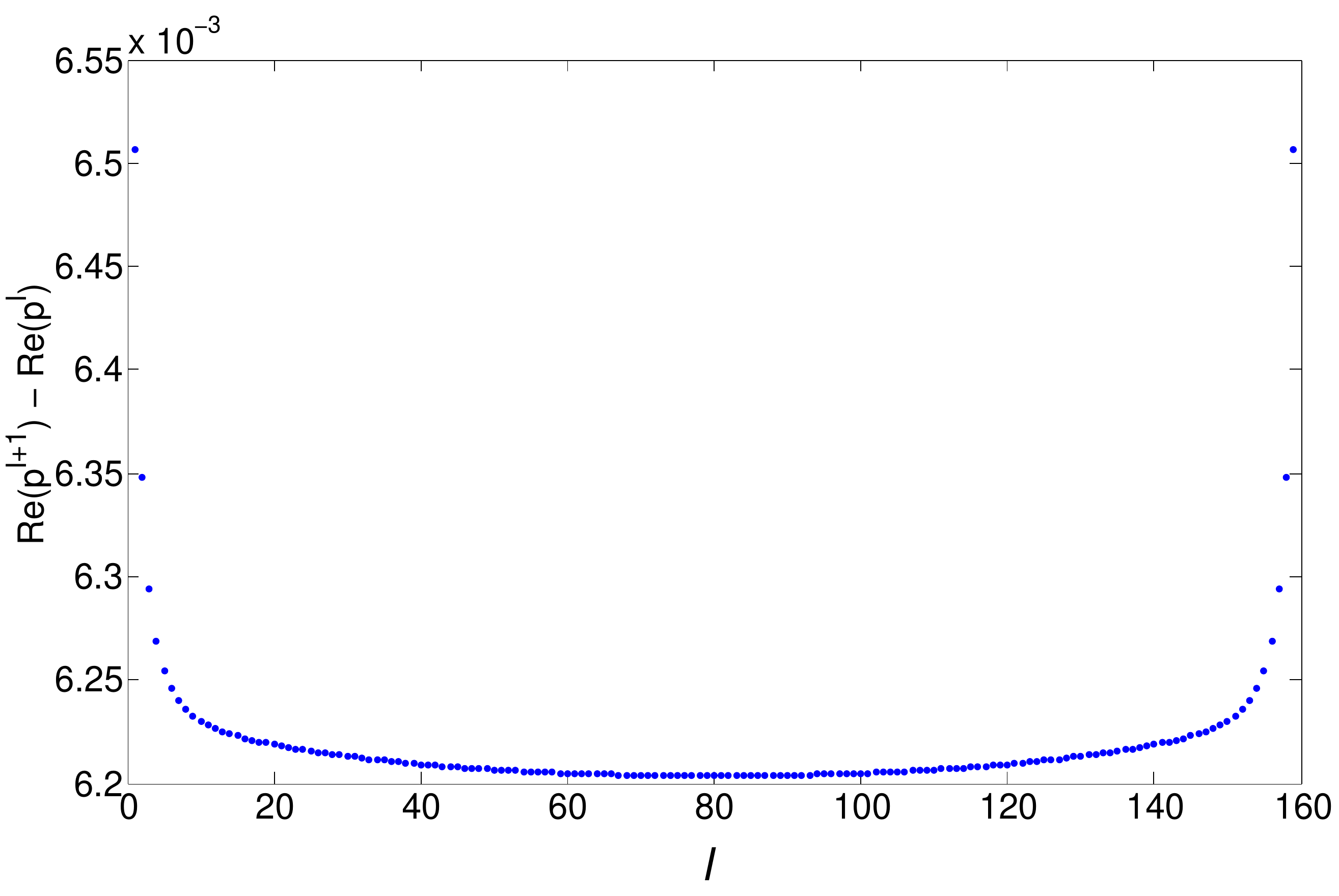}\includegraphics[scale=0.3]{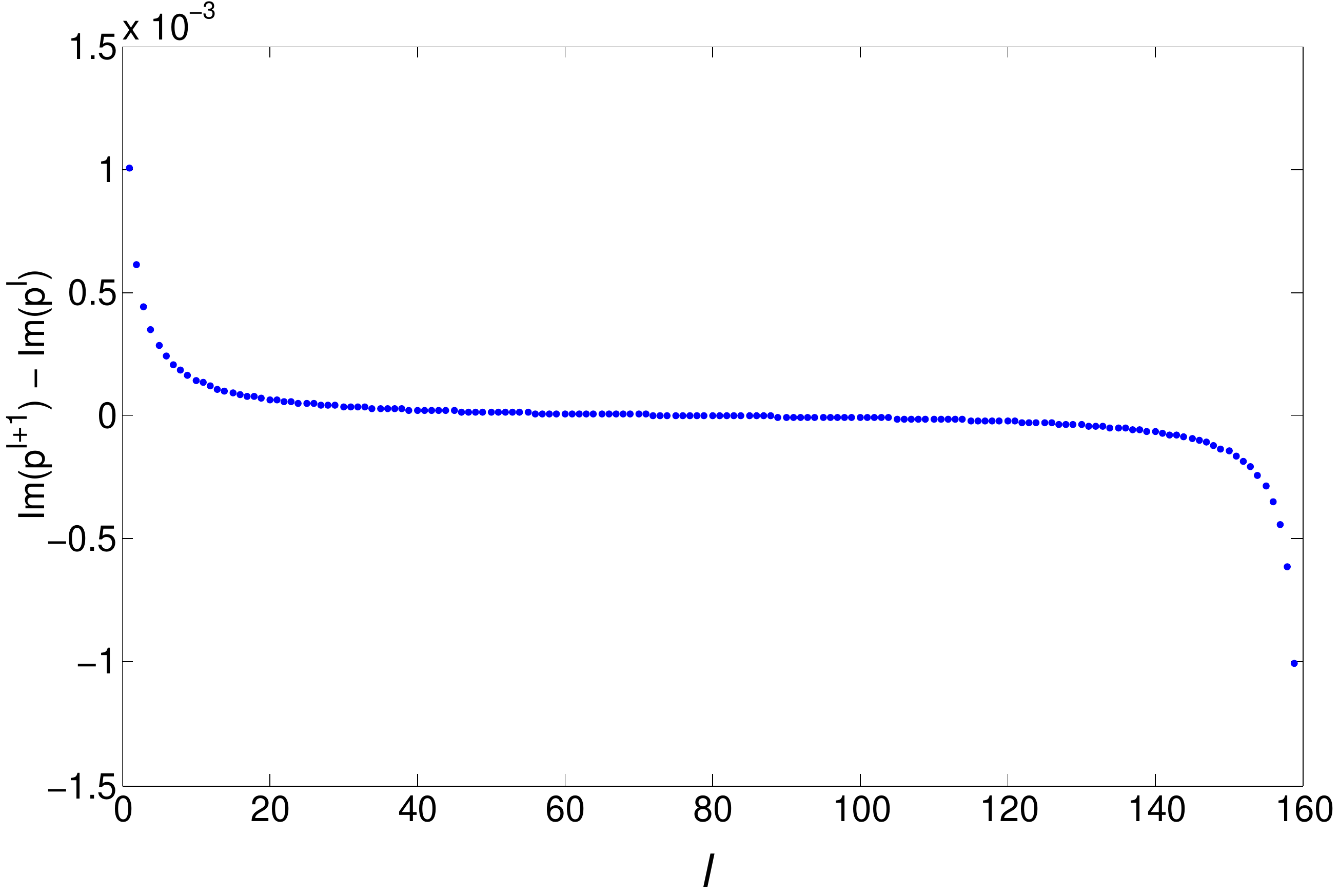}\caption{\label{fig:Diff_of_p}Real and imaginary parts of $p^{\ell+1}-p^{\ell}$
for the Toeplitz matrix $T$ (Eq. \eqref{eq:T_leo_Matrix}) of size $160$ with $\alpha=1/3$ and $\beta=-1/2$. As discussed in the text,
this difference is nearly constant away from the boundaries, where $0\ll\frac{j}{n}\ll1$.}
\end{figure}

In \cite{Dai2009} using a quasi-particle picture analogous to Landau's
Fermi-liquid theory it was argued that difference of $p$ values,
away from the two ends, is nearly constant and independent of $\ell$.
That is, $p^{\ell+1}-p^{\ell}\approx\frac{2\pi}{n}$, for all $\ell$
(Fig. \ref{fig:Diff_of_p}). One pictures eigenfunctions with momenta
that increase by $2\pi/n$; one wants to fit in a wavelength as $\ell$
increases by one. To illustrate this in Fig. \ref{fig:Diff_of_p},
we plot the real and imaginary parts of $p^{\ell+1}-p^{\ell}$ by
first extracting the eigenvalues using numerical exact diagonalization in MATLAB. We then solve for $p^{\ell}$'s
that are implicit function of $E^{\ell}$'s via Eq. \eqref{eq:LeoSymbol}
using the MATLAB function $\mathtt{solve}$ \footnote{In passing $\alpha$ and $\beta$ in Eq. \eqref{eq:LeoSymbol} into
$\mathtt{solve}$ we had to use the function $\mathtt{num2str}$ which
converts a number into a string with roughly $4$ digits of precision.
Because of the exponential dependence on $p^{\ell}$'s this can cause
jitters in the values of $p^{\ell}$ in the plots$ $ shown in Fig.
\ref{fig:Diff_of_p}. To fix it one can change the precision by using
$\mathtt{num2str}\left(\alpha,16\right)$ to get $16$ digits of accuracy
in the value of $\alpha$. Similarly for $\beta$.}.
\begin{figure}
\centering{}\includegraphics[scale=0.35]{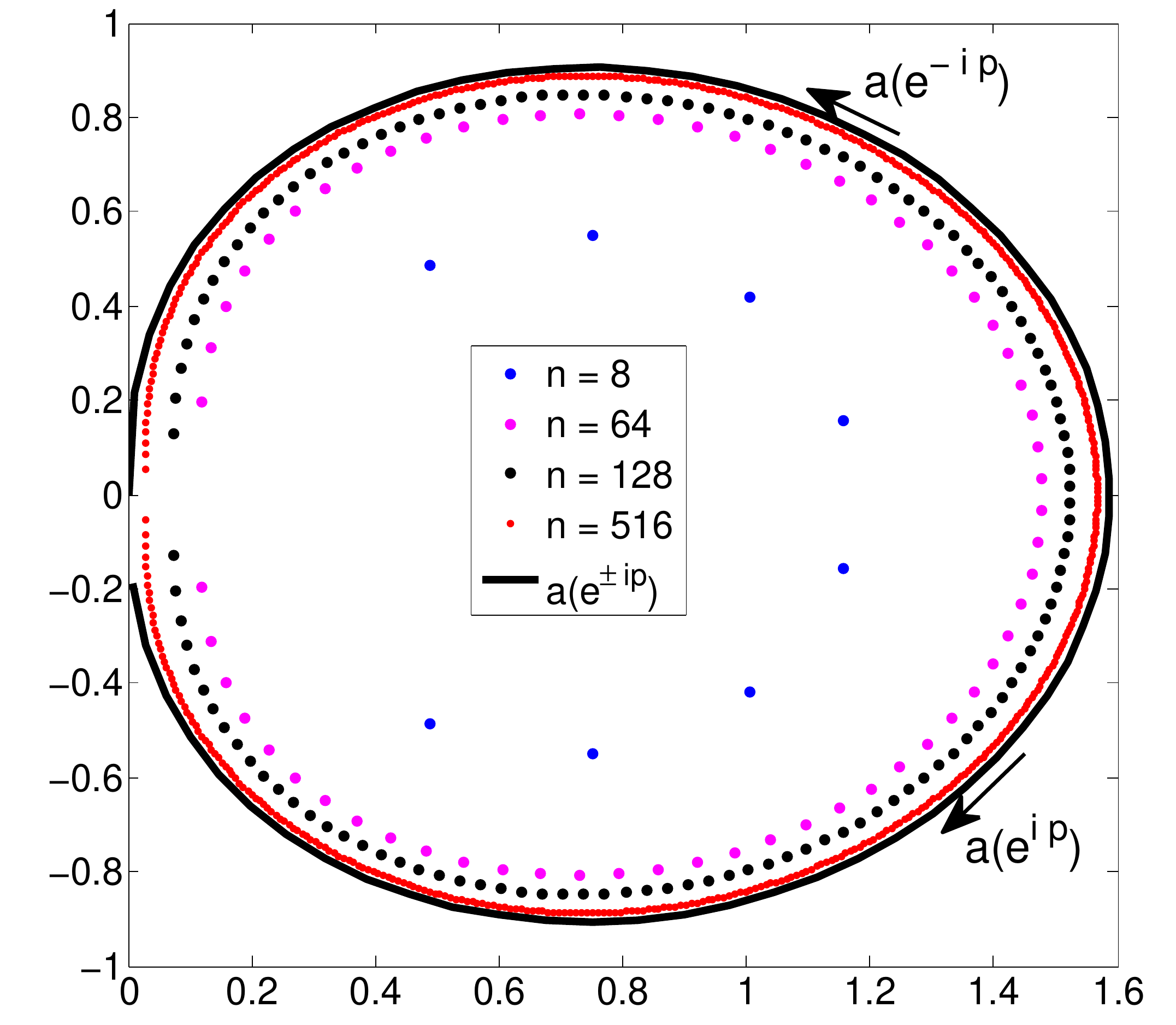}\caption{\label{fig:Eigs_finiten_imageofSym}The eigenvalues of the finite
Toeplitz matrix $T$ (Eq. \eqref{eq:T_leo_Matrix}) with $\alpha=1/3$ and $\beta=-1/2$ along with the
image of the symbol for various sizes of the matrix $n$. We put arrows to show {\it winding} of $a\left(z\right)$. }
\end{figure}
\begin{rem}
As mentioned above, it is not \textit{a priori} clear how one should
 index the eigenvalues. We found that the best way is to order them
according to the real part of $p^{\ell}$. In the following plots
the individual eigenvalues and their corresponding eigenvectors have
been labeled for a close analysis. These labelings are in one-to-one
correspondence with increasing order of the real part of $p^\ell$ (Eq.
\ref{eq:p_l}) and play a central part in our analysis.
\end{rem}

The actual spectrum of the finite Toeplitz matrix lies inside the image of the symbol \cite{Dai2009,Lee2007,BoettcherSilbermann1998}.
Therefore, in this subsection we take $n\gg1$ and think of the eigenvalues
$\left\{ E^{1},E^{2},\cdots,E^{\ell},\cdots,E^{n}\right\} $ as being
close to yet inside $a_{\alpha,\beta}\left(e^{ip}\right)$ (Fig. \ref{fig:Eigs_finiten_imageofSym}).

Starting from Eq. \eqref{eq:LeoSymbol} and using $z=\exp\left(-ip\right)$,
the rate of change of the symbol with $p$ can be calculated 
\[
\frac{da_{\alpha,\beta}\left(z\right)}{dp}=-\left(\alpha\tan\frac{p}{2}+i\beta\right)\mbox{ }a_{\alpha,\beta}\left(z\right)\quad.
\]
where, to be explicit, we put the subscripts $\alpha$ and $\beta$
on the symbol. This equation, in principle, shows the variation of
the eigenvalues with $p$ in the infinite case. In particular, recall
that $p$ is uniform on $\left[0,2\pi\right)$, and assume that the
image of the symbol is some bounded region in the complex plane. We
see that the eigenvalues $a_{\alpha,\beta}\left(\exp\left(-ip\right)\right)$
are far less dense when $p\approx\pi$ because $\tan\left(p/2\right)$
diverges and nearby eigenvalues get pulled apart arbitrarily fast;
this can be seen near the origin in Fig. \ref{fig:Eigs_finiten_imageofSym}.

We drop the subscripts $\alpha$ and $\beta$ when the symbol is given by Eq. \eqref{eq:LeoSymbol}. We now calculate $E^{\ell}$ in Eq. \eqref{eq:E_l}. Since $a(z)$  in
Eq. \eqref{eq:LeoSymbol} is equal to 
\[
a(z)=(z-1)^{2\alpha}z^{\beta-\alpha}\exp\left[-i(\alpha+\beta)\pi\right]
\]
ignoring $\ln n/n$ terms we have 
\begin{eqnarray}
E^{\ell} & \approx & (-1)^{\beta}4^{\alpha}\sin^{2\alpha}\left(\frac{\pi\ell}{n}\right)\exp\left(-2i\pi\beta\ell/n\right)\label{eq:E_Final_Approx}
\end{eqnarray}
Comment: Plot of the foregoing equation agrees well with the exact
eigenvalues of $T$.

\subsection{Eigenvectors}

The eigenvalues of the Toeplitz matricies we consider are simple. Consequently,
the matrices are not defective and the standard eigenvectors (also
called right eigenvectors) are linearly independent. Denote by $|\psi^{\ell}\rangle$
and $\langle\tilde{\psi}^{\ell}|$ the right (standard) and left
eigenvectors of the eigenvalue $E^{\ell}$ respectively. The eigenvalue equations are: 
\begin{eqnarray}
T\mbox{ }|\psi^{\ell}\rangle & = & E^{\ell}\mbox{ }|\psi^{\ell}\rangle\label{eq:EigValPro_R}\\
\langle\tilde{\psi}^{\ell}|\mbox{ }T & = & E^{\ell}\mbox{ }\langle\tilde{\psi}^{\ell}|\;.\label{eq:EigValProbL}
\end{eqnarray}
Note that the eigenvalue is the same in both equations. We normalize
the eigenvectors such that $||\mbox{ }|\psi^{\ell}\rangle||_{2}=1$
for all $\ell$. Let the $n\times n$ matrix $\Psi=\left[\mbox{ }|\psi^{1}\rangle\mbox{ }|\psi^{2}\rangle\mbox{ }\cdots\mbox{ }|\psi^{n}\rangle\mbox{ }\right]$
have columns that are the right eigenvectors, which is invertible because the spectrum is simple. Then the $\ell^{\mbox{th}}$
left eigenvector, $\langle\tilde{\psi}^{\ell}|$, is the $\ell^{\mbox{th}}$
row of $\Psi^{-1}$.

Below for simplicity we sometimes denote the $j^{\mbox{th}}$ component of
$|\psi^{\ell}\rangle$ and $\langle\tilde{\psi}^{\ell}|$
by $\psi_{j}^{\ell}$ and $\tilde{\psi}_{j}^{\ell}$ respectively, where $j\in [0\dots n-1]$. 

Taking the transpose of Eq. \eqref{eq:EigValPro_R} we obtain $\sum_{j}\left(\psi^{\ell}\right)_{j}^{\top}t_{ji}=E^{\ell}\left(\psi^{\ell}\right)_{i}^{\top}$,
which implies that $\left(\psi^{\ell}\right)^{\top}$ is a left
eigenvector. But because of the Toeplitz structure, the components
of the $\ell^{\mbox{th}}$ left eigenvector, denoted by $\tilde{\psi}_{j}^{\ell}$,
are proportional to 
\[
\tilde{\psi}_{j}^{\ell}\propto\psi_{n-j-1}^{\ell}.
\]

We build a dual basis from the left eigenvectors that is 
\begin{equation}
\langle\tilde{\psi}^{\ell}|\psi^{k}\rangle=\delta_{\ell,k},\label{eq:orthogonality}
\end{equation}
where $\delta_{\ell,k}$ is the Kronecker delta. This implies that the components of an left eigenvector are
\begin{eqnarray}
\tilde{\psi}_{j}^{\ell} & = & c^{\ell}\psi_{n-j-1}^{\ell}\label{eq:LEigVec_Compon}\\
c^{\ell} & \equiv & \left(\sum_{j=0}^{n-1}\psi_{n-j-1}^{\ell}\mbox{ }\psi_{j}^{\ell}\right)^{-1}\quad,\nonumber 
\end{eqnarray}
where $c^{\ell}$ is a normalization that ensures Eq. \eqref{eq:orthogonality}. These will be used later and especially in Section \ref{sec:Eigenvectors}.
\begin{rem}
Normality of the standard eigenvectors and Eq. \eqref{eq:orthogonality}
make it necessary to include $c^{\ell}$'s in the analysis. This important
point which was missed in the earlier work \cite{Dai2009} is directly
responsible for ill-conditioned behavior leading to high sensitivity
of eigenvalues to perturbations. This underpins the non-perturbative behavior of Runaways type II's. 
\end{rem}
Comment: For normal (e.g., Hermitian) matrices the matrix $\Psi$
can be taken to be unitary and $\langle\tilde{\psi}|=|\psi\rangle^{\dagger}$,
which among other things implies that $||\langle\tilde{\psi}|\mbox{ }||_{2}=1$.
However, for non-normal matrices (e.g., the Toeplitz matrix $T$),
$||\langle\tilde{\psi}|\mbox{ }||_{2}$ can be arbitrary large; i.e.,
right and left eigenvectors can become almost orthogonal \cite{TrefethenEmbree2005}. 

The starting point for analytically understanding the bare ($\sigma=0$) Toeplitz
matrix is to derive its eigenvectors from which $p^{\ell}$ can be
inferred (i.e., Eq. \eqref{eq:p_l}); the real part is $2\pi\ell/n$
as described above. We now analytically solve for the eigenvectors of $T$.

\subsubsection{\label{sub:Eigenvectors-from-Wiener-Hopf}Eigenvectors from Wiener-Hopf
method}
The Wiener-Hopf method is tailored for solving equations of type $\sum_{i=0}^{N}c_{N-i}x_{i}=b_{N}$,  without explicitly calculating the inverse
of the Toeplitz matrix, where  $N\ge0$,  $\mathbf{c}$ and $\mathbf{b}$ are known vectors and
$\mathbf{x}$ is the unknown vector. The method is mostly used for integral equations;
however, the discrete version has been used in statistical physics,
especially in calculation of magnetization in the two dimensional
Ising model \cite[Chapter IX]{McCoyWu1973}.

In this section we mainly summarize, improve and extend the previous
work \cite{Dai2009}. We are interested in the finite section method (see \cite{BoettcherSilbermann1998}) and will use the Wiener-Hopf method. The requirements for Wiener-Hopf factorizing break down for the Fisher-Hartwig singular symbol when $\alpha\ne0$. It was nevertheless argued by McCoy and Wu (also see \cite{Dai2009}) that the technique can be employed with the appropriate definition of the winding number to obtain the eigenvectors when the symbol is only continuous, yet non-analytic, with appropriate analyticity properties away from the unit circle (see below). 

The translationally invariance of the Toeplitz matrix implies that the eigenvalue problem is of Wiener-Hopf
type 
\begin{eqnarray}
\sum_{j=0}^{n}W_{i-j}\psi_{j}^{\ell,+} & = & 0\qquad i\ge0\label{eq:DiscreteConv}\\
\sum_{j=0}^{n}\tilde{\psi}_{j}^{\ell,+}W_{j-i} & = & 0\qquad i\ge0\nonumber 
\end{eqnarray}
where $W=T-E^{\ell}\mathbb{I}$ with $\mathbb{I}$ being the identity
matrix of size $n$, and we inserted a ``+'' sign to emphasize the
vanishing of $\psi_{j}^{\ell}$ for $j<0$. In these equations, had
the indices of the Toeplitz matrix been doubly infinite $\left(-\infty,\infty\right)$,
the equations would easily be solved using Fourier expansions; the
semi-infinite sum makes the problem harder. The actual sum runs up
to $n$, but for large $n$ the properties are well approximated by
the semi-infinite case, though the convergence may be non-uniform.
The spectrum of the Toeplitz matrix is inside the convex hull of the
image of the symbol (see Fig. \ref{fig:Eigs_finiten_imageofSym})  \cite{BoettcherSilbermann1998}.
Therefore in this subsection we take $n\gg1$ and the eigenvalues
$\left\{ E^{1},E^{2},\cdots,E^{\ell},\cdots,E^{n}\right\} $ will be
close to, yet inside, $a\left(e^{ip}\right)$. 

In what follows, we focus on the right eigenfunctions and use Wiener-Hopf
method following the exposition of McCoy and Wu \cite{McCoyWu1973}
and \cite{Dai2009}. We then obtain the left eigenvectors using Eqs. \eqref{eq:orthogonality} and \eqref{eq:LEigVec_Compon}.
For the convergence of expansions below we assume
\begin{eqnarray*}
\sum_{j=0}^{\infty}|\psi_{j}^{\ell,+}| & < & \infty,\\
\sum_{j=-\infty}^{\infty}|W_{j}| & < & \infty.
\end{eqnarray*}

To make use of Fourier expansion, we need to extend the summation
index in Eq. (\ref{eq:DiscreteConv}) from below to $-\infty$ and
demand uniform convergence as before. Let $\Theta$ denote the Heaviside
function and let
\begin{eqnarray*}
\psi_{j}^{\ell,+} & = & 0\qquad i<0\quad,\\
\psi_{i}^{\ell,-} & = & \Theta\left(-i\right)\sum_{j=0}^{\infty}W_{i-j}\psi_{j}^{\ell}
\end{eqnarray*}
be the contribution of the $i<0$ terms and set $\Theta\left(0\right)=1$;
in our case this contribution vanishes at $i=0$ as well. We can formally
rewrite Eq. (\ref{eq:DiscreteConv}) as
\[
\sum_{j=-\infty}^{\infty}W_{i-j}\psi_{j}^{\ell,+}=\psi_{i}^{\ell,-}\quad i\in\mathbb{Z}\mbox{ }.
\]
The left hand side is a discrete convolution, therefore a Fourier
series representation gives
\begin{eqnarray}
\mathcal{W}\left(z\right)\tilde{\Psi}^{\ell,+}\left(z\right) & = & \tilde{\Psi}^{\ell,-}\left(z\right)\label{eq:FTWienerHopf}
\end{eqnarray}
where $z=e^{i\theta}$ and $\mathcal{W}\left(z\right)=a\left(z\right)-E^{\ell}$
is the Fourier representation of $W_{i-j}$. Note that the problem
has become an algebraic equation. 

At the first sight it seems like we complicated the problem by introducing
a second unknown $\tilde{\Psi}^{\ell,-}\left(z\right)$; however Wiener-Hopf
factorization resolves this.

Note that $\tilde{\Psi}^{\ell,+}\left(z\right)$ yields a Taylor series expansion
(i.e., of the form $\sum_{n\ge0}a_{i}z^{i}$ with $\sum_{i\ge0}\left|a_{i}\right|<\infty$).
Such a Taylor series (i.e., sum with $i\ge0$) defines what are called
``$+$'' functions that are analytic for $\left|z\right|<1$ and
continuous for $\left|z\right|\le1$. 

Similarly $\tilde{\Psi}^{\ell,-}\left(z\right)$ for $\left|z\right|=1$ can
be expanded in Laurent series of the form $\sum_{i<0}a_{i}z^{i}$
with absolutely convergent coefficients. Such an expansion defines
a ``$-$'' function that is analytic for $\left|z\right|>1$ and
continuous for $\left|z\right|\ge1$ and approaches zero as $z\rightarrow\infty$
\cite{McCoyWu1973}. 

The continuity of $\ln\mathcal{W}\left(z\right)$ is equivalent to
having a non-vanishing {\it winding number} defined by
\[
\nu=\frac{1}{2\pi i}\left\{ \ln\mathcal{W}\left(e^{2\pi i}\right)-\ln\mathcal{W}\left(e^{0i}\right)\right\} .
\]
Suppose $\mathcal{W}\left(z\right)$ has a winding number $\nu$.
Then the winding number of $\ln\left(z^{-\nu}\mathcal{W}\left(z\right)\right)$
is zero. Recall that for the finite Toeplitz matrix the eigenvalues
are in the convex hull of the image of the symbol; therefore, one can meaningfully
assign winding numbers about any point in the convex hull. For $0<\beta<1$ the winding number is $\nu=+1$
and for $-1<\beta<0$, it is $\nu=-1$ \cite{Dai2009,LPK2009}. 

If $\nu\ne0$, the analyticity of $\tilde{\Psi}^{\ell,+}\left(z\right)$ and
$\tilde{\Psi}^{\ell,-}\left(z\right)$ can be used to factorize $z^{-\nu}\mathcal{W}\left(z\right)$
for $\left|z\right|=1$ as 
\begin{equation}
z^{-\nu}\mathcal{W}\left(z\right)=e^{-G_{-}\left(z\right)}e^{-G_{+}\left(z\right)},\label{eq:CFactorized}
\end{equation}
where 
\begin{eqnarray*}
G_{+}\left(z\right) & = & -\left[\ln\left(z^{-\nu}\mathcal{W}\left(z\right)\right)\right]_{+}\\
G_{-}\left(z\right) & = & -\left[\ln\left(z^{-\nu}\mathcal{W}\left(z\right)\right)\right]_{-}.
\end{eqnarray*}

When is the factorization in Eq. (\ref{eq:CFactorized}) possible? The factorization is guaranteed whenever the corresponding Toeplitz operator is invertible \cite{basor1991fisher}.
More generally, this is guaranteed by the following two theorems. 
\begin{thm*}
(Pollard's) If a function $f\left(z\right)$ is analytic inside and
continuous on a simple closed contour $C$, then 
\[
\oint_{C}f\left(z\right)\mbox{ }dz=0\mbox{ }.
\]
\end{thm*}
This extends Cauchy's theorem as $f\left(z\right)$ is required to
be only continuous and not necessarily analytic on $C$.

Let $\ell^{*}$ denote the space of all sequences that are Fourier
series of all absolutely summable sequences.
\begin{thm*}
(Wiener-Levy) If $\mbox{ }{\cal W}\left(e^{i\theta}\right)\in\ell^{*}\mbox{ }$
and if $\mbox{ }\ln{\cal W}\left(e^{i\theta}\right)$ is continuous
when $0\le\theta\le2\pi$, then $\ln{\cal W}\left(e^{i\theta}\right)\in\ell^{*}$. 
\end{thm*}
Hence, if $z^{-\nu}\mathcal{W}\left(z\right)$ is nonzero on the unit
circle $\left|z\right|=1$, then $\ln\left(z^{-\nu}\mathcal{W}\left(z\right)\right)$
can always be found such that it is continuous for $0\le\theta\le2\pi$
and if further $\ln\mathcal{W}\left(z\right)$ is continuous on the
unit circle, it would have a Laurent series expansion whose coefficients
are absolutely summable. 

So far the discussion has been general and applicable to general Toeplitz matrices. We now turn to our Toeplitz matrix with a Fisher-Hartwig singular symbol.

The factorization is possible when ${\cal W}\left(z\right)$ does
not have singularities or zeros on the unit circle and none at zero
and infinity. All of these break down with a Fisher-Hartwig symbol.
But using Pollard's theorem, McCoy and Wu argued that the factorization
works as long as $z^{-\nu}{\cal W}\left(z\right)$ is continuous and
not necessarily analytic on the unit circle with appropriate analytic
continuation away from the unit circle. In \cite{Dai2009}, it was
shown that the recipe covers $0<\alpha<\left|\beta\right|<1$, where
\begin{align*}
\nu & =-1\qquad-1<\beta<0\\
\nu & =+1\qquad+1>\beta>0
\end{align*}
the latter contains only trivial solutions. 

From Eq. \eqref{eq:FTWienerHopf} and for $\left|z\right|=1$, we find
\[
e^{-G_{+}\left(z\right)}\Psi^{\ell,+}\left(z\right)=z^{-\nu}e^{G_{-}\left(z\right)}\Psi^{\ell,-}\left(z\right);
\]
the left hand side is a + function and is analytic for $\left|z\right|<1$
and continuous for $\left|z\right|\le1$. The right hand side is not
necessarily a $-$ function because of $z^{-\nu}$, though still analytic
for $\left|z\right|>1$ and continuous for $\left|z\right|\ge1$.
So there is an entire function $F(z)$ such that 
\begin{eqnarray*}
e^{-G_{+}\left(z\right)}\Psi^{\ell,+}\left(z\right) & = & F\left(z\right)\qquad\left|z\right|\le1\\
z^{-\nu}e^{G_{-}\left(z\right)}\Psi^{\ell,-}\left(z\right) & = & F\left(z\right)\qquad\left|z\right|\ge1\mbox{ }.
\end{eqnarray*}
Since, $\nu=-1$, and the left hand side of the above equation is
$z^{-\nu}$ times a $-$ function $F\left(z\right)=\sum_{i=-\infty}^{\left|\nu\right|-1}\kappa_{i}z^{i}=\kappa_{0}$
is a constant. Substituting this in the foregoing equations, we find
\[
\Psi^{\ell,+}\left(z\right)=\kappa_{0}e^{G_{+}\left(z\right)},
\]
where using Eq. \eqref{eq:CFactorized}, $G_{+}\left(z\right)$ is 
\[
G_{+}\left(z\right)=-\frac{1}{2\pi i}\oint_{\mathtt{s}}dz'\mbox{ }\frac{\left[ \ln\left(z'\mathcal{W}\left(z'\right)\right)\right] _{+}}{z'-z}.
\]
Inside the natural log there is a factor of $z'$ because the winding number is $\nu=-1$. That is, although $\mathcal{W}\left(z'\right)$ does not possess a factorization, $z'\mathcal{W}\left(z'\right)$ does.
This was evaluated in \cite{Dai2009} and concluded that the eigenvectors have components
given by 
\[
\psi_{j}^{\ell}\sim Az_{crit}^{-j-1}+B\left(j+1\right)^{-\left(2\alpha+1\right)}
\]
where $a\left(z_{crit}\right)\equiv E^{\ell}$ and $A$ and $B$ are
constants depending on $\alpha$, $\beta$ and $E^{\ell}$ (see \cite[Eq. 47]{Dai2009}).
It was then argued that
\[
z_{crit}^{-n}\approx n^{-\left(2\alpha+1\right)}.
\]
Since $z_{crit}=a\left(e^{-ip^{\ell}}\right)$, $\Im\left(p^{\ell}\right)=\Re\left(\ln z_{crit}\right)=\left(2\alpha+1\right)\frac{\ln n}{n}+O\left(1/n\right)$.

We are now in the position to calculate the eigenvectors and we obtain
\begin{eqnarray}
\psi_{j}^{\ell} & \propto & \exp\left[\left(\frac{2\pi i\ell}{n}-(2\alpha+1)\frac{\ln n}{n}\right)j\right]\label{eq:Right_EVEC}
\end{eqnarray}
The normalized standard eigenvectors are 
\begin{equation}
\psi_{j}^{\ell}\approx\sqrt{\frac{2\left(1+2\alpha\right)\ln n}{n}}\exp\left[\left(\frac{2\pi i\ell}{n}-(2\alpha+1)\frac{\ln n}{n}\right)j\right].\label{eq:Eigenvector_Final}
\end{equation}

This along with Eqs. \eqref{eq:orthogonality} and \eqref{eq:LEigVec_Compon}  enable us to calculate the left eigenvectors
\begin{eqnarray*}
\tilde{\psi}_{j}^{\ell} & = & c^{\ell}\psi_{n-j-1}^{\ell}\\
c^{\ell} & \equiv & \frac{n}{2\left(1+2\alpha\right)\ln n}\mbox{ }\exp\left[\frac{2\pi i\ell}{n}+(2\alpha+1)\frac{\ln n}{n}\left(n-1\right)\right].
\end{eqnarray*}

Therefore, the left eigenvectors are 
\begin{equation}
\tilde{\psi}_{j}^{\ell}\approx\sqrt{\frac{n}{2\left(1+2\alpha\right)\ln n}}\exp\left[-\left(\frac{2\pi i\ell}{n}-(2\alpha+1)\frac{\ln n}{n}\right)j\right].\label{eq:LeftEigvector_Final}
\end{equation}

Comment: One can easily verify that $\sum_{j=0}^{n-1}\tilde{\psi}_{j}^{m}\psi_{j}^{\ell}=\delta_{m,\ell}$. 

As discussed at the beginning, for many applications the asymptotic
($n\rightarrow\infty$) behavior of the determinant of the Toeplitz
matrix are desired. For symbols given by Eq. \eqref{eq:LeoSymbol},
the determinant was explicitly calculated for $n\ge1$ if $\Re\alpha>-1/2$
with neither $\alpha+\beta$ nor $\alpha-\beta$ being negative integers
\cite{Ehrhart1997}. We now pause to point out that the trace and
the determinant of  our Toeplitz matrix (Eq. \eqref{eq:LeoSymbol})
evaluate to be 
\begin{eqnarray}
\mbox{Tr}\left(T\right) & = & n\mbox{ }t_{0}=\frac{n\mbox{ }\Gamma\left(2\alpha+1\right)}{\Gamma\left(\alpha+\beta+1\right)\Gamma\left(\alpha-\beta+1\right)}\label{eq:trace}\\
\mbox{Det}\left(T\right) & = & \frac{{\cal G}\left(1+\alpha+\beta\right){\cal G}\left(1+\alpha-\beta\right)}{{\cal G}\left(1+2\alpha\right)}\frac{{\cal G}\left(1+n\right){\cal G}\left(1+n+2\alpha\right)}{{\cal G}\left(1+n+\alpha+\beta\right){\cal G}\left(1+n+\alpha-\beta\right)}\label{eq:determinant}\\
 & \approx & \frac{{\cal G}\left(1+\alpha+\beta\right){\cal G}\left(1+\alpha-\beta\right)}{{\cal G}\left(1+2\alpha\right)}n^{\alpha^{2}-\beta^{2}};\qquad n\gg1\nonumber 
\end{eqnarray}
where ${\cal G}$ is the Barnes ${\cal G}$-function \cite{Barnes1900}.  It is an
entire function defined by
\[
{\cal G}\left(1+z\right)=\left(2\pi\right)^{z/2}e^{-\left(z+1\right)z/2-\gamma z^{2}/2}\Pi_{n=1}^{\infty}\left\{ \left(1+\frac{z}{n}\right)^{n}e^{-z+z^{2}/2n}\right\} ,
\]
where $\gamma$ is Euler's constant. Both the determinant and the
trace are real numbers as expected since the entries of $T$ are real, 
and non-real eigenvalues appear in complex conjugate pairs. 

\section{\label{sec:Eigenvectors}Presence of Disorder $\sigma>0$}
Consider the spectrum of $T+\sigma V$, where $V=\mbox{diag}\left(\epsilon_{1},\epsilon_{2},\dots,\epsilon_{n}\right)$
is considered a perturbation to $T$. For example if $\epsilon_{i}$
are drawn independently and randomly from a standard normal distribution
then the spectrum can change in dramatic ways as shown in Figs. \ref{fig:LargeN}
and \ref{fig:TDisorder}.

The eigenvalues of a perturbed matrix are continuous; i.e., their
motion follows a connected path in the complex plane as $\sigma$ increases.  This follows from the fact that eigenvalues are roots
of a characteristic polynomial, which itself is continuous and a theorem
due to Rouch\'e \cite[Chapter 4 ]{StewartSun}. 
\begin{figure}
\begin{centering}
\includegraphics[scale=0.35]{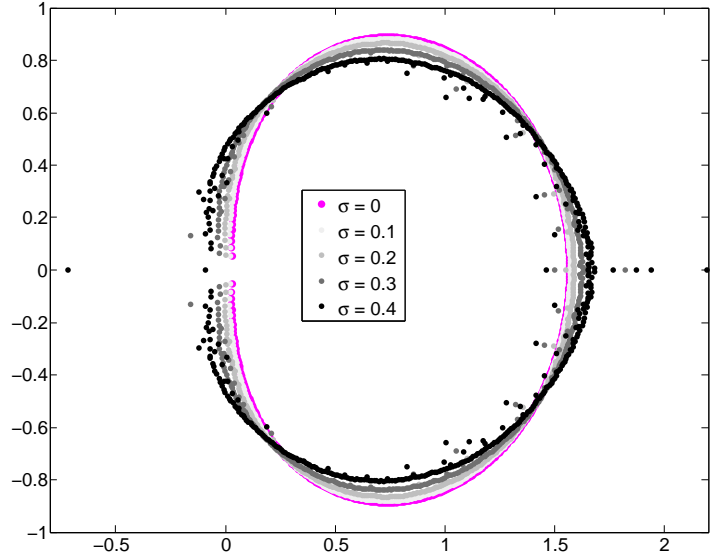}
\par\end{centering}
\caption{\label{fig:LargeN}Spectrum of $T+\sigma V$ for $n=512$. $V$ is
a real diagonal matrix with standard normal entries. We take $n$ large and $\sigma$ small such that the bulk properties
of the original spectrum are retained, yet $\sigma$ is large enough
to produce notable runaways.}
\end{figure}
\begin{rem}
In Fig. \ref{fig:LargeN} we show the spectrum and its deformations
with increasing $\sigma$ for $n=512$. However, to closely couple the theory
to the numerical work, for the rest of the figures we pick a working example with $n=160$
and a realization of randomness (for example Fig. \ref{fig:TDisorder}).
The conclusions that follow, we believe, do not depend on a given
realization of the disorder; rather the statements are generic. However,
for the sake of coherence and concreteness of the presentation, we
found it helpful to work with a single seed of randomness and demonstrate
the various aspects of the theory in its context.
\end{rem}

\begin{rem}
The parameter $\sigma$ sets the strength of perturbation and is a smooth parameter.
We think of $E^{\ell}(\sigma)$ as the evolution of the $\ell^{\mbox{th}}$ eigenvalue
with respect to  $\sigma$, which we think of as ``time.'' This way we can track the eigenvalues
in the complex plane as will be evident in the following sections. 
\end{rem}

\begin{figure}
\centering{}\includegraphics[scale=0.33]{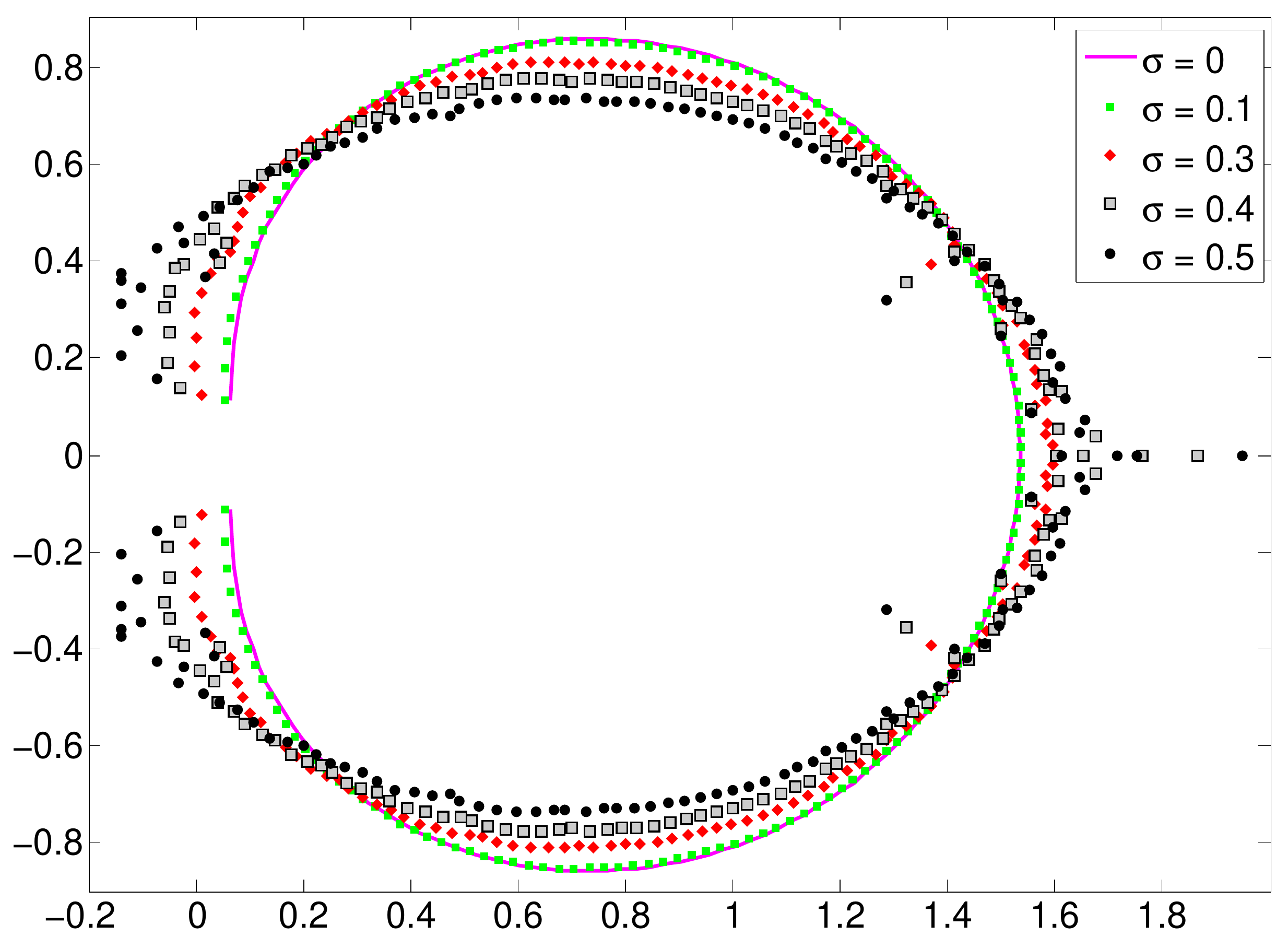}\includegraphics[scale=0.29]{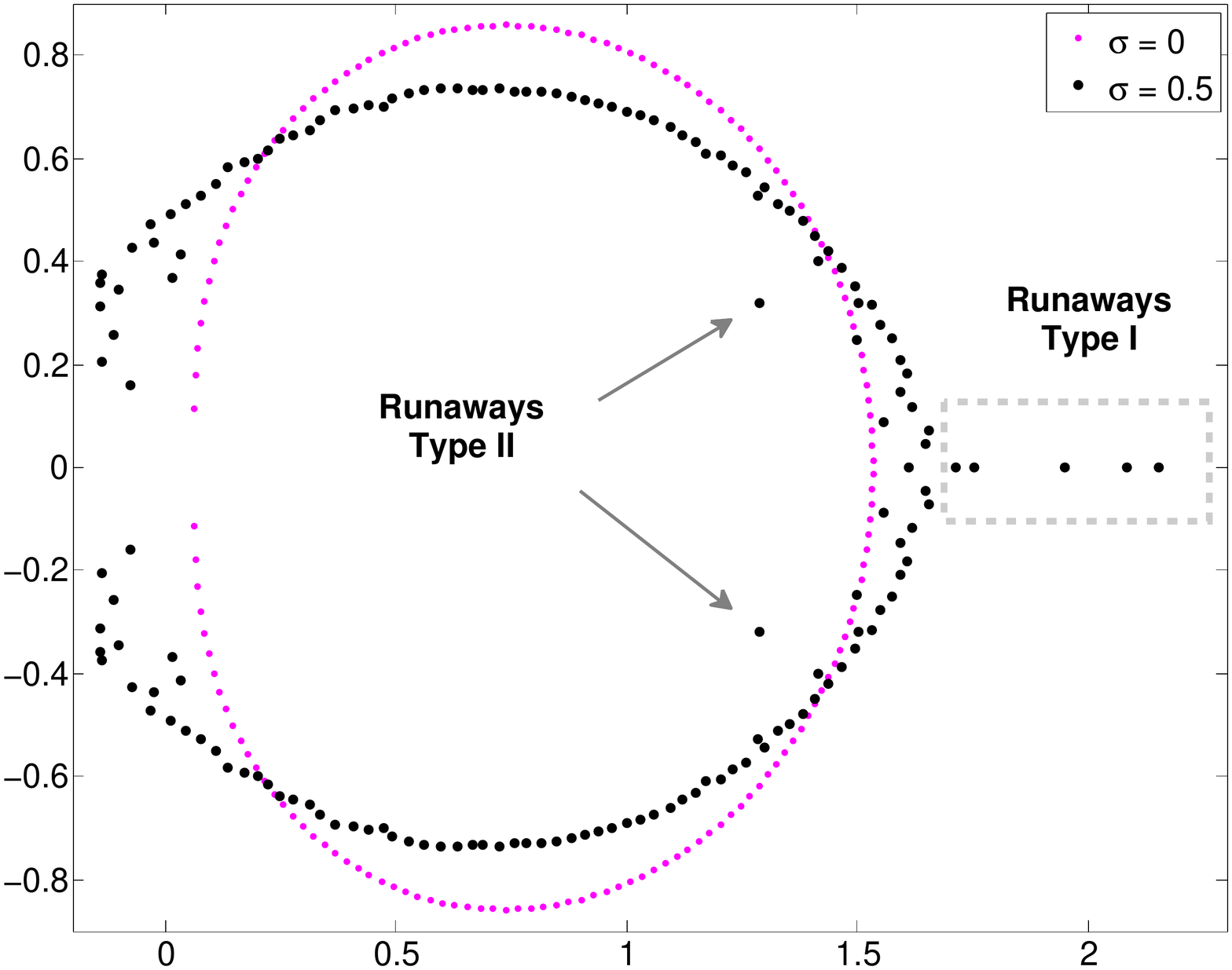}\caption{\label{fig:TDisorder}Left: We show the eigenvalues of $T+\sigma V$
for various $\sigma$, where the parameters of the Toeplitz matrix
$T$ are taken to be $\alpha=1/3$, $\beta=-1/2$ and with the size
$n=160$ . Right: For better contrast we show $\sigma=0$ and $\sigma=0.5$.
In both plots note that in addition to the runaways, there is a bulk
compression of imaginary part of the eigenvalues with increasing $\sigma$.
We only show the most prominent pair of Runaways type II's; however
generically (even in this case) there are more, especially for larger
$n$ as can be seen in Fig \ref{fig:LargeN}. }
\end{figure}

Before investigating the effect of disorder, we comment on three notable
spectral features seen in Fig. \ref{fig:TDisorder}:

First is the \textit{Bulk} motion of the eigenvalues. The net motion of any eigenvalue results from its
interaction with all other eigenvalues. One sees in the Figs.  \ref{fig:LargeN} and  \ref{fig:TDisorder}  that majority
of eigenvalues experience a compression of their imaginary parts as
if complex conjugates pull each other in and that there is a stretching apart
of the real parts. The bulk motion is well captured
by second order perturbation theory, which is discussed in subsection
\ref{sub:Perturbation-theory-of}. 

Second is the\textit{ Runaways type I }. These result from a strong
attraction of complex conjugate eigenvalues close to the real line
(mostly on the right sector of the spectrum near the real axis). These
eigenvalues approach one-another until the attraction becomes strong
enough that perturbation theory breaks down and complex conjugate eigenvalues collide on the real line and generically become real and distinct thereafter. We denote these as ``type I'' runaway eigenvalues,
where the ``type'' refers to particular type of dynamics leading
to breakdown of perturbation theory. We discuss these Runaways in
subsection \ref{sub:Non-Perturbative-TypeI}.

Third notable feature is the \textit{Runaways type II }eigenvalues\textit{.
}The second type of non-perturbative behavior of some of the eigenvalues
is exhibited by the ones that are relatively far from the real line
and leave the bulk by bulging into the inner part of the spectrum.
This behavior as far as we know has not been observed previously in
models of non-Hermitian quantum mechanics and the literature of Toeplitz-like
matrices. In subsection \ref{sub:Runaways-Type-II} we show that this behavior
is due to large angles between the left and right eigenvectors, which results in large norms of the left eigenvectors rendering ill-conditioning.

Recall that $T(\sigma)=T+\sigma V$. Now the eigenvalues and eigenvectors are functions of $\sigma$ as well. We denote by  $|\psi_{\sigma}^{\ell}\rangle$ and
$\langle\psi_{\sigma}^{\ell}|$ the left and right eigenvectors to
eigenvalue $E^{\ell}\left(\sigma\right)$ respectively, if 
\begin{eqnarray}
T(\sigma)\mbox{ }|\psi_{\sigma}^{\ell}\rangle & = & E^{\ell}\left(\sigma\right)\mbox{ }|\psi_{\sigma}^{\ell}\rangle,\label{eq:EigenValue_R}\\
\langle\tilde{\psi}_{\sigma}^{\ell}|\mbox{ }T(\sigma) & = & E^{\ell}\left(\sigma\right)\mbox{ }\langle\tilde{\psi}_{\sigma}^{\ell}|.\label{eq:EigenValue_L}
\end{eqnarray}

One recovers Eqs. \eqref{eq:EigValPro_R} and \eqref{eq:EigValProbL} by setting $\sigma=0$ in Eqs. \eqref{eq:EigenValue_R} and \eqref{eq:EigenValue_L}.

\subsection{\label{sub:Perturbation-theory-of}Perturbative regime: ``Bulk''
eigenvalues }
The Toeplitz matrices under consideration have real entries. The non-real
eigenvalues of a real matrix occur in complex conjugate pairs (e.g.,
Fig. \ref{fig:ImageSymbol}).
\begin{prop}
Let $T(\sigma)=T+\sigma V$, where $V$ is a diagonal real matrix
whose diagonal entries are random and drawn independently and identically
from a distribution with mean zero. Then the expected first order
corrections from perturbation theory vanish and the expected second
order correction is given by $\mathbb{E}(E_{2}^{\ell})=n\mbox{ }\mathbb{E}\left(v^2\right)\sum_{j\ne\ell}\left(E_{0}^{\ell}-E_{0}^{j}\right)^{-1}$. \end{prop}
\begin{proof}
The standard perturbation theory for the eigenvalues of Hermitian matrices does not suffice because the Toeplitz
matrix is not symmetric. However, one can use the right and left eigenvectors
to generalize the standard perturbation theory results to arbitrary
orders; here we stop at the second order. The $\ell^{\mbox{th}}$
eigenpairs in the presence of disorder have the following perturbation expansions
\begin{eqnarray}
E^{\ell} & = & E_{0}^{\ell}+\sigma E_{1}^{\ell}+\sigma^{2}E_{2}^{\ell}+\cdots\label{eq:E_perturbExpand}\\
|\psi^{\ell}\rangle & = & |\psi_{0}^{\ell}\rangle+\sigma\mbox{ }|\psi_{1}^{\ell}\rangle+\sigma^{2}\mbox{ }|\psi_{2}^{\ell}\rangle+\cdots\label{eq:Evec_perturbExpand},
\end{eqnarray}
where quantities with the zero subscript denote the eigenvalues of
the unperturbed problem (i.e., no disorder), that were analytically derived in Section \ref{sec:Eigenvalues}. Then, standard perturbation theory of non-Hermitian matrices (see for example Section 52 in \cite{TrefethenEmbree2005,movassagh2016eigenvalue}) to second order gives
\begin{eqnarray}
E_{1}^{\ell} & = & \langle\tilde{\psi}_{0}^{\ell}|\mbox{ }V\mbox{ }|\psi_{0}^{\ell}\rangle,\nonumber \\
E_{2}^{\ell} & = & \sum_{j\ne\ell}\frac{\langle\tilde{\psi_{0}}^{\ell}|\mbox{ }V\mbox{ }|\psi_{0}^{j}\rangle\langle\tilde{\psi_{0}}^{j}|\mbox{ }V\mbox{ }|\psi_{0}^{\ell}\rangle}{E_{0}^{\ell}-E_{0}^{j}}.\label{eq:E2}
\end{eqnarray}

Since $V=\mbox{diag}(v_{1},v_{2},\dots,v_{n})$ is diagonal and real,
we find that 
\begin{eqnarray}
\mathbb{E}\left(E_{1}^{\ell}\right) & = & \mathbb{E}\langle\tilde{\psi}_{0}^{\ell}|\mbox{ }V\mbox{ }|\psi_{0}^{\ell}\rangle=\frac{1}{c^{\ell}}\sum_{i=0}^{n-1}\mathbb{E}\left(v_{i}\right)\mbox{ }\psi_{0,i}^{\ell}\psi_{0,n-i-1}^{\ell}=0\label{eq:Expectation_1stOrder}
\end{eqnarray}
\begin{eqnarray}
\mathbb{E}\left(E_{2}^{\ell}\right) & = & \mathbb{E}\sum_{j\ne\ell}\frac{\langle\tilde{\psi_{0}}^{\ell}|\mbox{ }V\mbox{ }|\psi_{0}^{j}\rangle\langle\tilde{\psi_{0}}^{j}|\mbox{ }V\mbox{ }|\psi_{0}^{\ell}\rangle}{E_{0}^{\ell}-E_{0}^{j}}=\mathbb{E}\sum_{j\ne\ell}\frac{\sum_{i,k}v_{i}v_{k}\psi_{0,i}^{\ell}\psi_{0,n-i-1}^{\ell}\mbox{ }\psi_{0,k}^{j}\psi_{0,n-k-1}^{j}}{E_{0}^{\ell}-E_{0}^{j}}\nonumber \\
 & = & \sum_{j\ne\ell}\frac{1}{c^{\ell}c^{j}}\frac{\sum_{i}\left\{ \mathbb{E}\left(v_{i}^{2}\right)\mbox{ }\psi_{0,i}^{\ell}\psi_{0,n-i-1}^{\ell}\mbox{ }\psi_{0,i}^{j}\psi_{0,n-i-1}^{j}\right\} }{E_{0}^{\ell}-E_{0}^{j}},\label{eq:Expectation_2ndOrder}
\end{eqnarray}
where we used independence $\mathbb{E}(v_{i}v_{k})=\mathbb{E}(v_{i}^{2})\delta_{i,k}$. 

From Eqs. \eqref{eq:Eigenvector_Final} and \eqref{eq:LeftEigvector_Final}
we have
\begin{eqnarray*}
\mathbb{E}\left(E_{2}^{\ell}\right) & = & \mathbb{E}(v^2)\sum_{j\ne\ell}\frac{n}{E_{0}^{\ell}-E_{0}^{j}}.
\end{eqnarray*}
It is remarkable that the dependence on the eigenvectors drops out and the eigenvalues dictate the expected behavior to the second
order in the presence of diagonal disorder. 
\end{proof}
\begin{figure}
\begin{centering}
\includegraphics[scale=0.3]{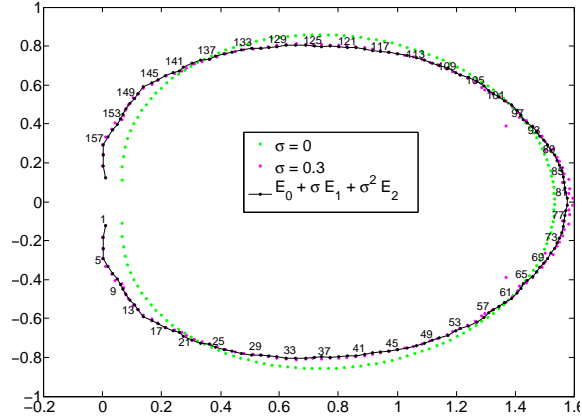}
\par\end{centering}
\caption{\label{fig:Second-order-success}Comparison of empirical eigenvalues
$E_{\sigma}^{\ell}$ (black dots) to $E_{0}^{\ell}+\sigma\mbox{ }E_{1}^{\ell}+\sigma^{2}\mbox{ }E_{2}^{\ell}$
(curve shown in red). Second order perturbation theory captures the
bulk motion of eigenvalues. Starting at this figure, we number the
eigenvalues based on their $p$ values as discussed above. See Fig.
\ref{fig:The-bulk-compression} for the statistical average of the
second order perturbation theory.}
\end{figure}

Now by the arguments following Proposition 2 in \cite{movassagh2016eigenvalue},
as long as the difference of the real parts is larger than that of
the imaginary parts, there is a compressive push towards the real
line from any complex conjugate pair of eigenvalues (i.e., $E^{j}$
and $\overline{E^{j}}$) on $E^{\ell}$ with a small net magnitude.
Moreover, in \cite{movassagh2016eigenvalue} it was shown that eigenvalues
near the real line feel a push from other eigenvalues in the real
direction. This explains the compression towards and the stretching
along the real axis of the spectrum.

Since perturbation theory only requires the knowledge of the unperturbed eigenpairs and the perturbation matrix, we could directly calculate $E^\ell_1$ and $E^\ell_2$ for a given random diagonal perturbation matrix $V$. This way  we have directly computed $E^\ell_0+\sigma E^\ell_1 +\sigma^2 E^\ell_2$ and can compare it with the eigenvalues of $T(\sigma)$, denoted by $E^\ell(\sigma)$, obtained by numerical exact diagonalization. The calculation shows that the positions of the majority of the eigenvalues in the complex plane are 
very well approximated by the second order perturbation theory. 

We demonstrate this in Fig.
\ref{fig:Second-order-success}, where the eigenvalues
of $T(\sigma)$ are obtained by exact diagonalization and compared
to 
\[
E^\ell(\sigma)\approx E_{0}^{\ell}+\sigma E_{1}^{\ell}+\sigma^{2}E_{2}^{\ell},
\]
for all $1\le\ell\le n$. It is evident that, except for the runaways,
second order perturbation theory successfully captures the bulk motion
of the spectrum. Therein, the compression along the imaginary axis
and stretching along the real axis of the spectrum is evident. 
\begin{figure}
\includegraphics[scale=0.3]{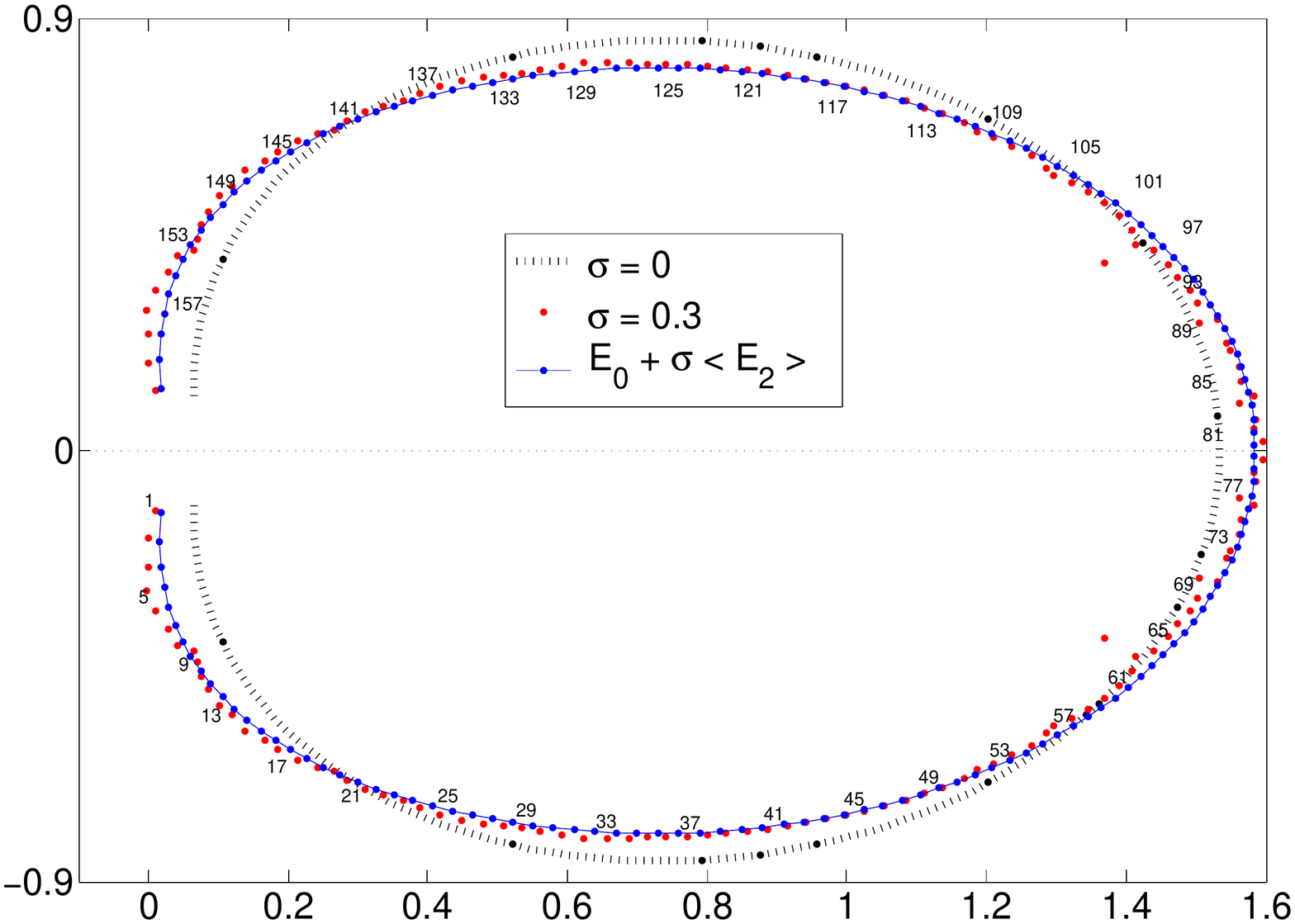}\includegraphics[scale=0.3]{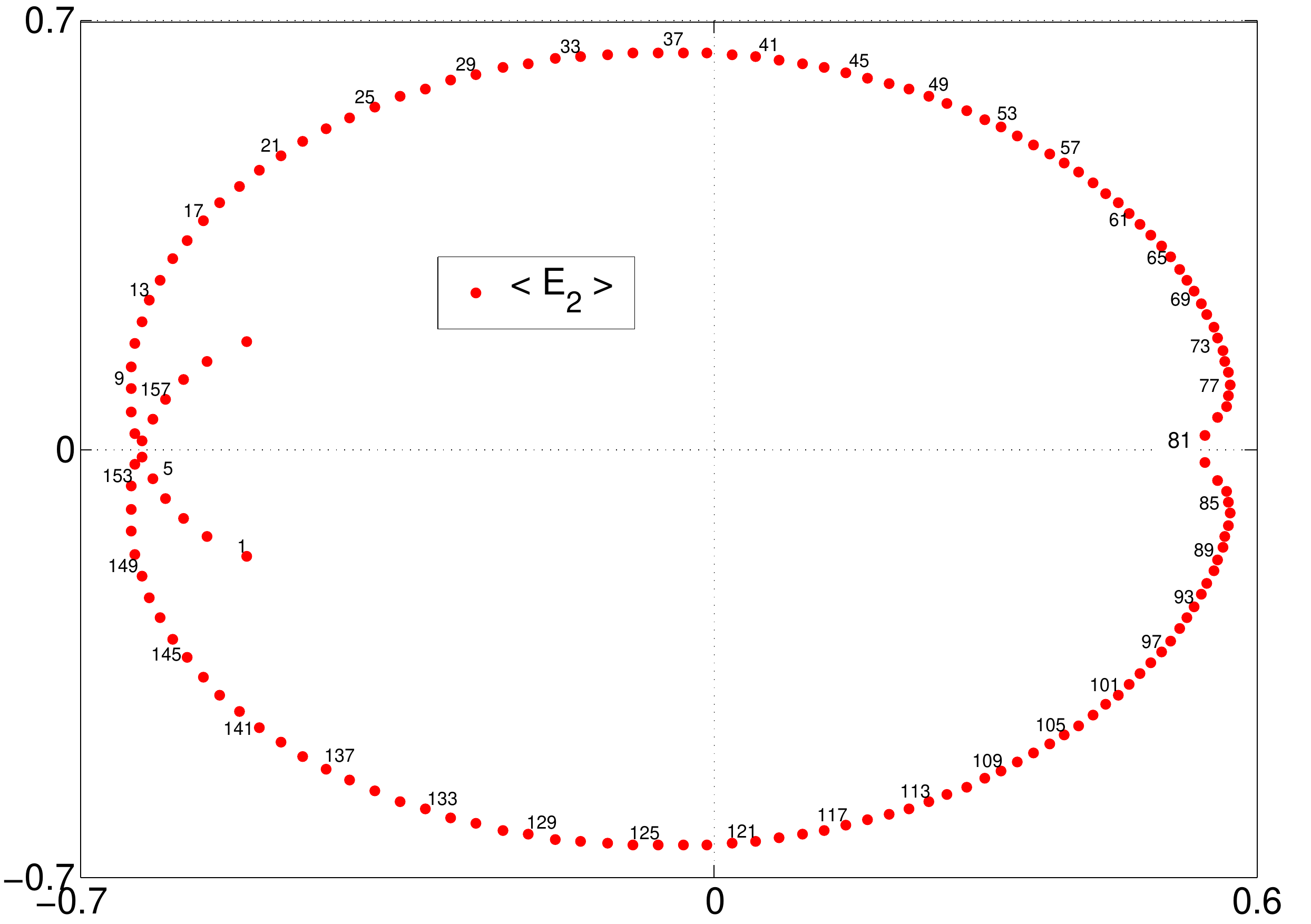}
\caption{\label{fig:The-bulk-compression}The statistical dynamics of eigenvalues
is captured by second order perturbation theory. Only here we denote
$\mathbb{E}\left(\bullet\right)\equiv\langle\bullet\rangle$ for clarity.
The left figure shows $E_{0}+\sigma\langle E_{2}\rangle$ with $\sigma=0.3$
. We have labeled every fourth point by their $p$ values. The figure
on the right shows that statistically there is a compression along
the imaginary axis and a stretch along the real axis. }
\end{figure}

We next consider the expected motion of the eigenvalues in response to random diagonal perturbations. Hence the expectation is take with respect to random $V$ and we can quantify the success of perturbation theory in an expectation sense. Since
$\mathbb{E}(E_{1}^{\ell})=0$ for all $\ell$, in Fig. \ref{fig:The-bulk-compression}
we compare the empirical eigenvalues with $E_{0}^{\ell}+\sigma^{2}\mathbb{E}\left(E_{2}^{\ell}\right)$.
We find that perturbation theory, even in an expectation sense, is
sufficient in accounting for the bulk dynamics of the eigenvalues. 

\begin{rem}
In Section \ref{sec:Eigenvalues}, we analytically solved the eigenvalues and eigenvectors of the unperturbed matrix. We now leverage on these results. As discussed above, this knowledge along with the perturbation matrix $V$ is sufficient to carry out the perturbation expansion to any order. Since $\mathbb{E}(E_{1}^{\ell})=0$, we analytically calculate  $\mathbb{E}(E_{2}^{\ell})$ and obtain $E^{\ell}_0+ \sigma^2 \mathbb{E}(E^{\ell}_2)$. This means that once $E_{2}^{\ell}$ is  calculated, we take expectation with respect to the random diagonal entries of $V$, which we  take to be standard real normals.  {\it It is quite remarkable that the analytical calculation of the expected values of first and second order corrections captures most of what happens to bulk eigenvalues in each instance. Moreover, it proves the qualitative deformation of the spectrum as we now discuss.}\end{rem}

Analytical calculation of $\mathbb{E}\left(E_{2}^{\ell}\right)$ given by Eq. \eqref{eq:Expectation_2ndOrder} shows
the imaginary compression and stretching along the real axis of the
spectrum by disorder. In Fig. \ref{fig:The-bulk-compression} we label
$E_{0}^{\ell}$ in one to one correspondence with calculated $\mathbb{E}\left(E_{2}^{\ell}\right)$.
Note that the upper eigenvalues are pushed down and the lower ones
pushed up; i.e., a compression along the imaginary axis. Moreover,
the eigenvalues with real parts to the right (left) of the center
of the spectrum get a positive (negative) real contribution, which
shows that the spectrum becomes stretched along the real axis. These
explain the bulk features observed in numerical evaluation of eigenvalues
of $T(\sigma)$. 
\begin{figure}
\includegraphics[scale=0.26]{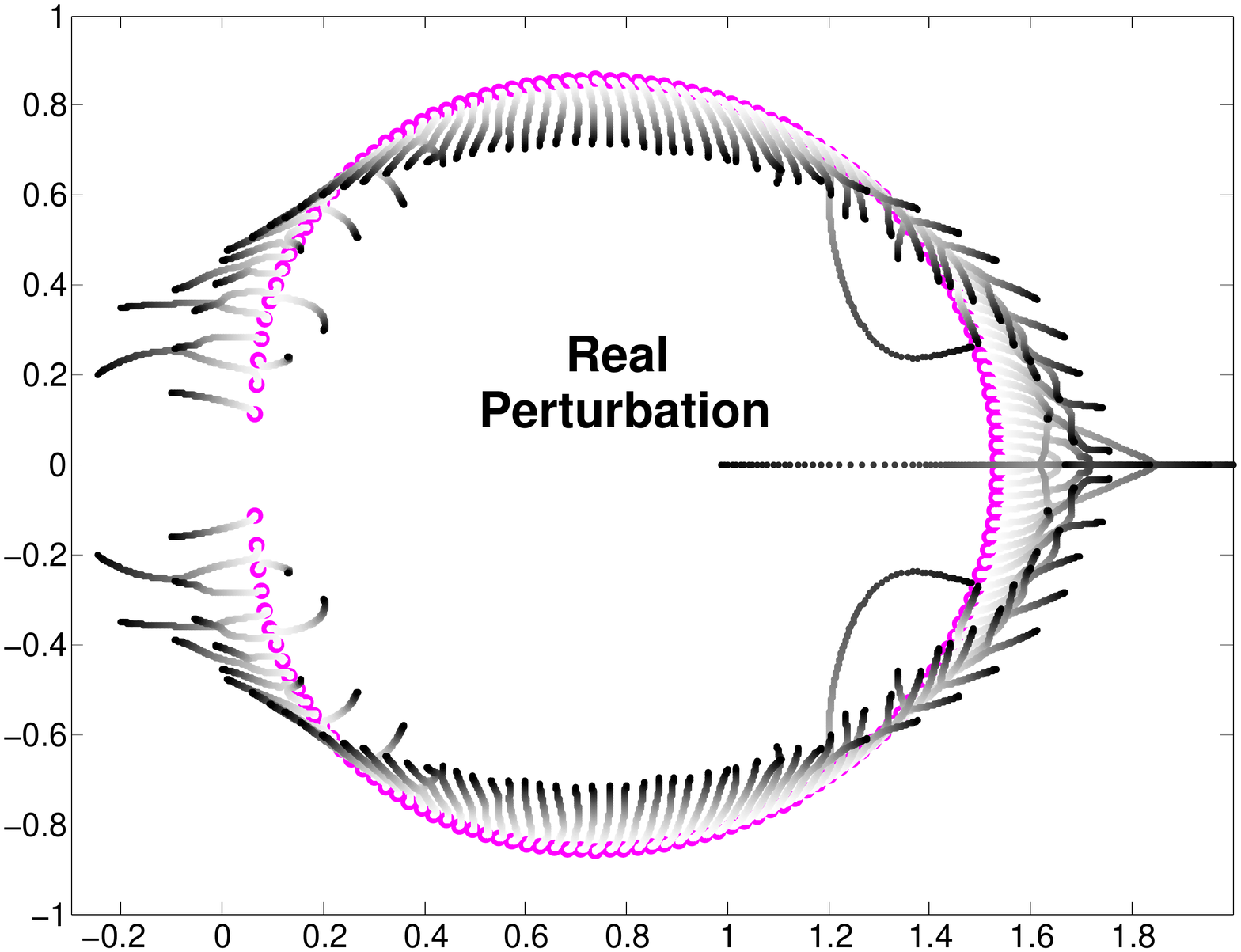}\includegraphics[scale=0.26]{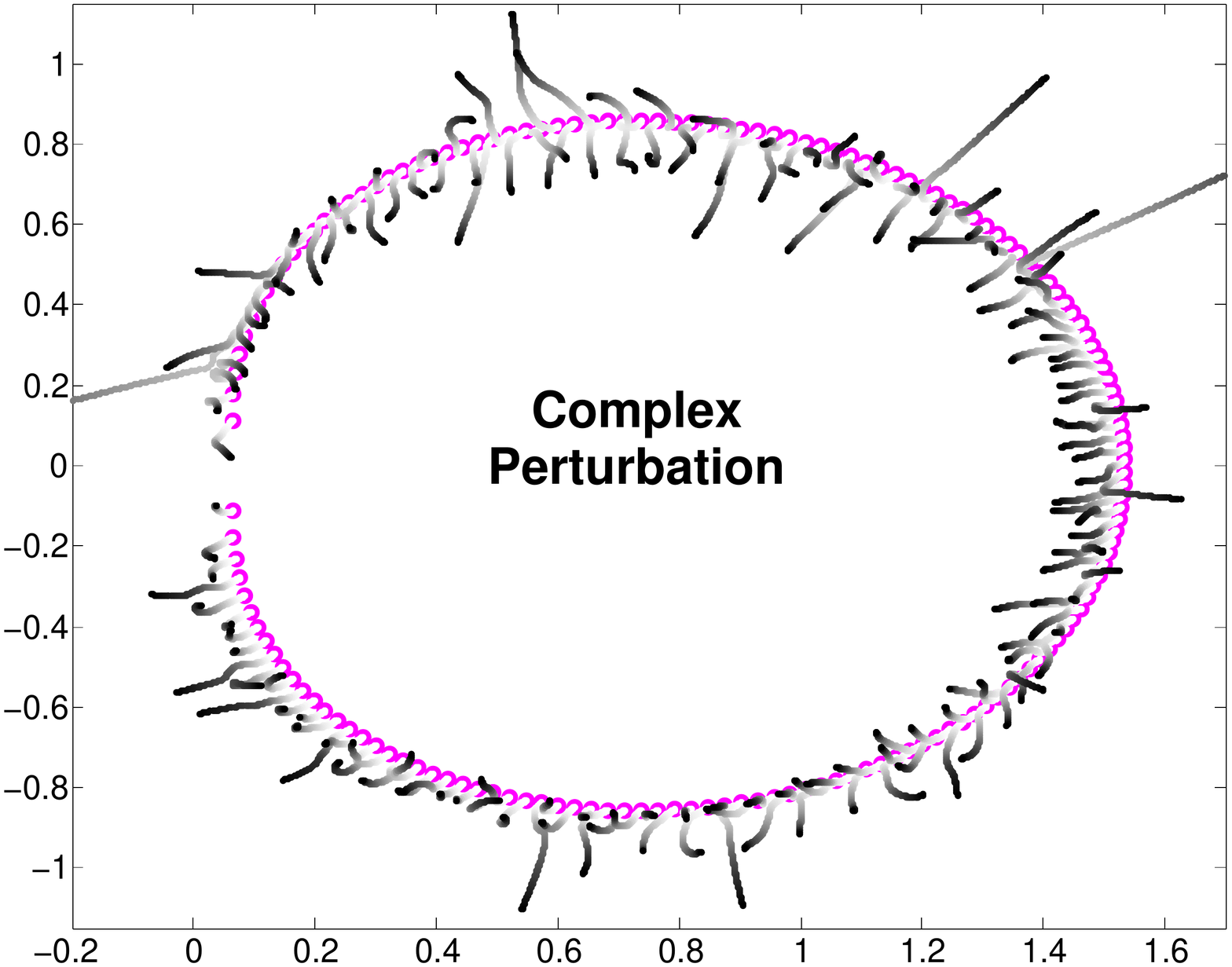}
\caption{\label{fig:Runaway-Type-II:Trajectory} On the left: the perturbation
is a real random diagonal and the right a complex random diagonal.
Note that in the latter the runaways move in random directions; moreover
the symmetry with respect to the imaginary axis is broken.}
\end{figure}

\subsection{\label{sub:Non-Perturbative-GENERAL}Non-Perturbative regime:  ``Runaway''
eigenvalues}
\subsubsection{\label{sub:Non-Perturbative-TypeI}Runaways type I}
\begin{rem}
\textbf{Explanation of gray scales in figures :} In figures such as
Fig. \ref{fig:Runaway-Type-II:Trajectory}, the magenta are the unperturbed
eigenvalues where $\sigma=0$ and eigenvalues evolve to their final
position $\sigma_{max}$ where they are shown in black. The grey-scale
shows the eigenvalues in the intermediate regime $0<\sigma<\sigma_{max}$.
In other words, they start as white dots at $\sigma=0$ which coincides
with the magenta and reach their final (black) position at $\sigma_{max}$.
This depiction captures the spectral dynamics with respect to $\sigma$
on a static plot. Therefore, sometimes we refer to $\sigma$ as ``time''. 
\end{rem}
A general real matrix, under random real perturbations exhibits an
attraction between any complex conjugate eigenvalues as proved elsewhere \cite{movassagh2016eigenvalue}
. In \cite{movassagh2016eigenvalue} we did not rely from the onset
on a Toeplitz structure nor perturbation theory and worked directly
with spectral dynamics theory. 

Above we proved that $\mathbb{E}(E_{1}^{\ell})=\mathbb{E}\langle\tilde\psi_0^{\ell} |V|\psi_0^\ell\rangle= 0$ for all $\ell$. 
\begin{prop}
\label{attraction} Any complex conjugate pair of the eigenvalues
of  $T+\sigma V$ attract. Let $V$ be a real diagonal matrix with independently and identically distributed entries with zero mean, any of which is denoted by $v$.  We have $\mathbb{E}\left(v\right)=0$, and the expected second order correction in Eq. \eqref{eq:E2} is
\[
\mathbb{E}(E_2^\ell)=-i\frac{n\mbox{ }\mathbb{E}\left(v^{2}\right)}{2\Im\left(E_{0}^{\ell}\right)}\quad.
\]
\end{prop}
\begin{proof}
We proved complex conjugate attraction under much more general set
of assumptions elsewhere \cite{movassagh2016eigenvalue} and the problem
at hand is a special case. From Eq. \eqref{eq:E2}, dropping the subscripts
$0$ of eigen{\it vectors} with an understanding that they are those of
the Toeplitz matrix ($\sigma=0$), we have
\begin{eqnarray}
\mathbb{E}\left(E_{2}^{\ell}\right) & = & \sum_{j\ne\ell}\frac{\sum_m \left(\tilde{\psi}_{m}^{\ell}\psi_{m}^{j}\tilde{\psi}_{m}^{j}\psi_{m}^{\ell}\right)\mathbb{E}\left(v_{m}^{2}\right)}{E_{0}^{\ell}-E_{0}^{j}}.\label{eq:E_E2_analyze}
\end{eqnarray}
Suppose $V$ is real, $\mathbb{E}\left(v_{i}^{2}\right)>0$ in the
sum (Eq. \eqref{eq:E_E2_analyze}). Let us pick any $E^{\ell}$ in that
sum and analyze the effect of its complex conjugate on its motion.
If $E_{0}^{\ell}$ and $E_{0}^{j}$ are a complex conjugate pair,
then $|\psi^{\ell}\rangle$ and $|\psi^{j}\rangle$ are complex conjugates
as well and we have (we explicitly insert the sum over $m$)
\[
\sum_m \tilde{\psi}_{m}^{\ell}\tilde{\psi}_{m}^{j}\psi_{m}^{j}\psi_{m}^{\ell}=\sum_{m}|\tilde{\psi}_{m}^{\ell}|^{2}|\psi_{m}^{\ell}|^{2}
\]
and hence
\begin{eqnarray}
\frac{\sum_m \left(\tilde{\psi}_{m}^{\ell}\psi_{m}^{j}\tilde{\psi}_{m}^{j}\psi_{m}^{\ell}\right) \mathbb{E}\left(v_{m}^{2}\right)}{E_{0}^{\ell}-E_{0}^{j}} & = & -i\frac{\mbox{ }\mathbb{E}\left(v^{2}\right)\sum_{m}|\psi_{m}^{\ell}|^{2}|\tilde{\psi}_{m}^{\ell}|^{2}}{2\Im\left(E_{0}^{\ell}\right)}=-i\frac{n\mbox{ }\mathbb{E}\left(v^{2}\right)}{2\Im\left(E_{0}^{\ell}\right)},\label{eq:attraction}
\end{eqnarray}
where we assumed that the entries have equal second moments. The last equality follows from Eqs. \eqref{eq:Eigenvector_Final}
and \eqref{eq:LeftEigvector_Final}. This is non-zero for real $V$,
i.e., $\mathbb{E}(v^{2})\ne0$ if $v\in\mathbb{R}$ and $\mathbb{E}(v^{2})=0$
if $v\in\mathbb{C}$ (see Fig. \ref{fig:Runaway-Type-II:Trajectory}).

Note that if $\Im\left(E_{0}^{\ell}\right)>0$, then the eigenvalue
is pushed down along the imaginary axis as the right hand side of
Eq. \eqref{eq:attraction} is a negative imaginary number. Further,
if $\Im\left(E_{0}^{\ell}\right)<0$, then the right hand side of
Eq. \eqref{eq:attraction} is a positive imaginary number and the eigenvalue
is pushed up along the imaginary axis. This establishes that in the
summation Eq. \eqref{eq:E_E2_analyze} the complex conjugate pairs attract
one another. Moreover, the eigenpairs close to the real axis attract
most strongly as $\Im\left(E_{0}^{\ell}\right)$ in the denominator
of Eq. \eqref{eq:attraction} will be smallest (see the right most part
of Fig. \ref{fig:TDisorder}). \end{proof}

We make some comments on the context and corollaries to this :
\begin{enumerate}
\item Complex conjugate pairs of eigenvalues close to the real line attract one another
until they collide on the real line by becoming momentarily degenerate.
See \cite{movassagh2016eigenvalue} for more general discussions. Generically such collisions lead to so called an exceptional point, where the rank of the matrix decreases by one. However, in our case we find that at the moment of collision the algebraic and geometric multiplicities are equal and the matrix is invertible. This is further confirmed by continuity of eigenvectors and the low condition number of the real eigenvalues observed immediately after the collision. In Fig. \ref{fig:Runaways-type-I:} see  eigenvalues labeled $78$, $79$, $80$, $81$.
\item The numerator in Eq. \eqref{eq:attraction} is proportional to $n$.
Hence as the size of the matrices become larger, the eigenvalue attraction
becomes more dominant.
\item For majority of eigenvalues in the bulk and away from the real line, the contribution to $E^{\ell}$ from $E^{\ell-1}$
and $E^{\ell+1}$ nearly cancel if $E^{\ell}$ is close to $E^{\ell\pm1}$.
\begin{figure}
\includegraphics[scale=0.3]{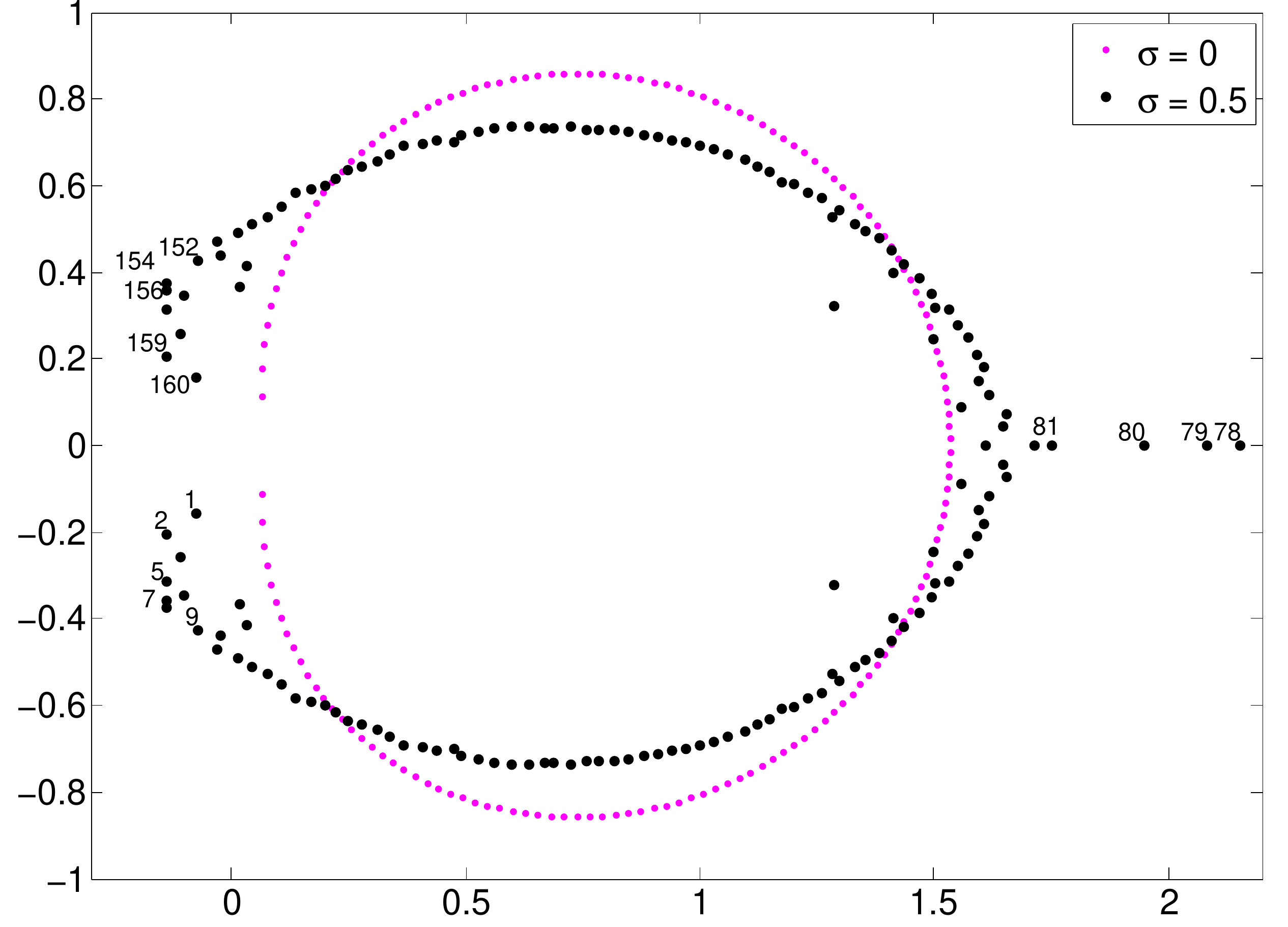}\includegraphics[scale=0.26]{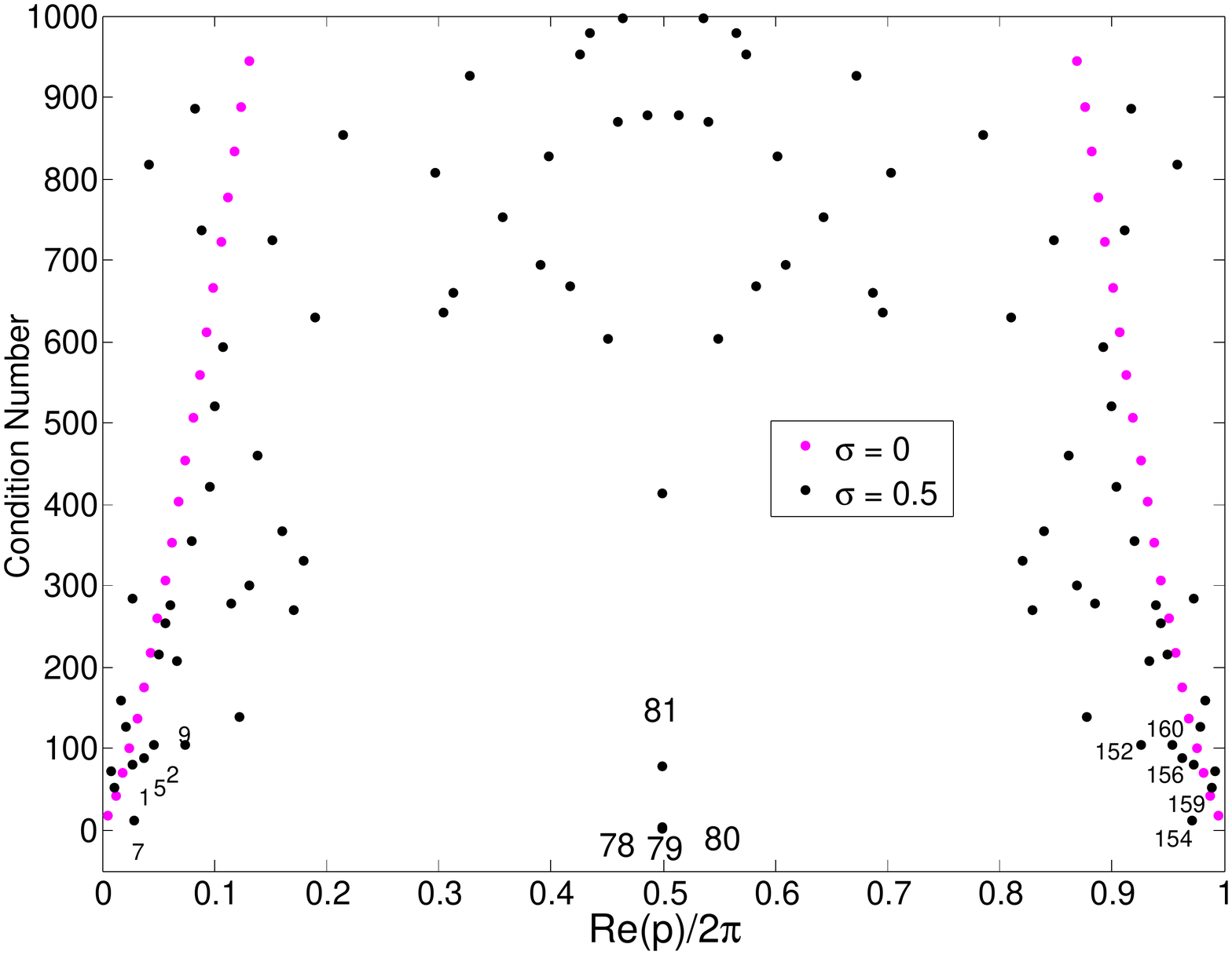}\caption{\label{fig:Runaways-type-I:}Runaways type I: We have zoomed in the
lower part of the condition number plot (compare with Fig. \ref{fig:Runaways-type-II}).
These eigenvalues, labeled $78$, $79$, $80$, $81$, act more normal
(well-conditioned) than the unperturbed counter part whose condition
numbers are about $3000$. The condition numbers for $78$, $79$
and $80$ are very similar in value (overlapping dots). }
\end{figure}
\item In traditional quantum mechanics  one works with Hermitian matrices and
\[
E_{2}^{\ell}=\sum_{j\ne\ell}\frac{\langle\psi^{\ell}|\mbox{ }V\mbox{ }|\psi_{0}^{j}\rangle\langle\psi^{j}|\mbox{ }V\mbox{ }|\psi_{0}^{\ell}\rangle}{E_{0}^{\ell}-E_{0}^{j}}=\sum_{j\ne\ell}\frac{|\langle\psi_{0}^{\ell}|\mbox{ }V\mbox{ }|\psi_{0}^{j}\rangle|^{2}}{E_{0}^{\ell}-E_{0}^{j}},
\]
where because hermitian matrices are normal $\langle\widetilde{\psi}|=|\psi\rangle^{\dagger}=\langle\psi|$
as well as $\overline{\langle\psi^{\ell}|\mbox{ }V\mbox{ }|\psi_{0}^{j}\rangle}=\langle\psi^{j}|\mbox{ }V\mbox{ }|\psi_{0}^{\ell}\rangle$. A \textit{repulsion} of  eigenvalues is evident
with a strength proportional to the inverse of the distance.
\end{enumerate}

In the right part of Fig. \ref{fig:Runaways-type-I:} we explicitly
show the condition number corresponding to normal eigenvalues labeled
$78$ $79$, $80$, $81$; one sees that the perturbation makes these
eigenvalues more well-conditioned. It is possible that the attraction
of distant complex conjugate pairs to be strong despite the denominator
in Eq. \eqref{eq:attraction} being large. This can result when the
eigenvalue is ill-conditioned; i.e., it has a large condition number
because $\left\Vert \langle\tilde{\psi}^{\ell}|\right\Vert _{2}$
becomes large. 
\begin{figure}
\includegraphics[scale=0.25]{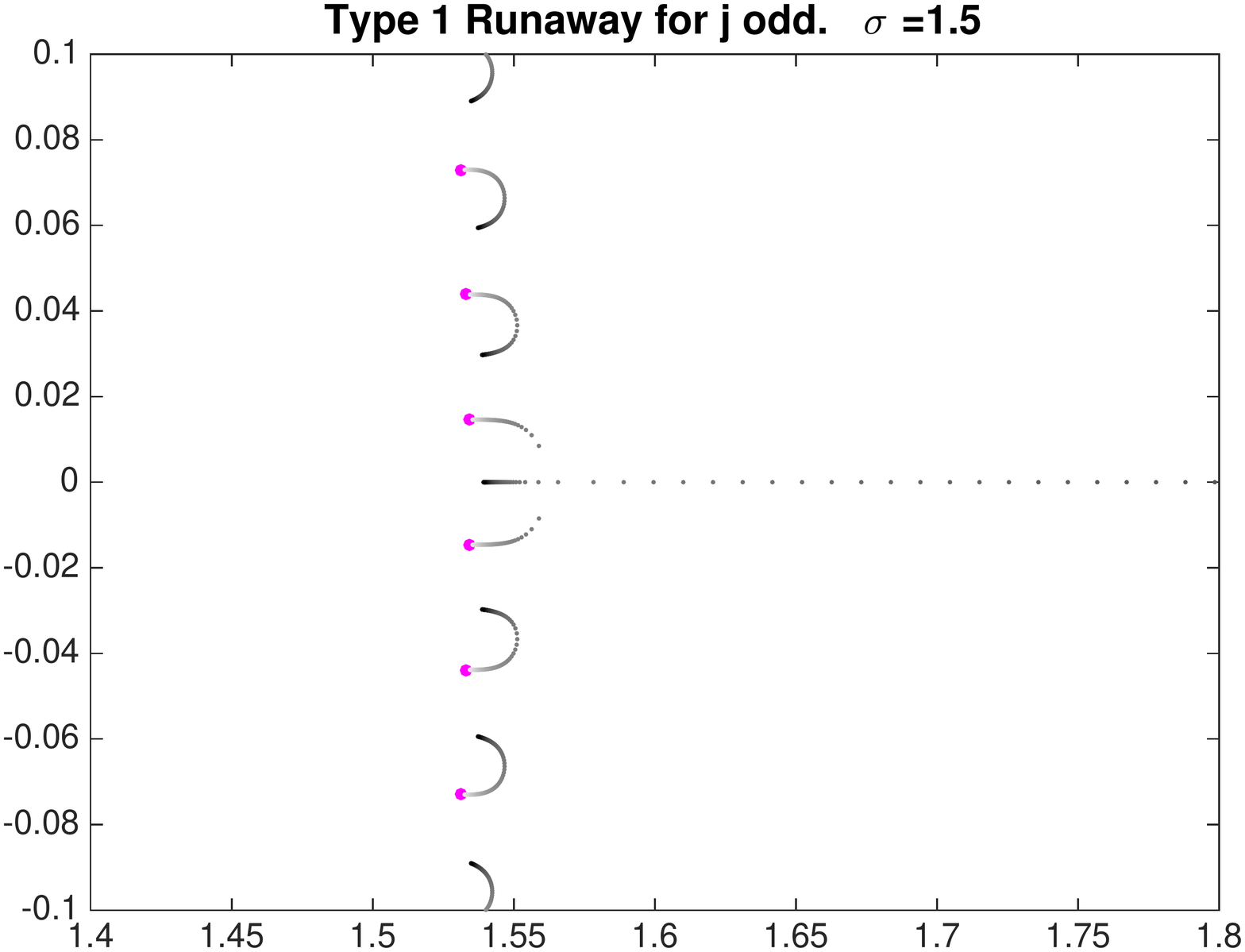}\includegraphics[scale=0.25]{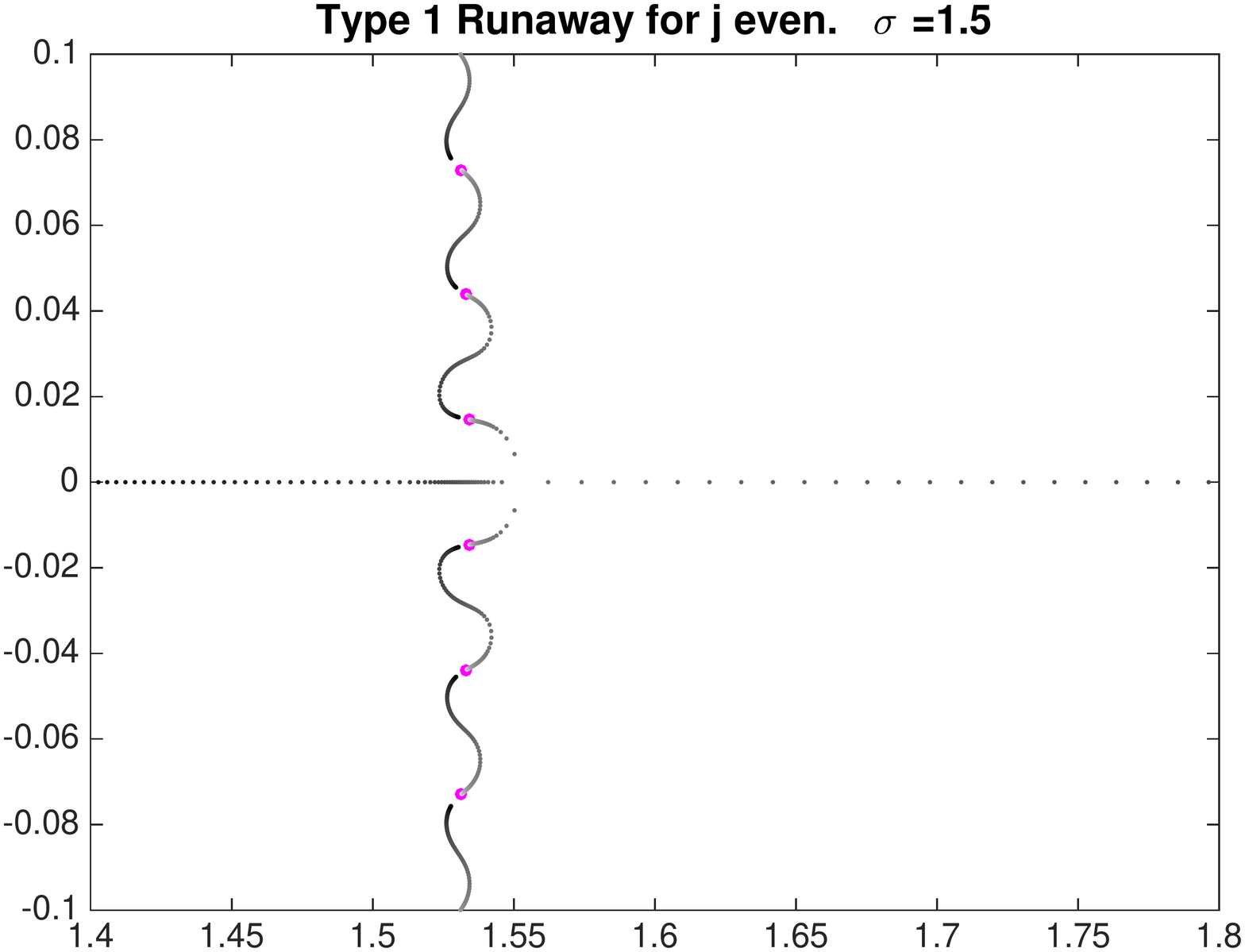}\caption{\label{fig:Type-I-Runaways}Real-valued Runaways for $T+\sigma A_{jj}$.
If $j$ is even one of the eigenvalues gets pushed into the spectrum
by the rest.}
\end{figure}
\begin{rem}
In addition to the runaway type I eigenvalues of the Toeplitz matrix herein, the eigenvalue attraction in its general form was shown to account for the formation of the "wings" seen in the Hatano-Nelson model \cite{movassagh2016eigenvalue}.  For the discovery and earlier discussions of the real eigenvalues of Hatano-Nelson model see \cite{HatanoNelson1997,trefethen2000spectra, TrefethenEmbree2005}.
\end{rem}

\subsubsection{\label{sub:Runaways-Type-II}Runaways type II}

The lack of stability and high sensitivity of an eigenvalue to
perturbations suggests that the eigenvalue is \textit{ill-conditioned}.
Are there scalar measure of non-normality? Embree and Trefethen \cite{TrefethenEmbree2005}
give an overview of such measures. Here we quote what is useful to
this work. Let $M$ be any diagonalizable matrix, i.e., $M=W^{-1}\Lambda W$,
where $\Lambda$ is the diagonal matrix of the eigenvalues and $W$
is the matrix of eigenvectors. A measure of non-normality is the condition
number of a matrix $\kappa\left(W\right)=\left\Vert W\right\Vert _{2}\left\Vert W^{-1}\right\Vert _{2}$. If $M$ is normal $\kappa\left(W\right)=1$ with the
right choice of $W$, while it can be arbitrary large for near-defective
matrices. 

In the problem at hand, we have realized that the eigenvalues have
a rich behavior (see Figs. \ref{fig:Type-I-Runaways} and \ref{fig:Runaway-Type-II:Trajectory})
some of which seem to be relatively stable against perturbations (bulk eigenvalues),
others attract and move to the real line and become very stable thereafter, i.e.,
nearly normal (Runaways type I). There are some that act differently
and fall into the category of ill-conditioned, which we denote by
Runaways type II. 

In Fig. \ref{fig:TDisorder}, some of the eigenvalues
leave the bulk and move into the complex plane (inward motion). As
stated above the motion of eigenvalues is continuous; what happens
is that Runaways type II's move faster 
than the ones belonging to the bulk (those tracing an approximate
ellipse).

We need a more refined definition applicable to an individual eigenvalue
to quantify ill-conditioning. For any simple eigenvalue, $E_{0}^{\ell}$,
one defines its condition number  by \cite[Sec. 52]{TrefethenEmbree2005}
\begin{equation}
\kappa\left(E_{0}^{\ell}\right)=\frac{\left\Vert \langle\tilde{\psi}_{0}^{\ell}|\right\Vert \left\Vert |\psi_{0}^{\ell}\rangle\right\Vert }{|\langle\tilde{\psi}_{0}^{\ell}|\psi_{0}^{\ell}\rangle|}=\frac{1}{\left|\cos\theta_{0}^{\ell}\right|},\label{eq:conditionNumber}
\end{equation}
where $\theta_{0}^{\ell}$ is the angle between the right and left
eigenvectors corresponding to the $\ell^{\mbox{th}}$ eigenvalue and
we used orthonormality of right and left eigenvectors and unity of
the norm of $|\psi_{0}^{\ell}\rangle$ . By Cauchy-Schwarz $\kappa\left(E_{0}^{\ell}\right)\ge1$.
In contrast eigenvalues for which $\kappa\left(E^{\ell}\right)\gg1$
are called \textit{ill-conditioned eigenvalues}.

Let us consider the following simple model 
\begin{equation}
T\left(\sigma\right)=T+\sigma A_{jk}, \label{eq:Rank_1_pertur-1}
\end{equation}
where $A_{jk}$ is a rank-1 matrix that has a one in the $j,k$ entry
and zeros everywhere else; mathematically $A_{jk}\equiv e_{j}e_{k}^{*}=|e_{j}\rangle\langle e_{k}|$.
The eigenvalues of $T(\sigma)$ are the zeros of 
\[
\det(T+\sigma A_{jk}-\lambda\mathbb{I})=\det(T-\lambda\mathbb{I})\det\left(\mathbb{I}+\sigma\frac{A_{jk}}{T-\lambda\mathbb{I}}\right),
\]
because after a perturbation generically $\lambda\notin\mbox{spec}(T)$,
and it must be that $\det(\mathbb{I}+\sigma\frac{A_{jk}}{T-\lambda\mathbb{I}})=0$.
Suppose $|u\rangle$ is the eigenvector corresponding to a zero
eigenvalue, then
\[
(\mathbb{I}+\sigma\frac{A_{jk}}{T-\lambda\mathbb{I}})|u\rangle=|u\rangle+\sigma\langle e_{k}|\frac{1}{T-\lambda\mathbb{I}}|u\rangle|e_{j}\rangle.
\]
But $\langle e_{k}|\frac{1}{T-\lambda\mathbb{I}}|u\rangle$ is just
a number, and it must be that $|u\rangle\propto|e_{j}\rangle$, so
we use $|u\rangle=|e_{j}\rangle$ as the eigenvector to get 
\[
|e_{j}\rangle+\sigma\langle e_{k}|\frac{1}{T-\lambda\mathbb{I}}|e_{j}\rangle|e_{j}\rangle=0.
\]
Therefore, $1+\sigma\langle e_{k}|\frac{1}{T-\lambda\mathbb{I}}|e_{j}\rangle=0$
and we have that $\lambda$ is the implicit solution of 
\[
\langle e_{k}|\frac{1}{T-\lambda\mathbb{I}}|e_{j}\rangle=-\frac{1}{\sigma},
\]
where  $T-\lambda\mathbb{I}$ is also a Toeplitz matrix. Let the matrix representation of the resolvent be $R(\lambda)=1/(T-\lambda\mathbb{I})$,
then $\lambda$ is found by solving 
\begin{equation}
\left[R(\lambda)\right]_{kj}\equiv\langle e_{k}|R(\lambda)|e_{j}\rangle=-\frac{1}{\sigma}.\label{eq:SokolovCondition}
\end{equation}
Any further progress requires that we solve for the resolvent. Recall
that we have the eigenvectors and eigenvalues of $T$ and we can write
$T= U^{*}\Lambda W$ where $W$ is the matrix
of eigenvectors and $U$ is the matrix of left eigenvectors. We have
$\frac{1}{T-\lambda\mathbb{I}}=U^{*}\frac{1}{\Lambda-\lambda\mathbb{I}}W$
\begin{equation}
\langle\tilde{\psi}^{k}|\frac{1}{\Lambda-\lambda\mathbb{I}}|\psi^{j}\rangle=-\frac{1}{\sigma}\label{eq:TheRawEquation}
\end{equation}
where as before $|\psi^{j}\rangle$ is the $j^{\mbox{th}}$ eigenvector and
$\langle\tilde{\psi}^{k}|$ is the $k^{\mbox{th}}$ left eigenvector. Moreover,
\[
\frac{1}{\Lambda-\lambda\mathbb{I}}=\left(\begin{array}{cccc}
\frac{1}{E^{0}-\lambda}\\
 & \frac{1}{E^{1}-\lambda}\\
 &  & \ddots\\
 &  &  & \frac{1}{E^{(n-1)}-\lambda}
\end{array}\right)=\sum_{m=0}^{n-1}\frac{|e_{m}\rangle\langle e_{m}|}{E^{m}-\lambda}.
\]
Using the above expression for left and right eigenvectors, we seek
$\lambda$ that solves 
\begin{eqnarray}
\sum_{m=0}^{n-1}\frac{1}{E^{m}-\lambda}\langle\tilde{\psi}^{k}|e_{m}\rangle\langle e_{m}|\psi^{j}\rangle & = & \sum_{m=0}^{n-1}\frac{\tilde{\psi_{m}}^{k}\psi_{m}^{j}}{E^{m}-\lambda}=-\frac{1}{\sigma}.\label{eq:SumToSolve}
\end{eqnarray}
The solution with respect to $\lambda$ of the foregoing equation
is implicit and predicts where the eigenvalues are as a function of
$\sigma$ and $A_{jk}$. Even though we have the analytical expression
for the eigenvectors, the solution of the above equation for $\lambda$
is in general hard to obtain.

Consider the following simple special case 
\begin{equation}
T\left(\sigma\right)=T+\sigma A_{jj}\quad.\label{eq:Rank_1_pertur}
\end{equation}

This is the simplest and an insightful deformation of $T$. We have empirically discovered many features of this simple deformation that we do not have proofs for. In particular, upon examining the perturbation corrections (e.g., see Fig. \ref{fig:First-and-second_Ejj}),
we are lead to the following conjecture:
\begin{figure}
\includegraphics[scale=0.25]{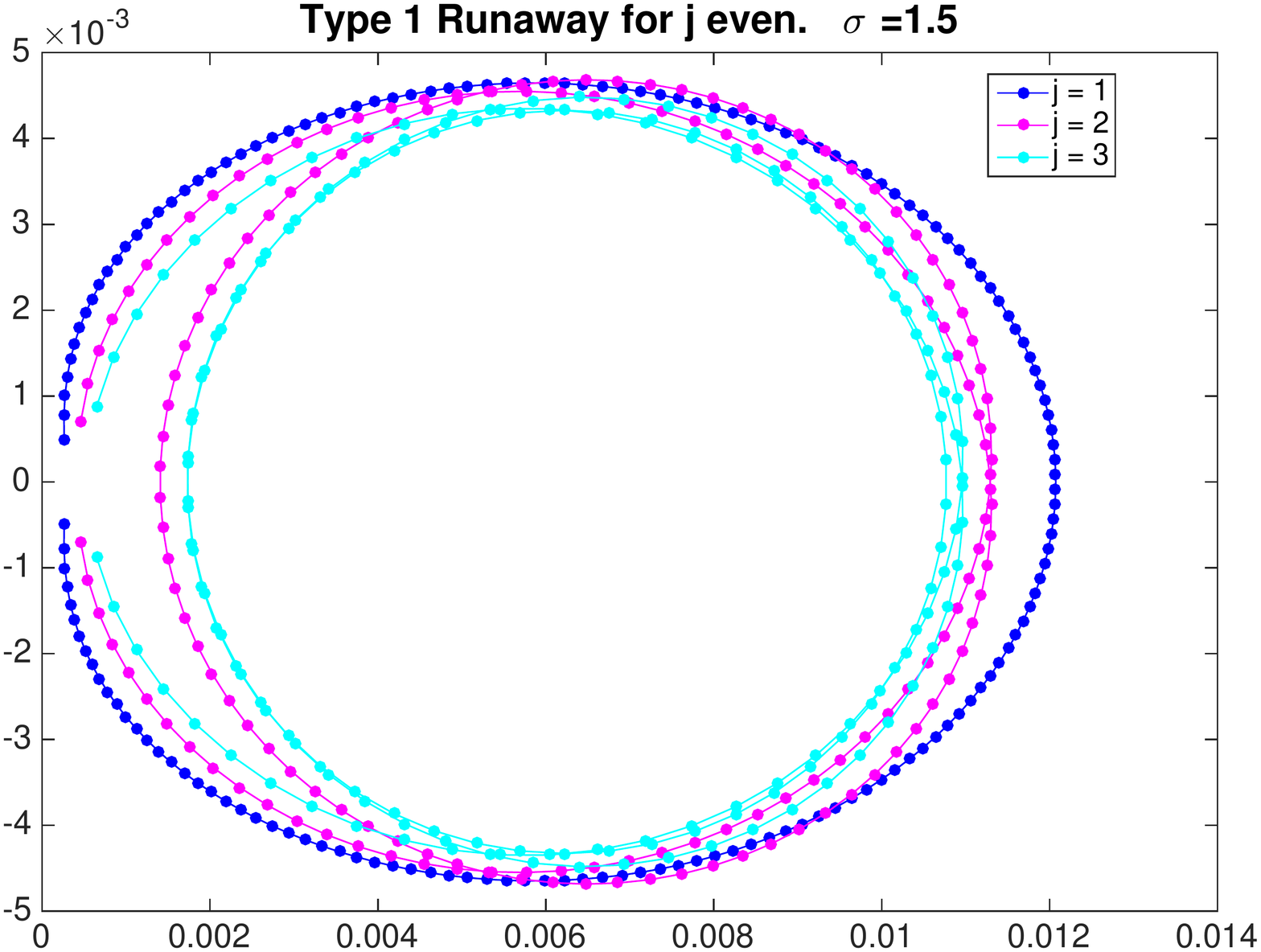}\includegraphics[scale=0.25]{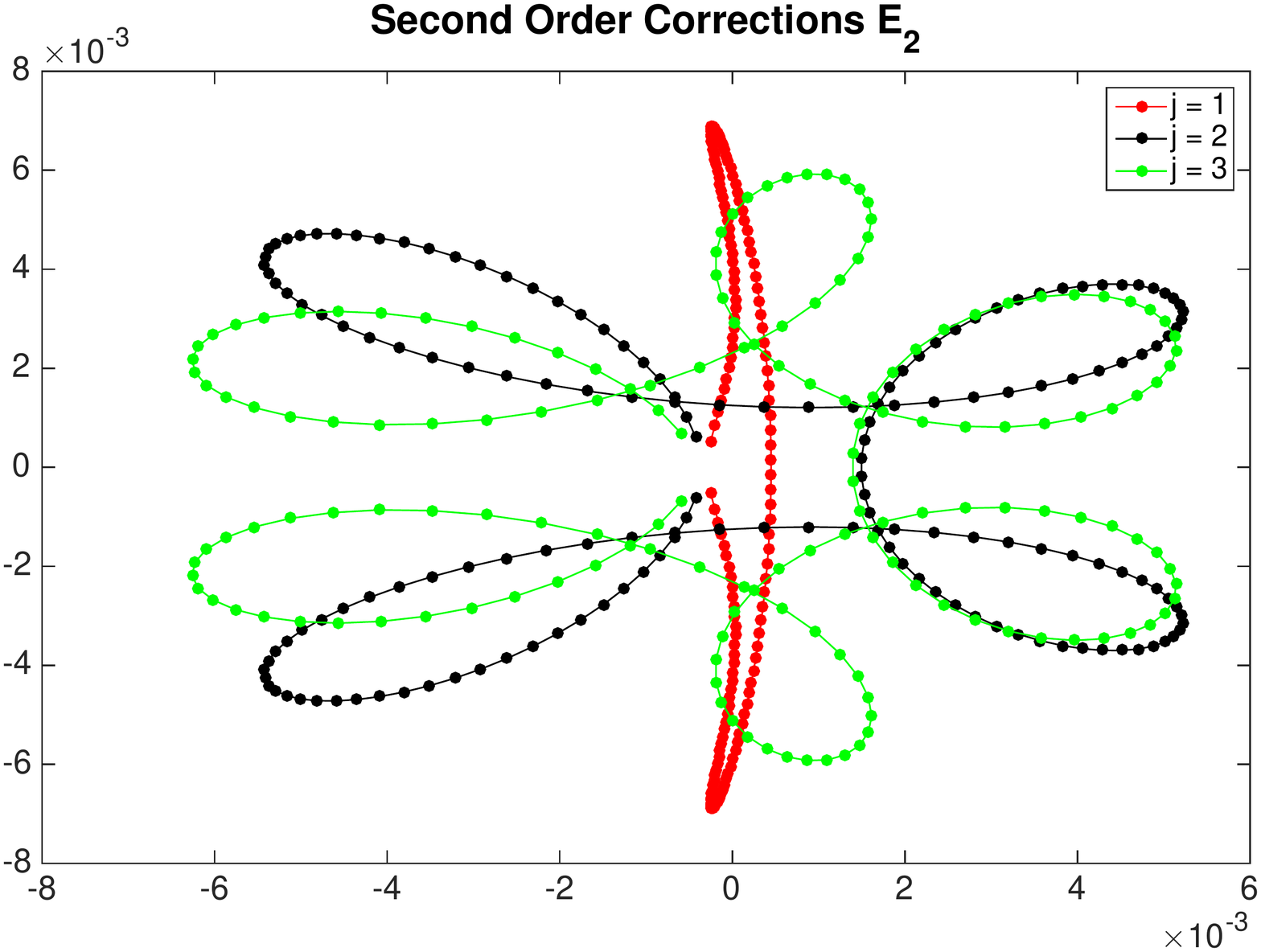}
\caption{\label{fig:First-and-second_Ejj}Illustration of Conjecture \ref{conj:winding}. First and second order perturbation
corrections for rank-one perturbation of the Toeplitz matrix $T$ (Eq. \eqref{eq:T_leo_Matrix}) of size $n=160$  with $\alpha=1/3$ and $\beta=-1/2$.}
\end{figure}

\begin{conjecture}
\label{conj:winding}
The number of Runaways type II's in $T(\sigma)=T+\sigma A_{jj}$, is
equal to the winding number of the first order perturbation correction
about any point in the interior of the convex hull of the first order
corrections. 
\end{conjecture}

This conjecture is illustrated in Fig. \ref{fig:First-and-second_Ejj}, where examining the winding number about an interior point of the first and second order perturbation corrections suggest that the number of Runaways type II's are in a one to one correspondence. We hope that in the future  this connection becomes clearer.

For all $j$ the complex conjugate. attraction forces the pair with the smallest
imaginary parts to collide on the real line and become real. However,
only when $j$ is even does one of them move substantially farther
into the spectrum (i.e., to the left). We show this behavior in Figs.
\ref{fig:Eig_Cond_E_jj_jOdd} and \ref{fig:Eig_Cond_E_jj_jEven} and
we find that, when $j$ is even, the first and second order corrections
in perturbation theory are comparable in value, i.e., $|E_{2}/E_{1}|\approx1$.
Whereas, when $j$ is odd, $|E_{2}/E_{1}|<1$ and is especially small
for smaller $j$'s.  We are lead to the following conjecture:

\begin{conjecture}
\label{conj:j_Runaways}
The number of Runaways in $T(\sigma)=T+\sigma A_{jj}$ is exactly
$j$ for $1\le j<n/2$. By the Toeplitz symmetry the number of Runaways
in $T(\sigma)=T+\sigma A_{(n-j+1),(n-j+1)}$ is exactly $j$ as well
for $1\le j<n/2$.
\end{conjecture}

\begin{figure}
\includegraphics[scale=0.22]{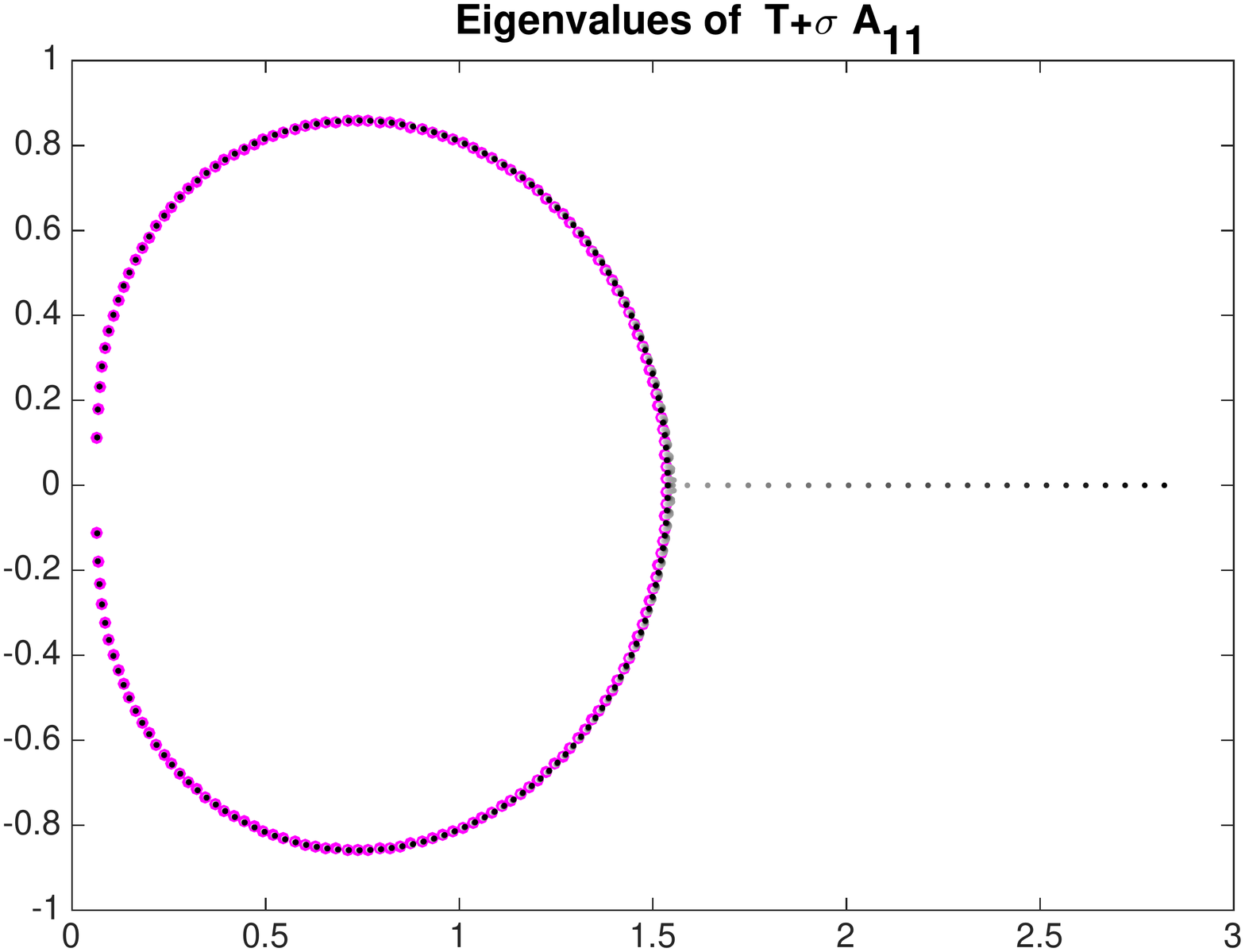}\includegraphics[scale=0.22]{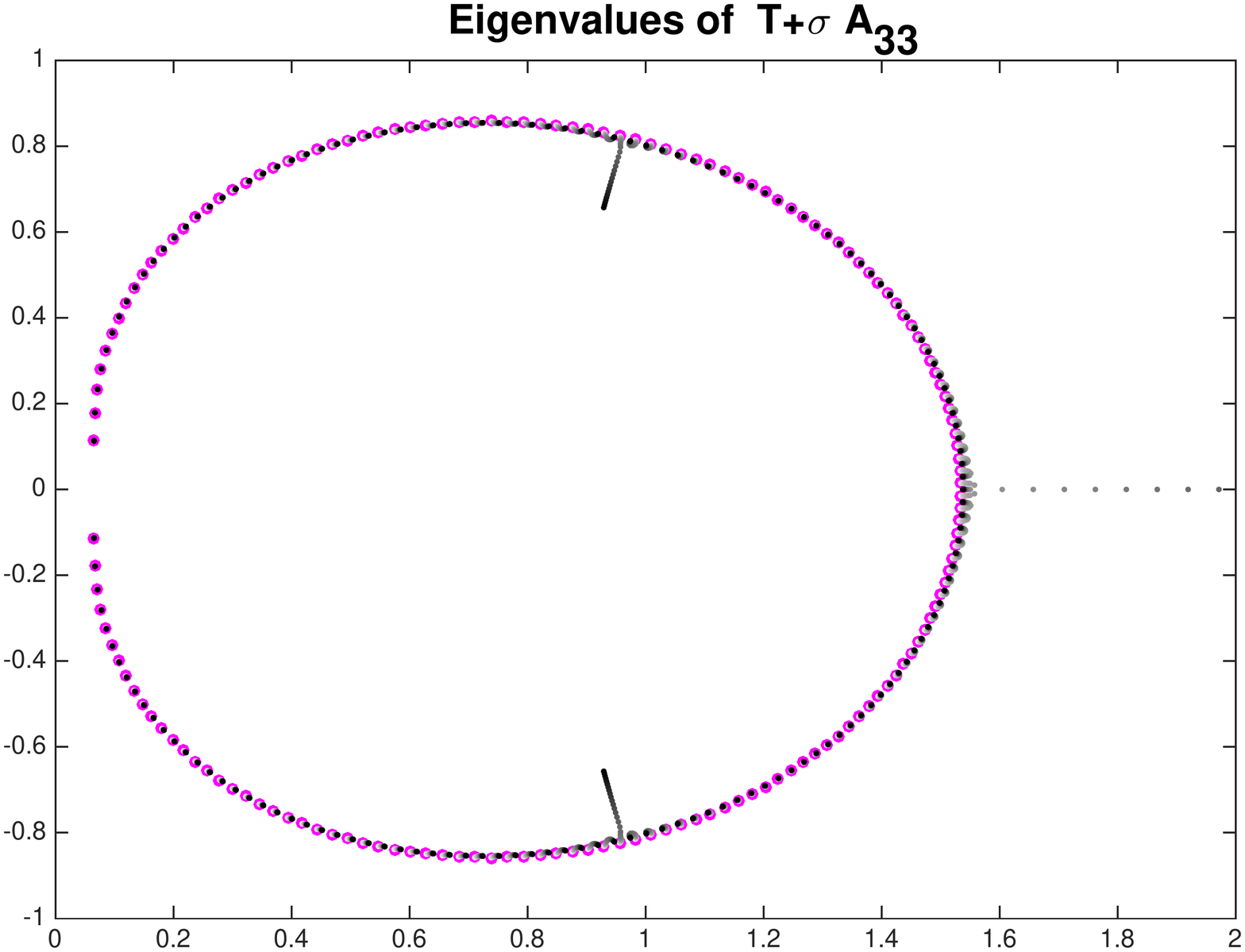}\includegraphics[scale=0.22]{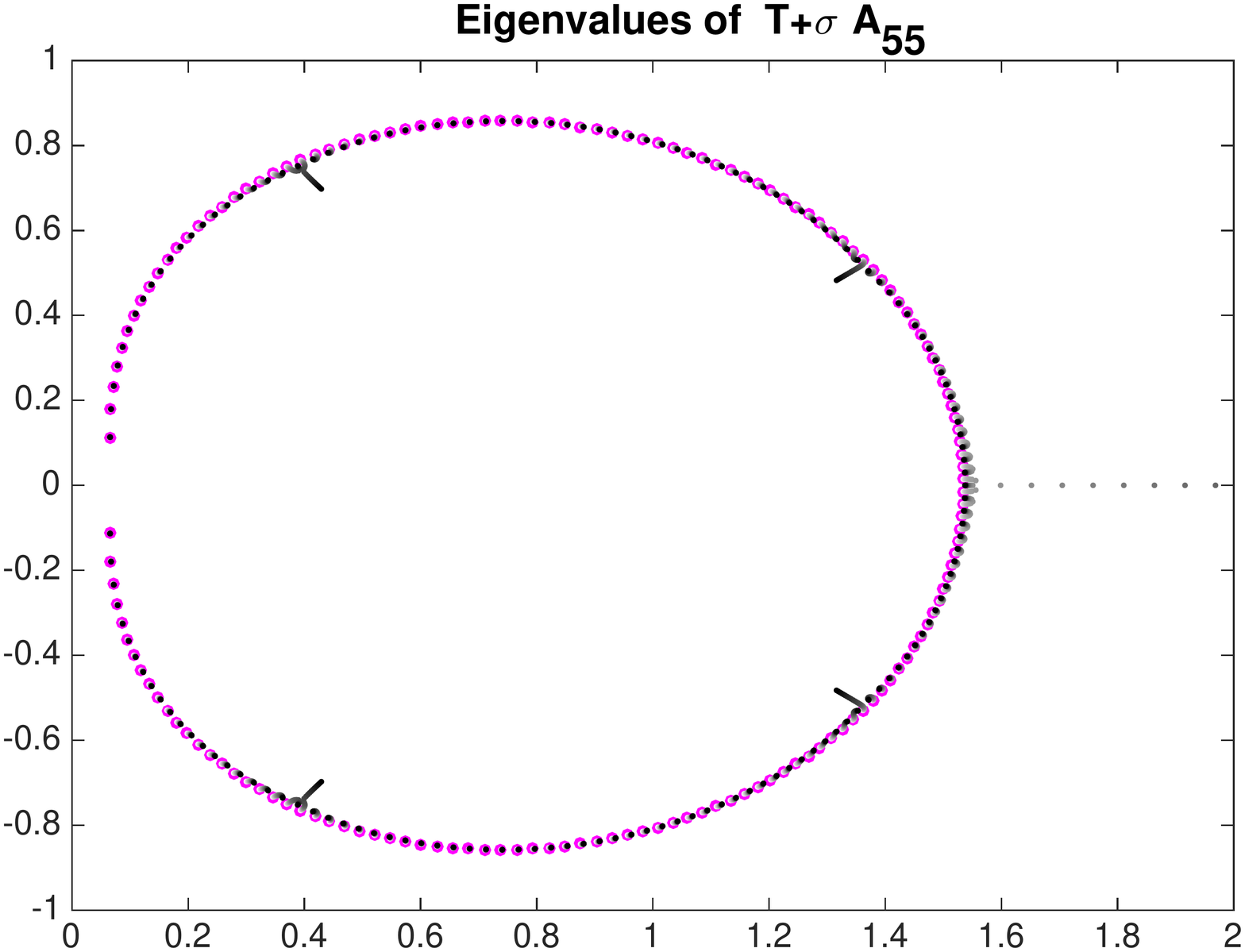}\\
\includegraphics[scale=0.22]{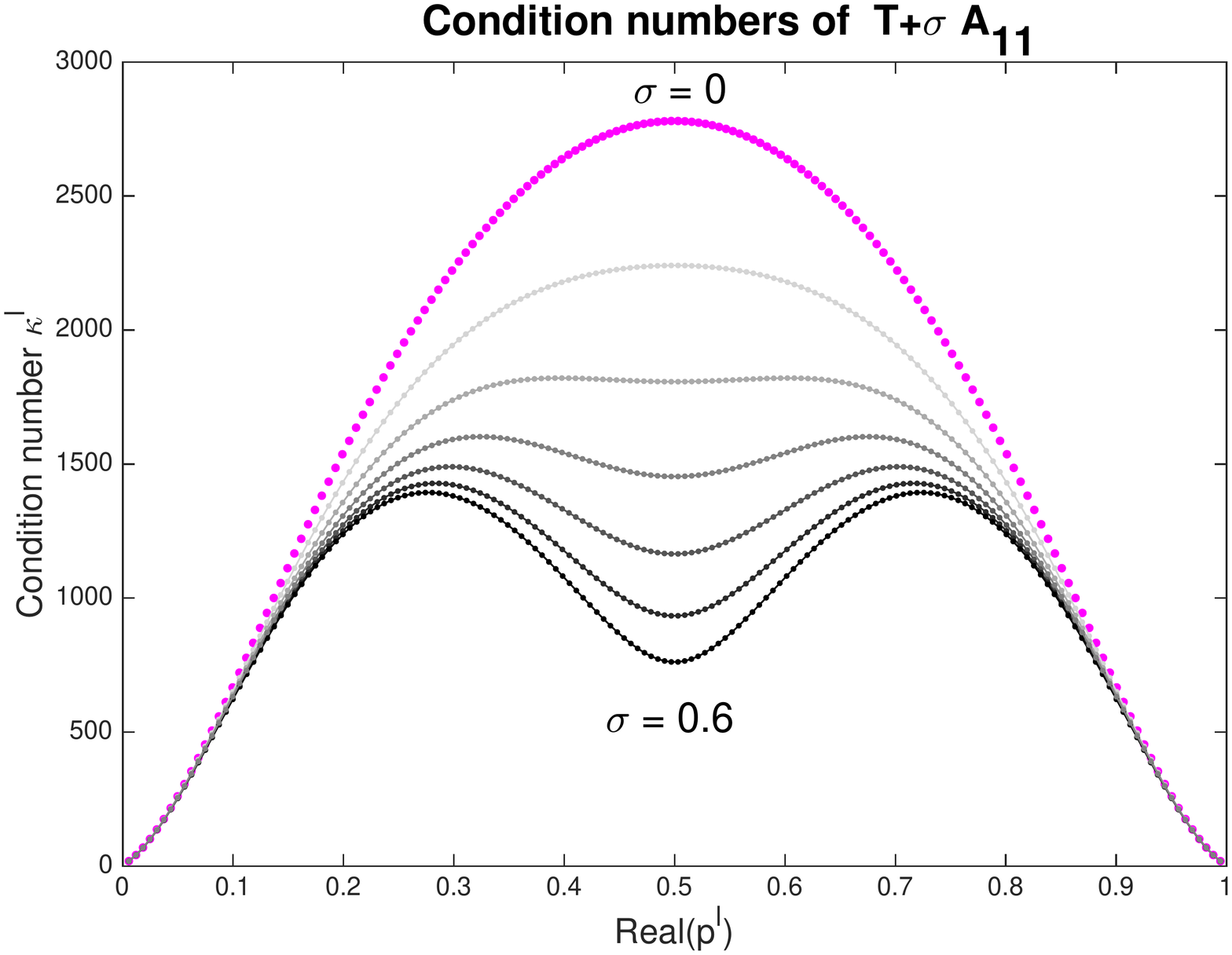}\includegraphics[scale=0.22]{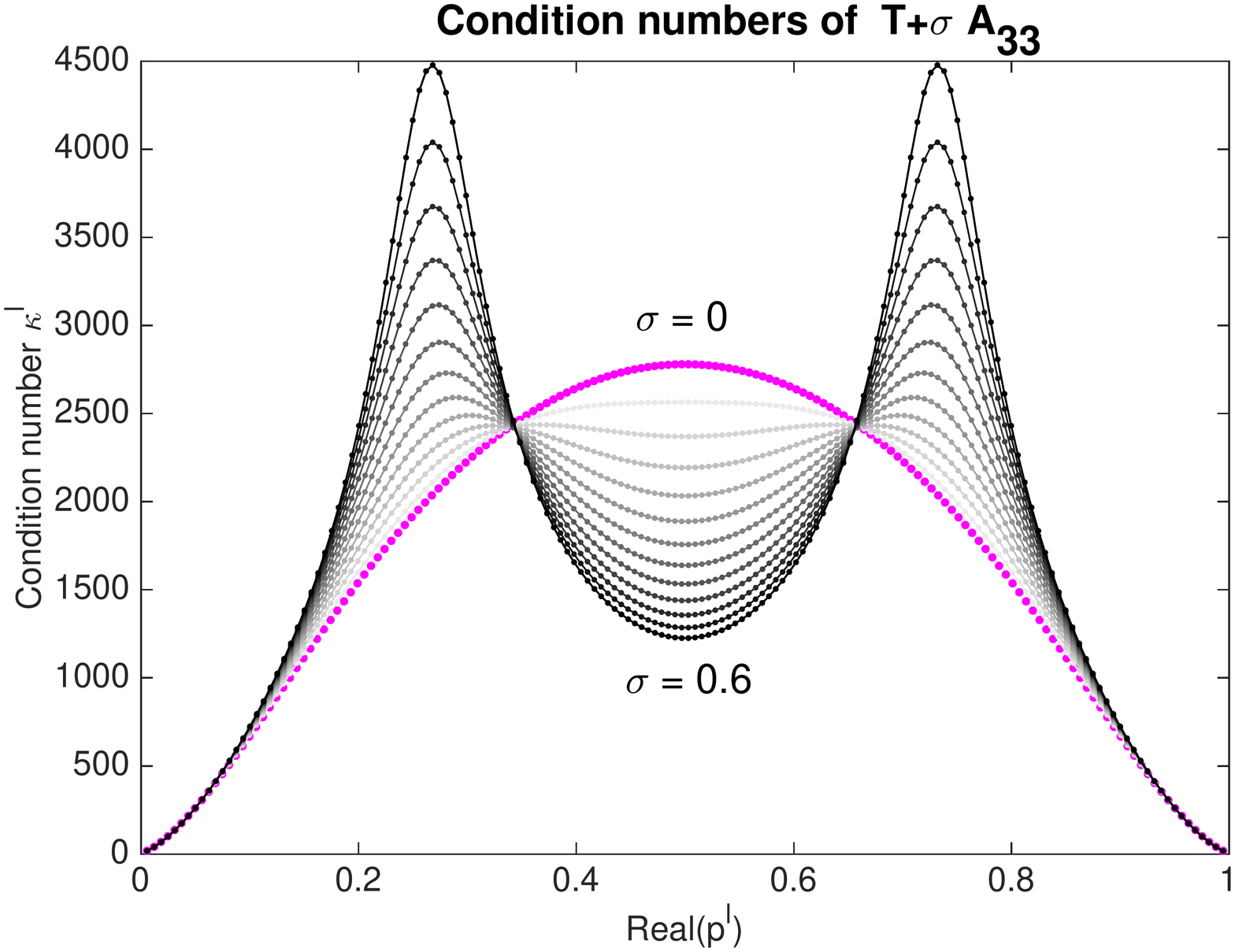}\includegraphics[scale=0.22]{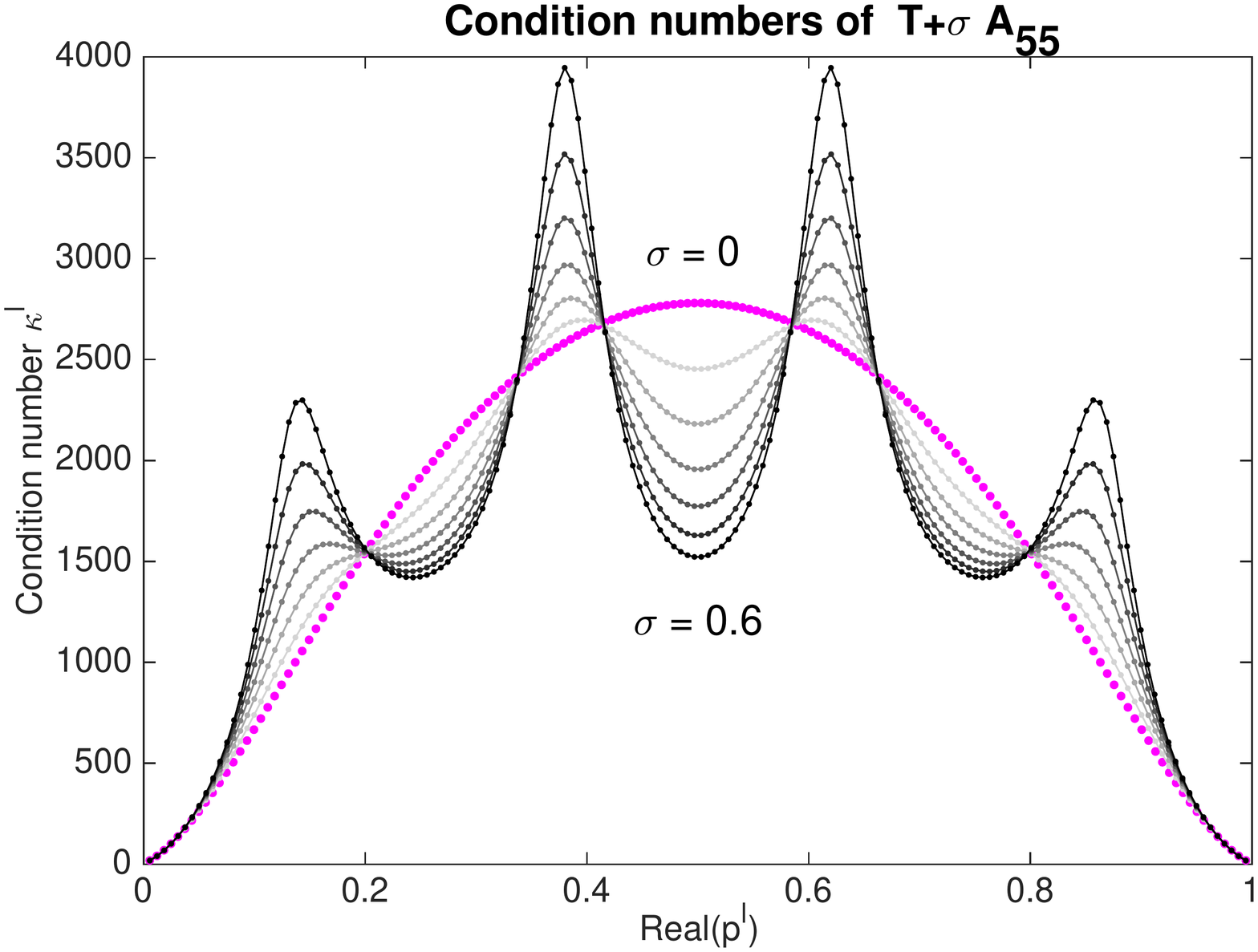}

\caption{\label{fig:Eig_Cond_E_jj_jOdd}Illustration of Conjecture \ref{conj:j_Runaways}: Eigenvalues and Condition numbers for
$T+\sigma A_{jj}$ with $j$ odd, where $T$ is the $n=160$ Toeplitz matrix as before.}
\end{figure}
\begin{figure}
\includegraphics[scale=0.22]{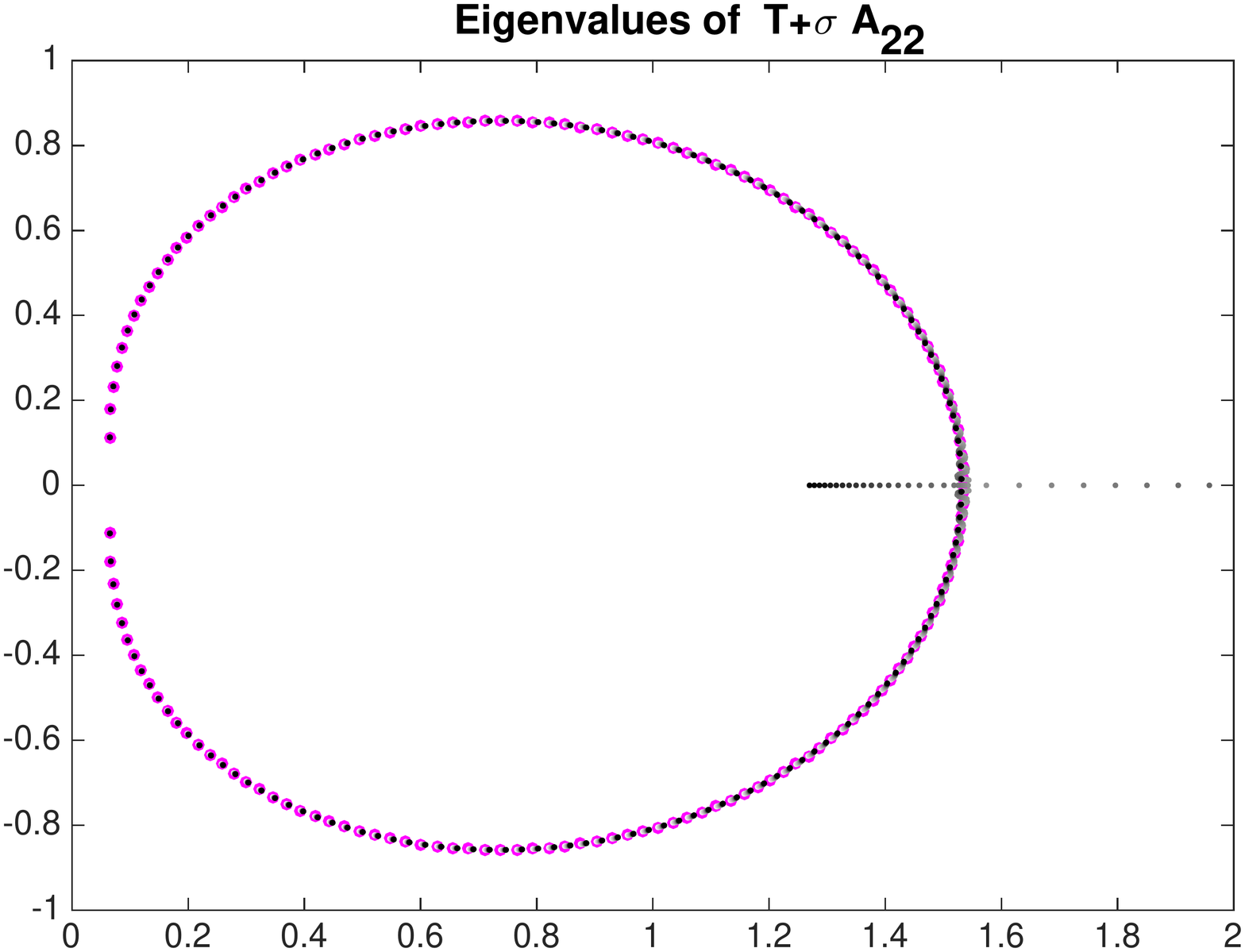}\includegraphics[scale=0.22]{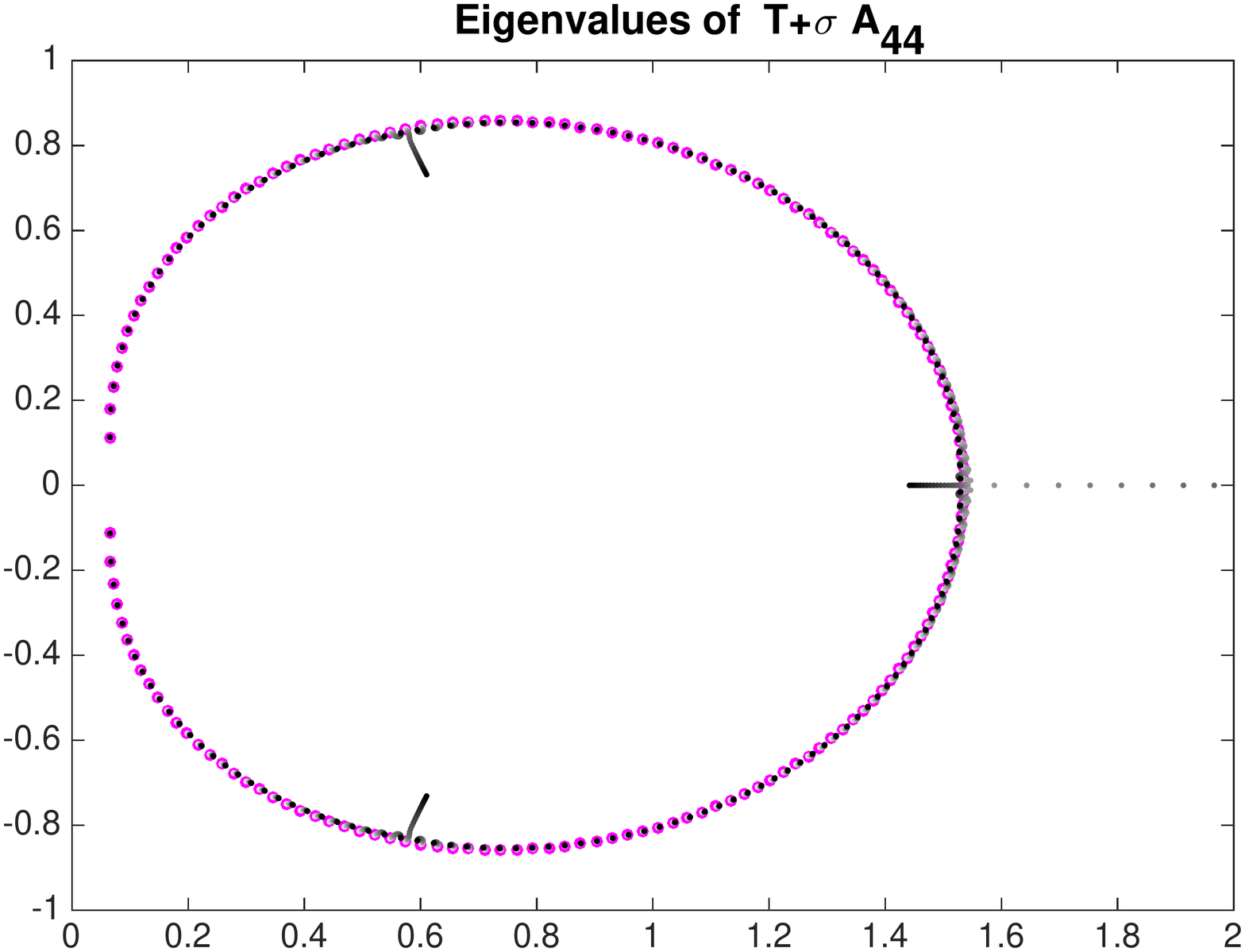}\includegraphics[scale=0.22]{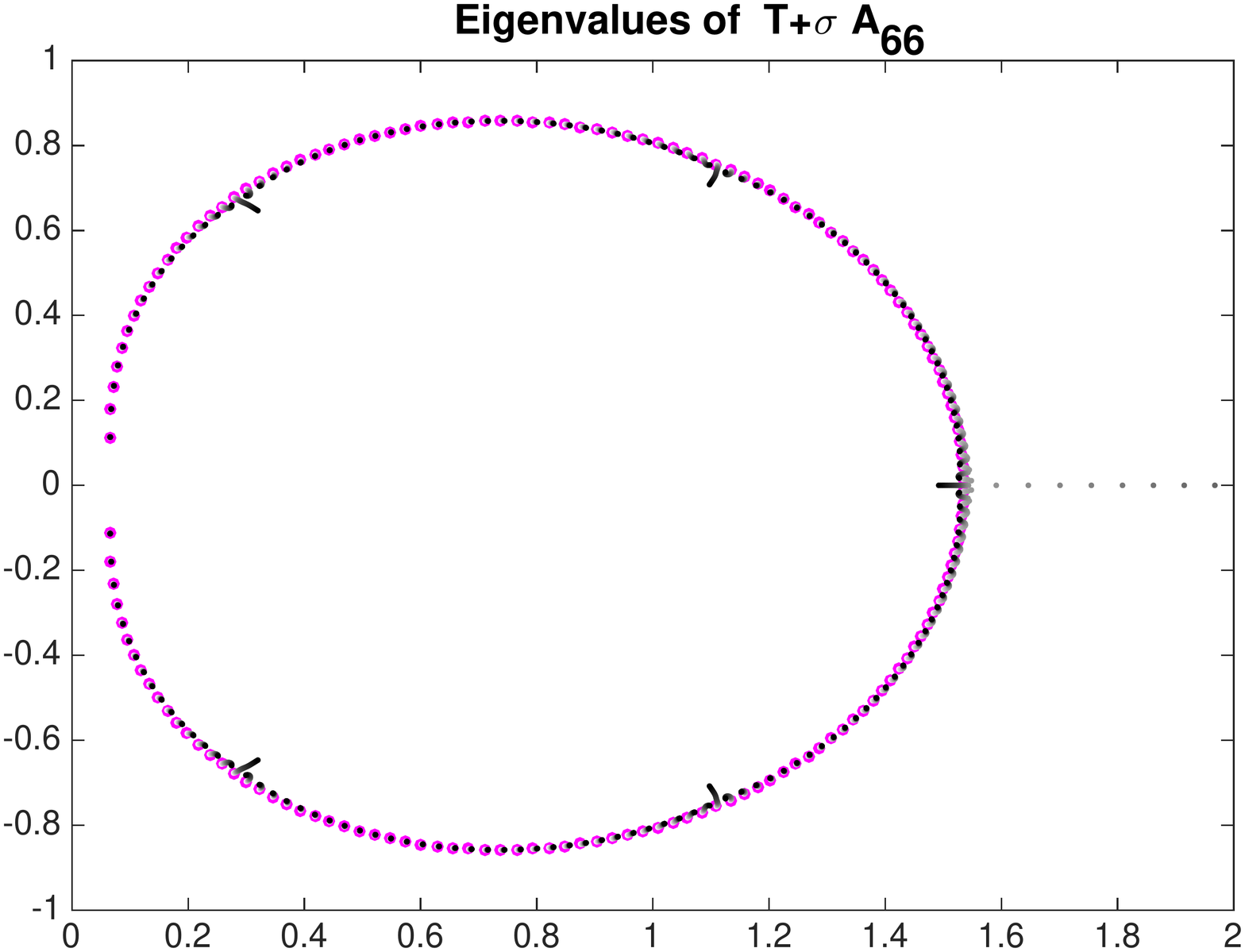}\\
\includegraphics[scale=0.22]{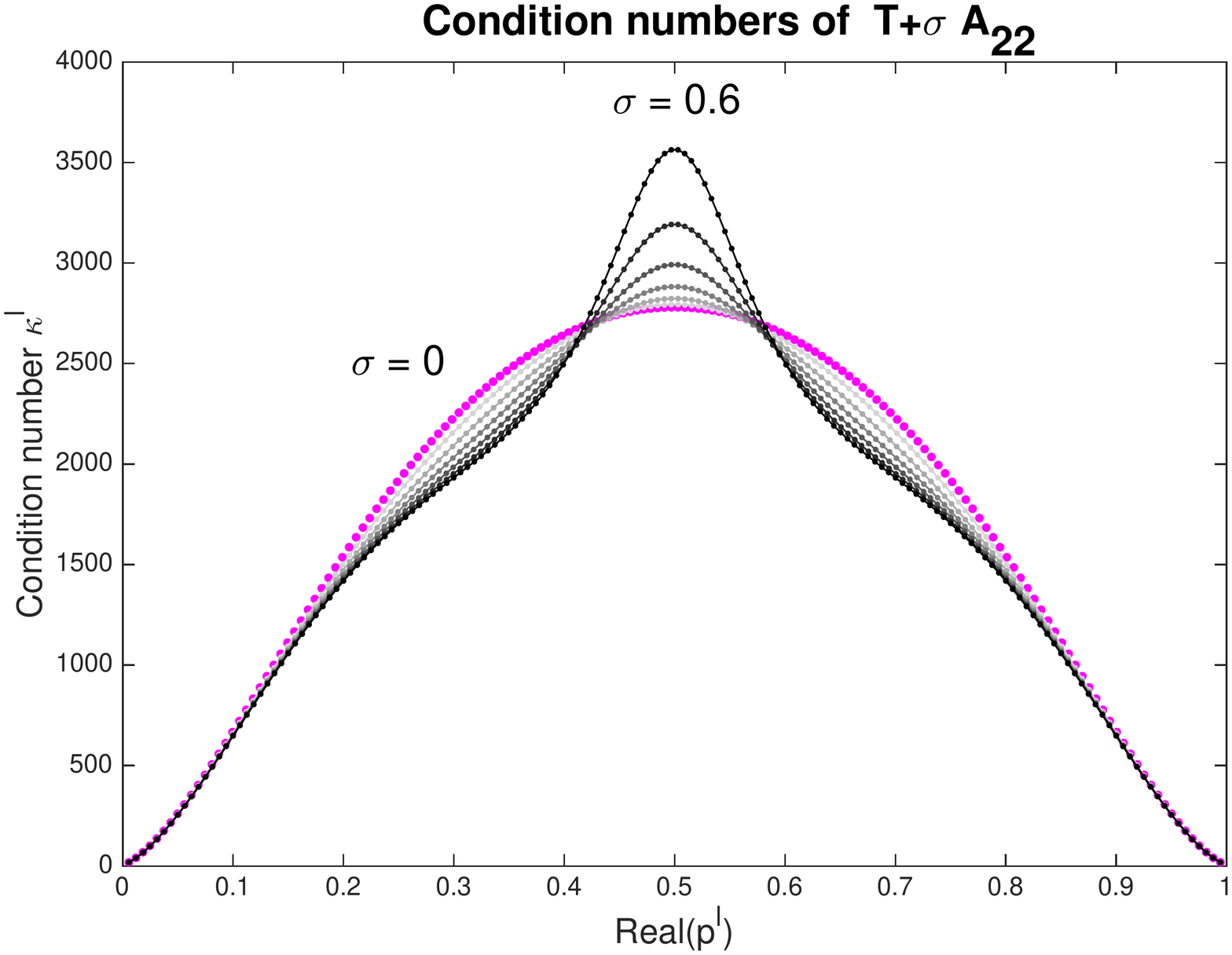}\includegraphics[scale=0.22]{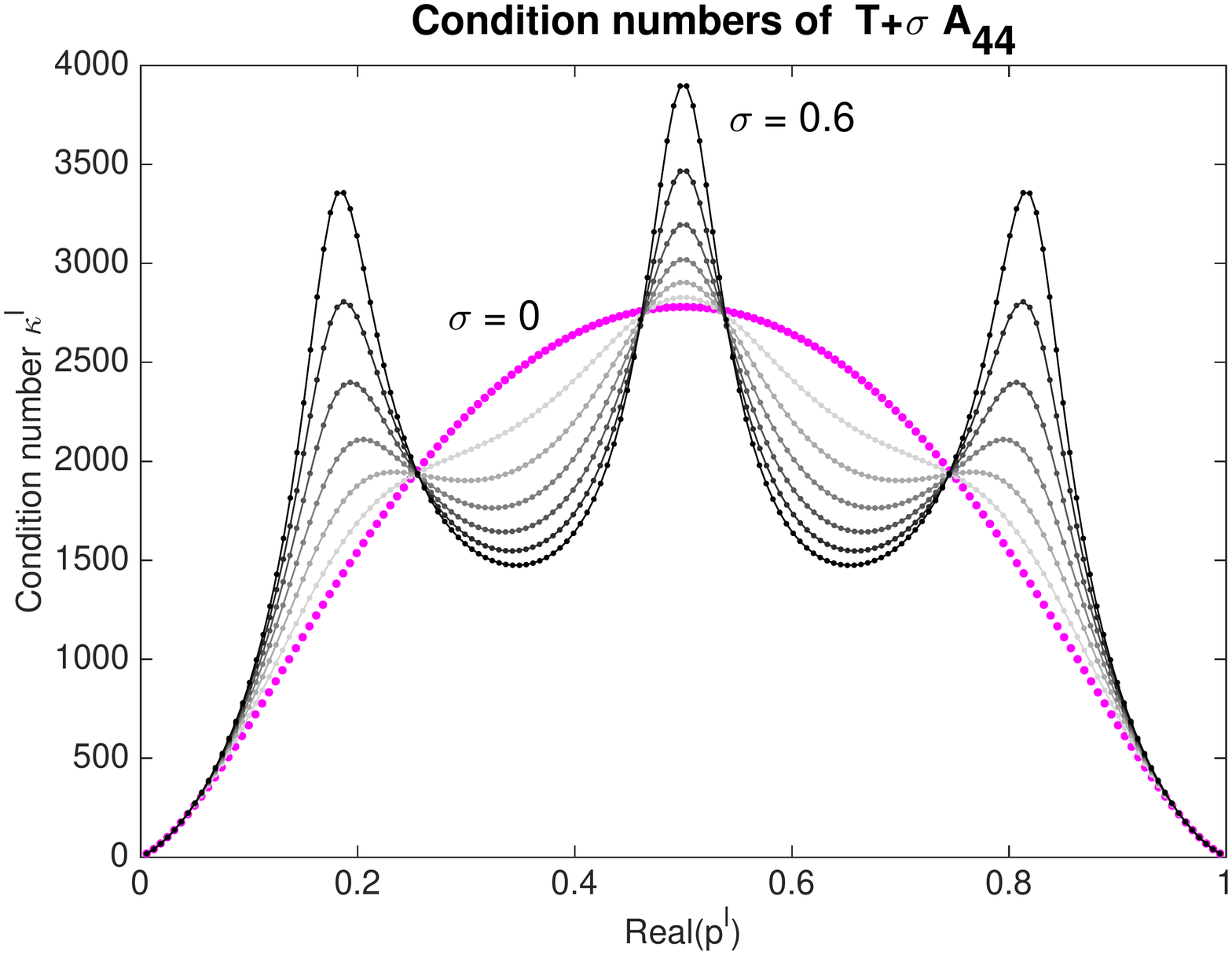}\includegraphics[scale=0.22]{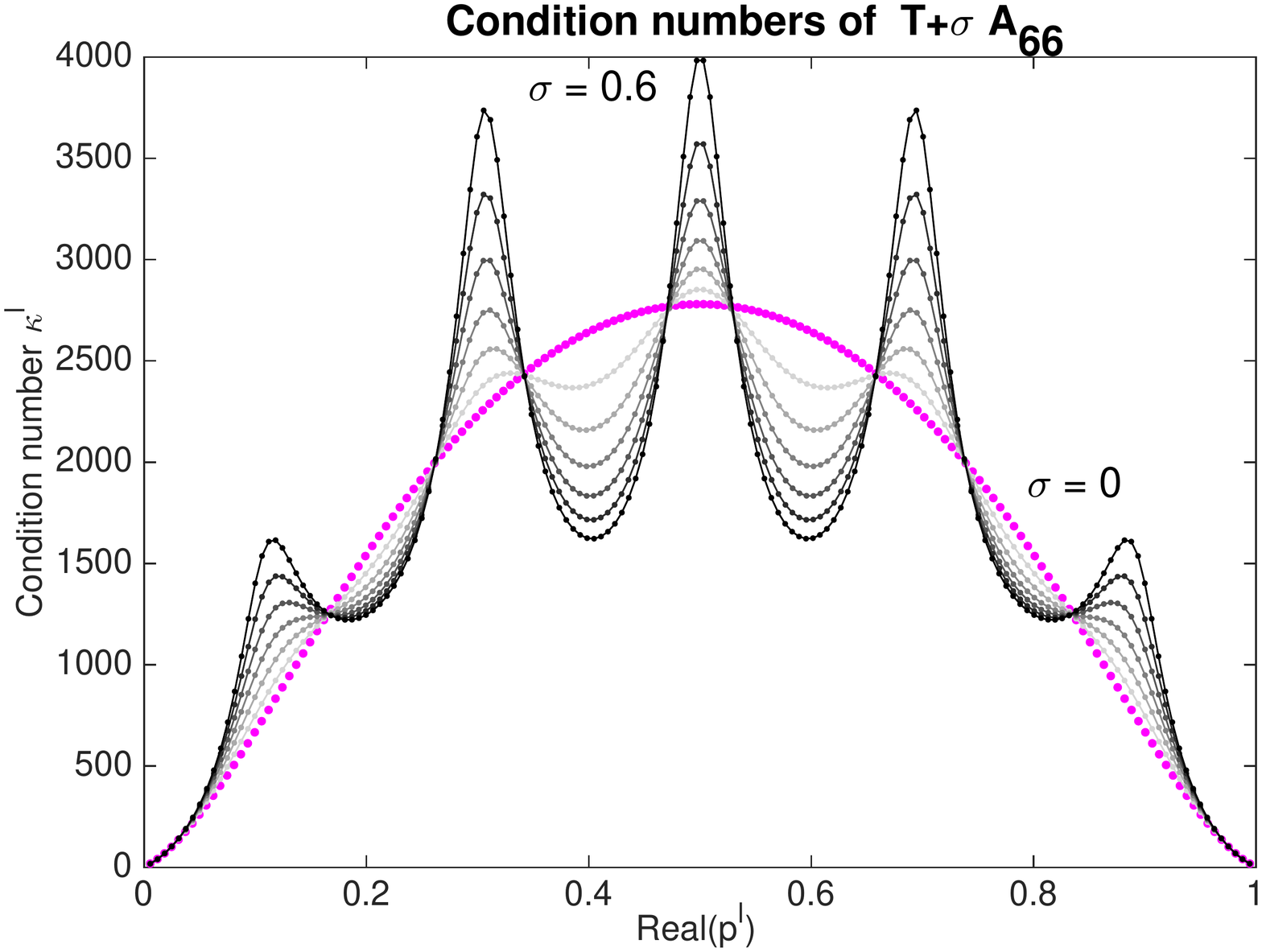}

\caption{\label{fig:Eig_Cond_E_jj_jEven}Illustration of Conjecture \ref{conj:j_Runaways}: Eigenvalues and Condition numbers
for $T+\sigma A_{jj}$ with $j$ even, where $T$ is the $n=160$ Toeplitz matrix as before.}
\end{figure}

\begin{figure}
\includegraphics[scale=0.22]{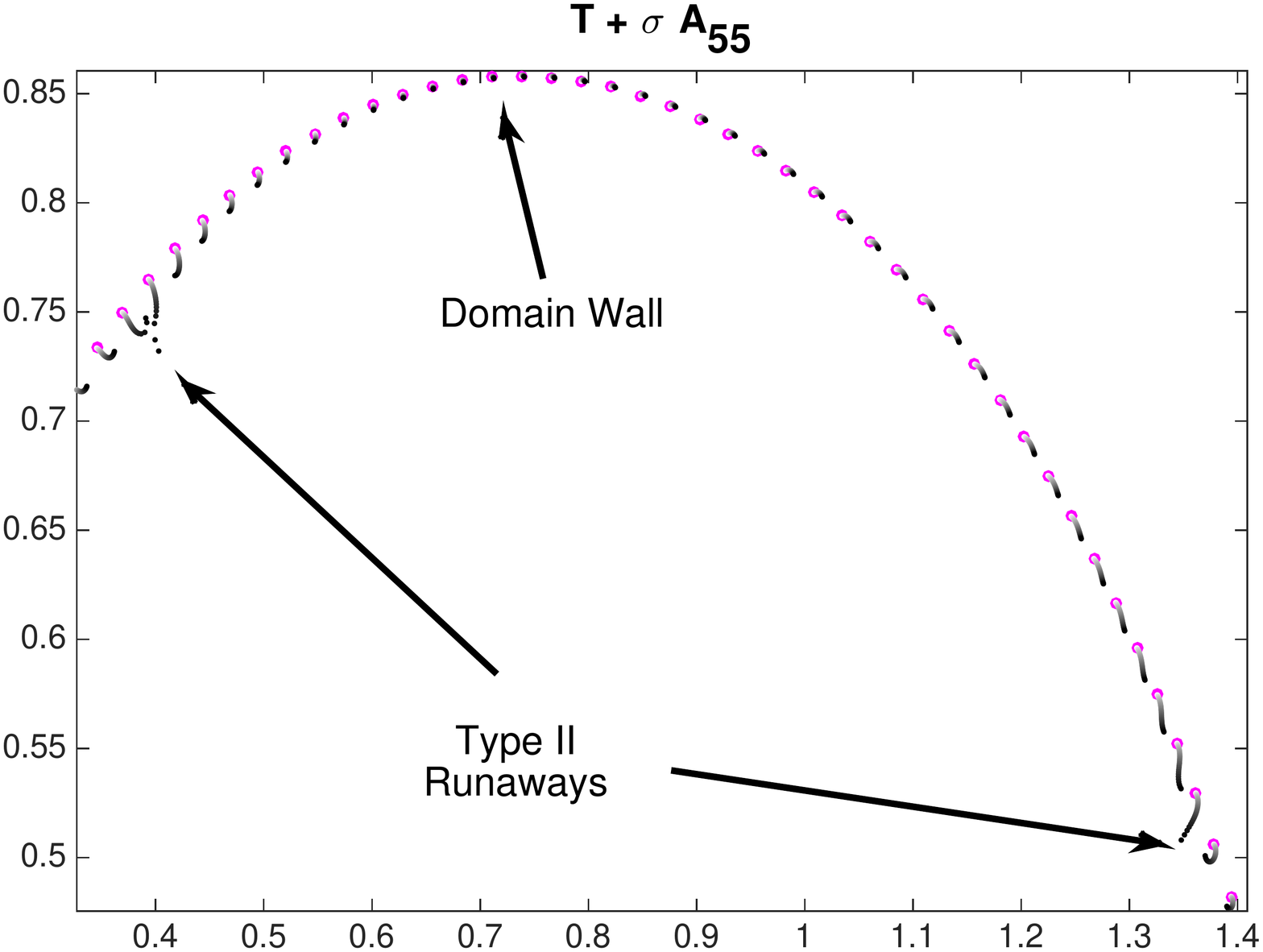}\includegraphics[scale=0.22]{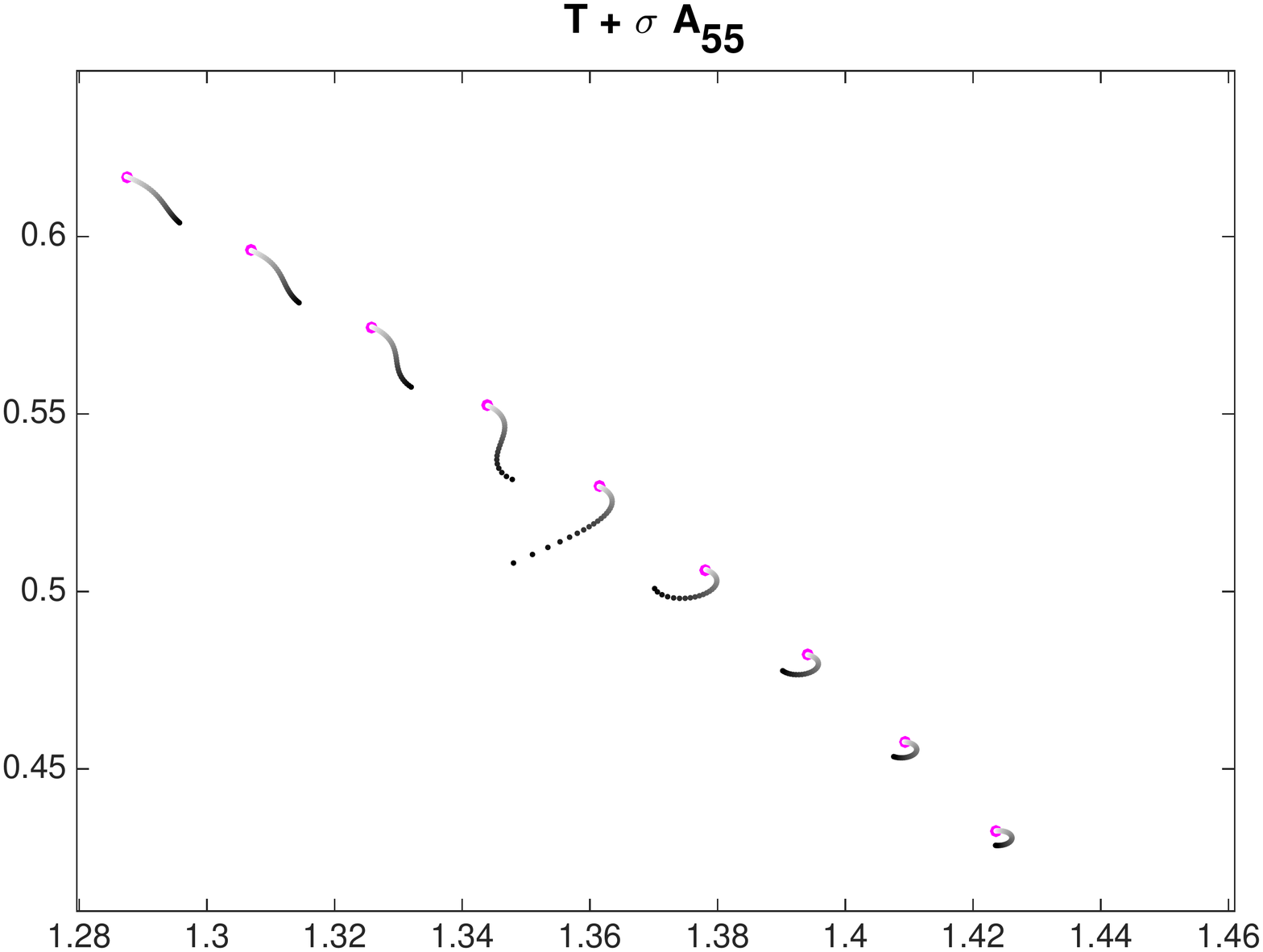}\includegraphics[scale=0.22]{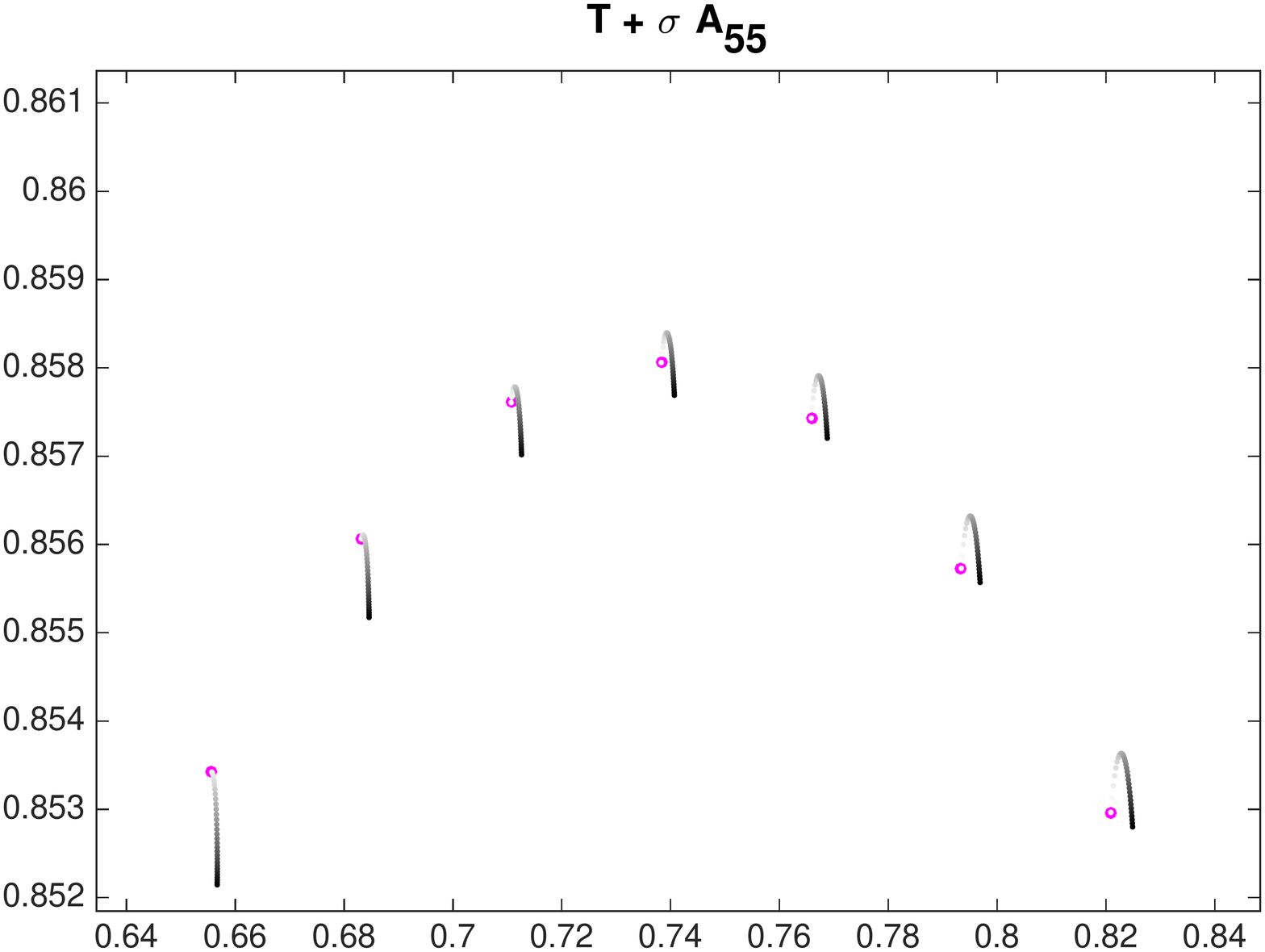}

\caption{\label{fig:Domain}We take $j=5$. Left: . Left: Motion of eigenvalues
between two runaways with a domain wall in between. Middle: We zoom
into the spectrum to show a type II runaway. Right: We show the location
of the behavior near a domain wall.}
\end{figure}

The behavior of type II runaway eigenvalues as a function of $j$
is quite interesting. We observe exactly $j$ type II Runaways moving
{\it into} the spectrum. Moreover, they all have very large condition numbers.
We show these eigenvalues of $T(\sigma)$ and the condition numbers
of the corresponding eigenvalues in Figs. \ref{fig:Eig_Cond_E_jj_jOdd}
and \ref{fig:Eig_Cond_E_jj_jEven}. 

Eigenvalues that have nearly zero imaginary velocities serve as kind
of domain walls. For example, to the left (right) of a given domain
wall the eigenvalues have positive (negative) imaginary parts. The
switching of the imaginary component implies that there must be a place
between domain walls where the imaginary downward velocity is maximum.
We observe that those are the places where type II eigenvalues are
born. The second order corrections also show increasing winding with
increasing $j$. We show examples of these  in Fig. \ref{fig:Domain}.

We now examine Eq. \eqref{eq:SokolovCondition} more closely. This equation
implies that the runaways correspond to the $k,j$ entry of the resolvent
being real $R(\lambda)_{kj}\in\mathbb{R}$. This follows from varying
$\sigma$ from $-\infty$ to $+\infty$. Above we took $\sigma\ge0$;
in  Fig. \ref{fig:The-curves-ResolventReals} we show the eigenvalues
of $T(\sigma)=T+\sigma A_{jj}$ and $T(\sigma)=T-\sigma A_{jj}$ for
$j=\{3,4,5\}$ and $\sigma\in[0,+20]$. Blue circles are the
eigenvalues in the limit $\sigma\rightarrow\infty$, where $R(\lambda)_{jj}\rightarrow0$.
Please compare these plots with with Figs. \ref{fig:Eig_Cond_E_jj_jOdd}
and \ref{fig:Eig_Cond_E_jj_jEven}. 
\begin{figure}
\centering{}\includegraphics[scale=0.29]{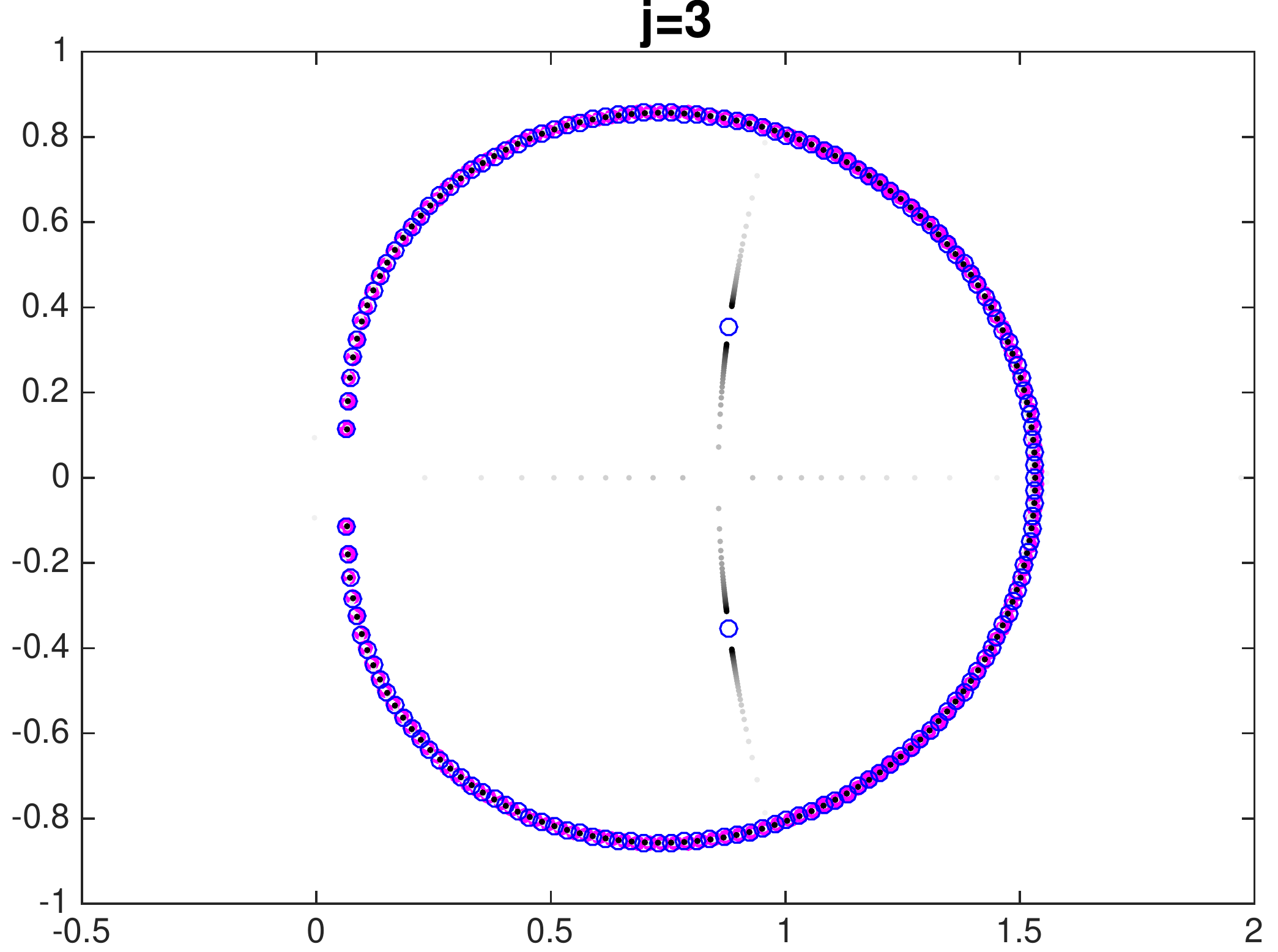}\includegraphics[scale=0.29]{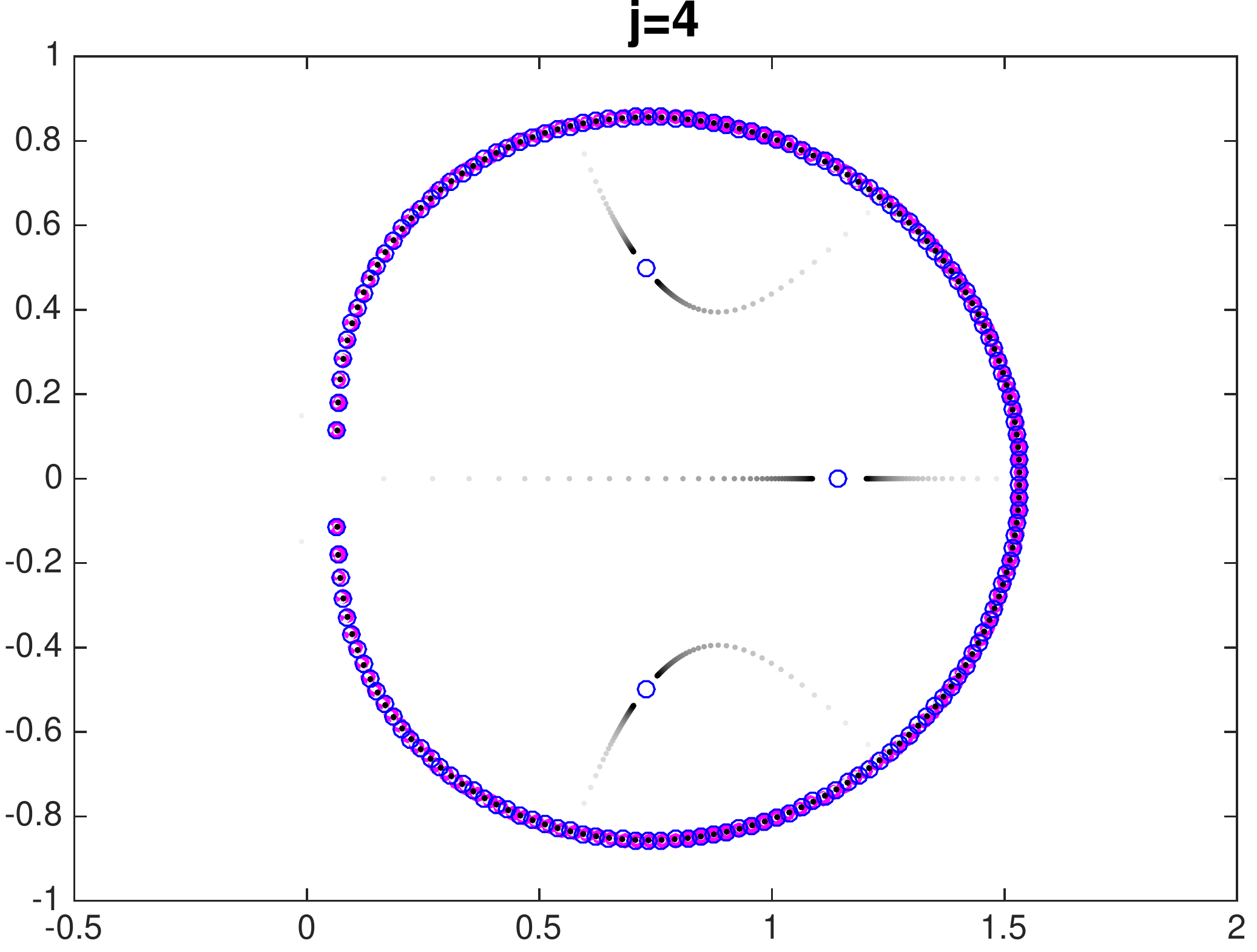}\includegraphics[scale=0.29]{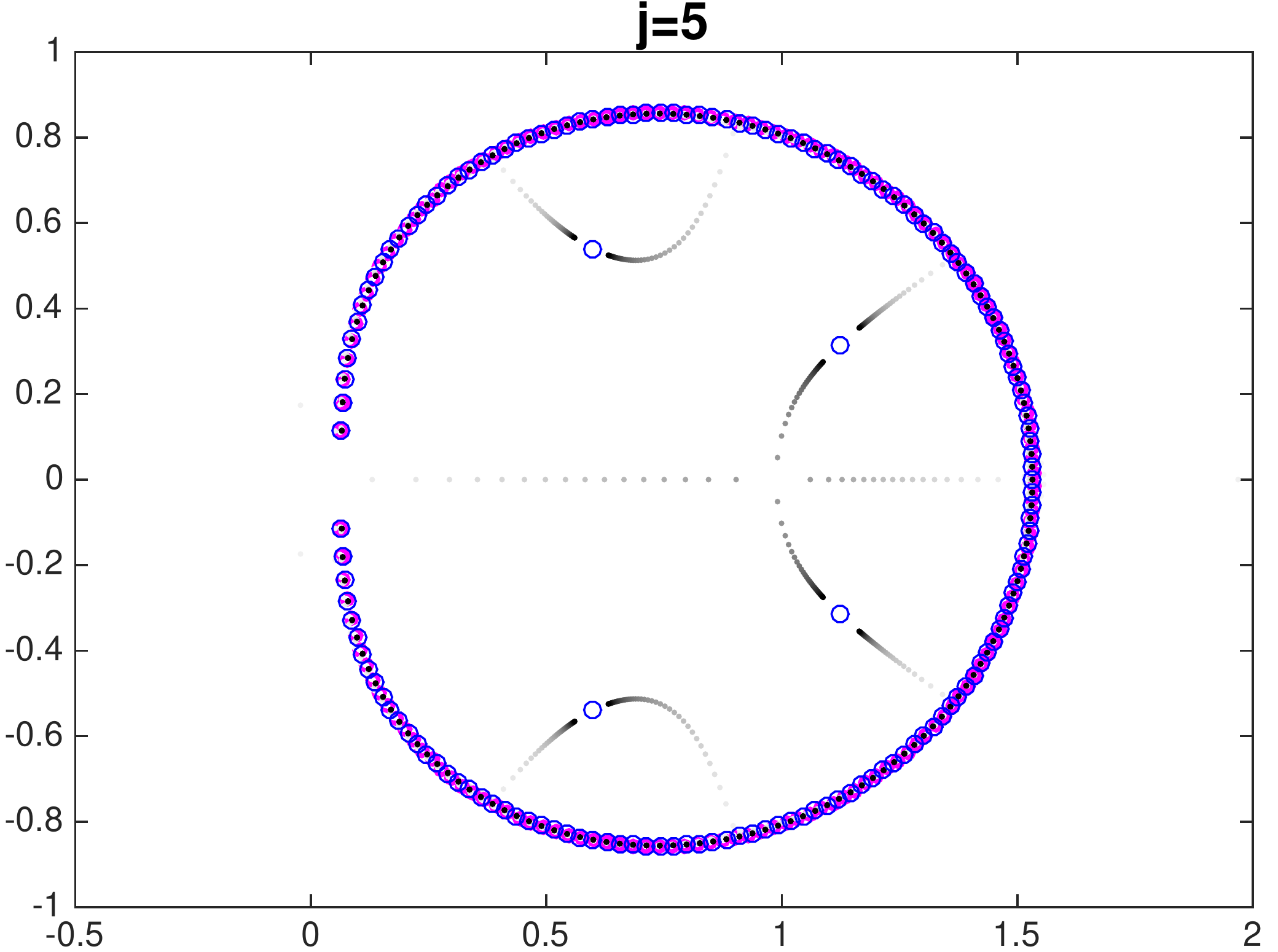}\caption{\label{fig:The-curves-ResolventReals}The curves on which $R(\lambda)_{jj}$, as given in Eq. \eqref{eq:SokolovCondition}, is real. $\sigma\rightarrow\infty$ is shown in blue circles. }
\end{figure}

\begin{conjecture}
\label{conj:Runaways_jk}The number of Runaways in $T(\sigma)=T+\sigma A_{1k}$
is $k-1$ all of which move inwards and the number of Runaways in
$T(\sigma)=T+\sigma A_{j1}$ is $j$ all of which move outwards.
\begin{figure}
\begin{centering}
\includegraphics[scale=0.4]{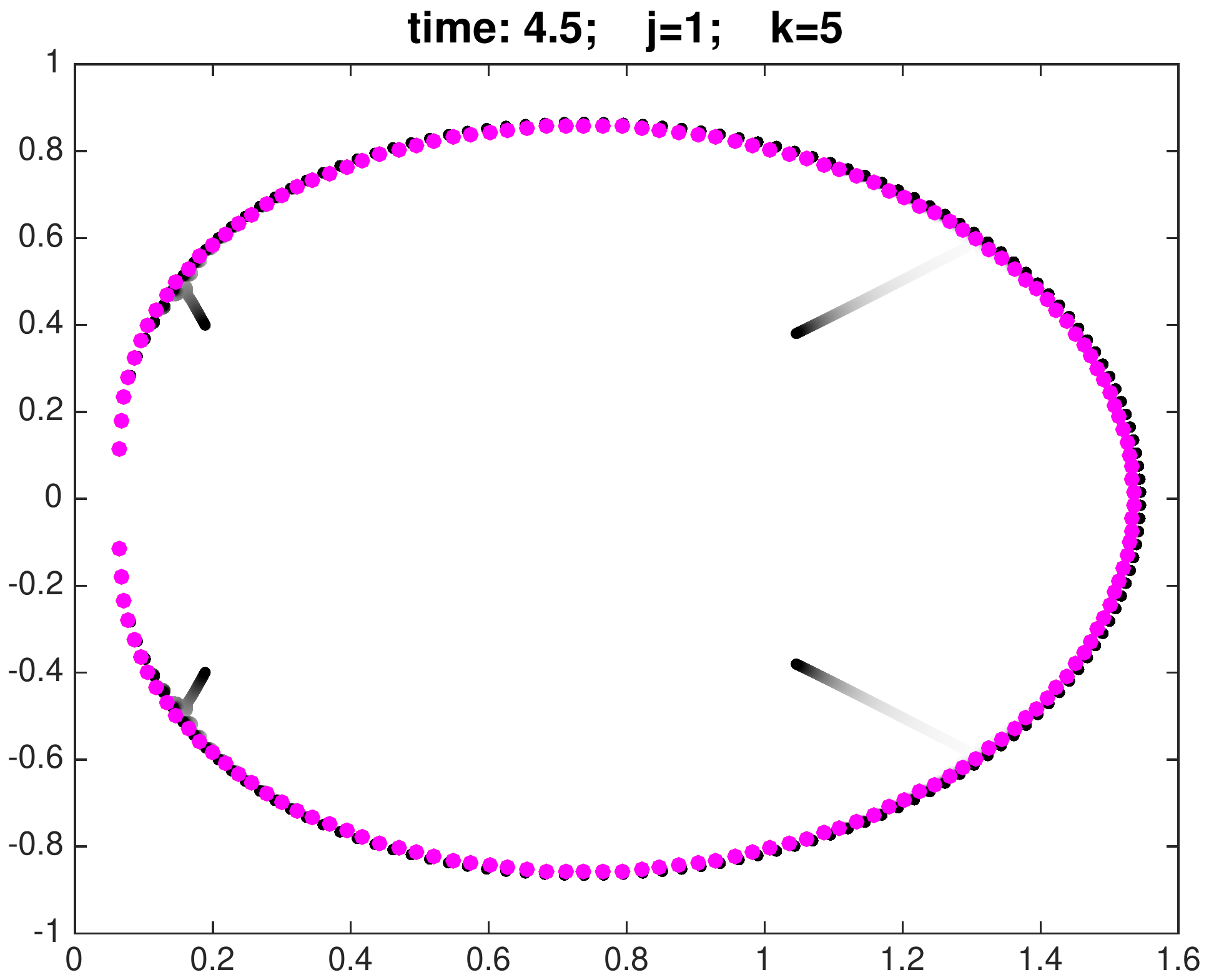}\includegraphics[scale=0.4]{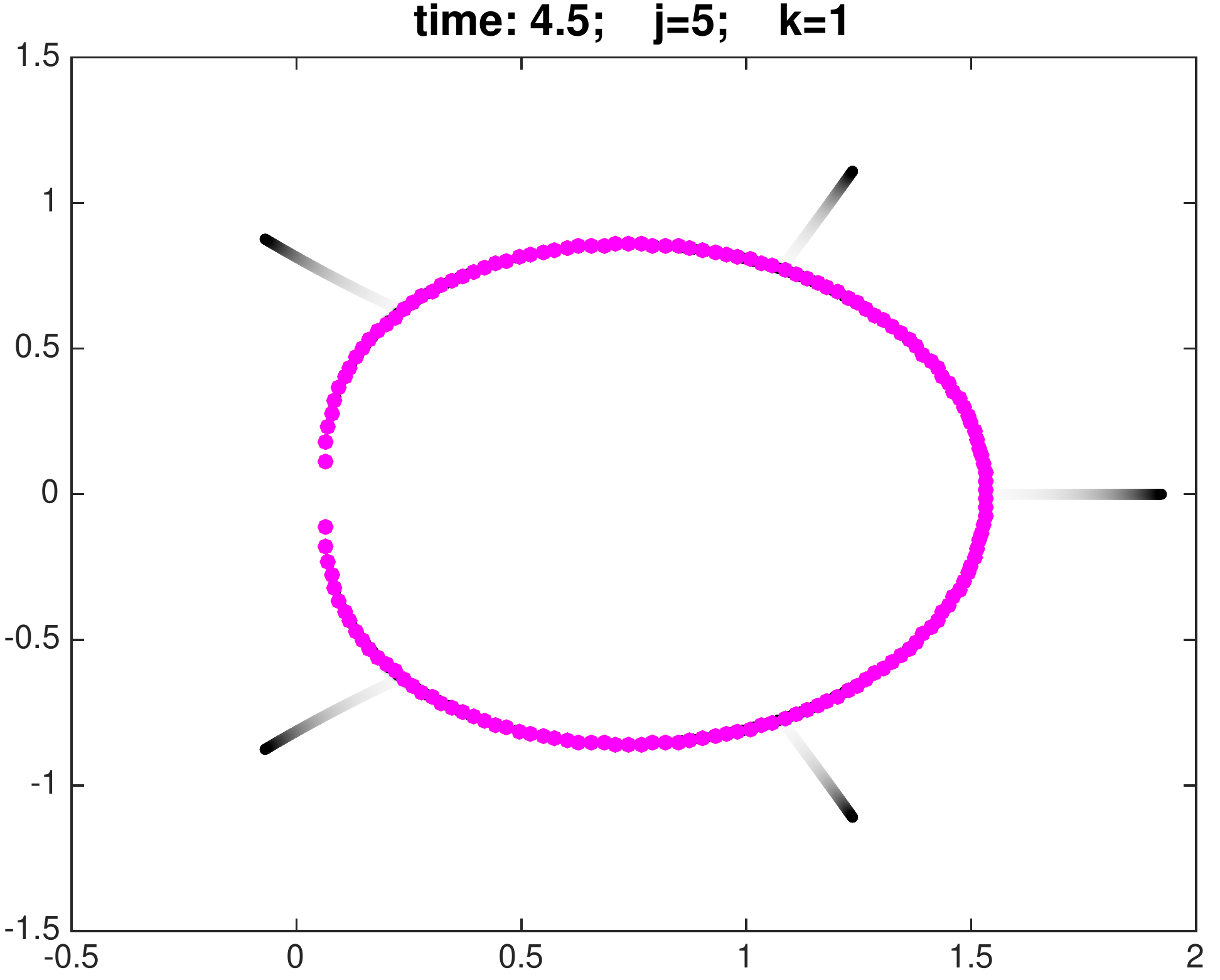}
\par\end{centering}
\caption{\label{fig:Illustration-of-Conjecture2}Illustration of Conjecture
\ref{conj:Runaways_jk}: the eigenvalues of $T$ (Eq. \eqref{eq:T_leo_Matrix}) are shown in magenta and eigenvalues of $T(\sigma)$ result from rank-$1$ perturbations of row or columns of $T$.}
\end{figure}

\end{conjecture}
We illustrate this conjecture in Figure \ref{fig:Illustration-of-Conjecture2}.
We comment that to see this for larger values of $k$ and $j$ one
needs to run the simulation for longer times (i.e., larger $\sigma$).

The pattern for general $j\ne k$ is more complex and there will be
some eigenvalues that move outwards and some inwards. We leave a thorough
investigation of general rank$-1$ perturbation for future work.

The Runaways type II, relative to the pure Toeplitz case $\sigma=0$,
have very large condition numbers. Consequently they move substantially
relative to the bulk and exhibit non-perturbative behavior. 
\begin{conjecture}
For $j$ even, the matrix $T(\sigma)=T+\sigma A_{jj}$ is defective
when Runaways type I eigenvalues collide (eigenvalue become degenerate).
When $j$ is odd or when $T(\sigma)=T+\sigma V$, the matrix has the
same geometric multiplicity as the algebraic multiplicity at the moment
of collision.
\end{conjecture}
Recall that Eqs. \eqref{eq:EigenValue_R} and \eqref{eq:EigenValue_L} are the eigenvalue equations  $T\left(\sigma\right)=T+\sigma V$.

\begin{figure}
\includegraphics[scale=0.3]{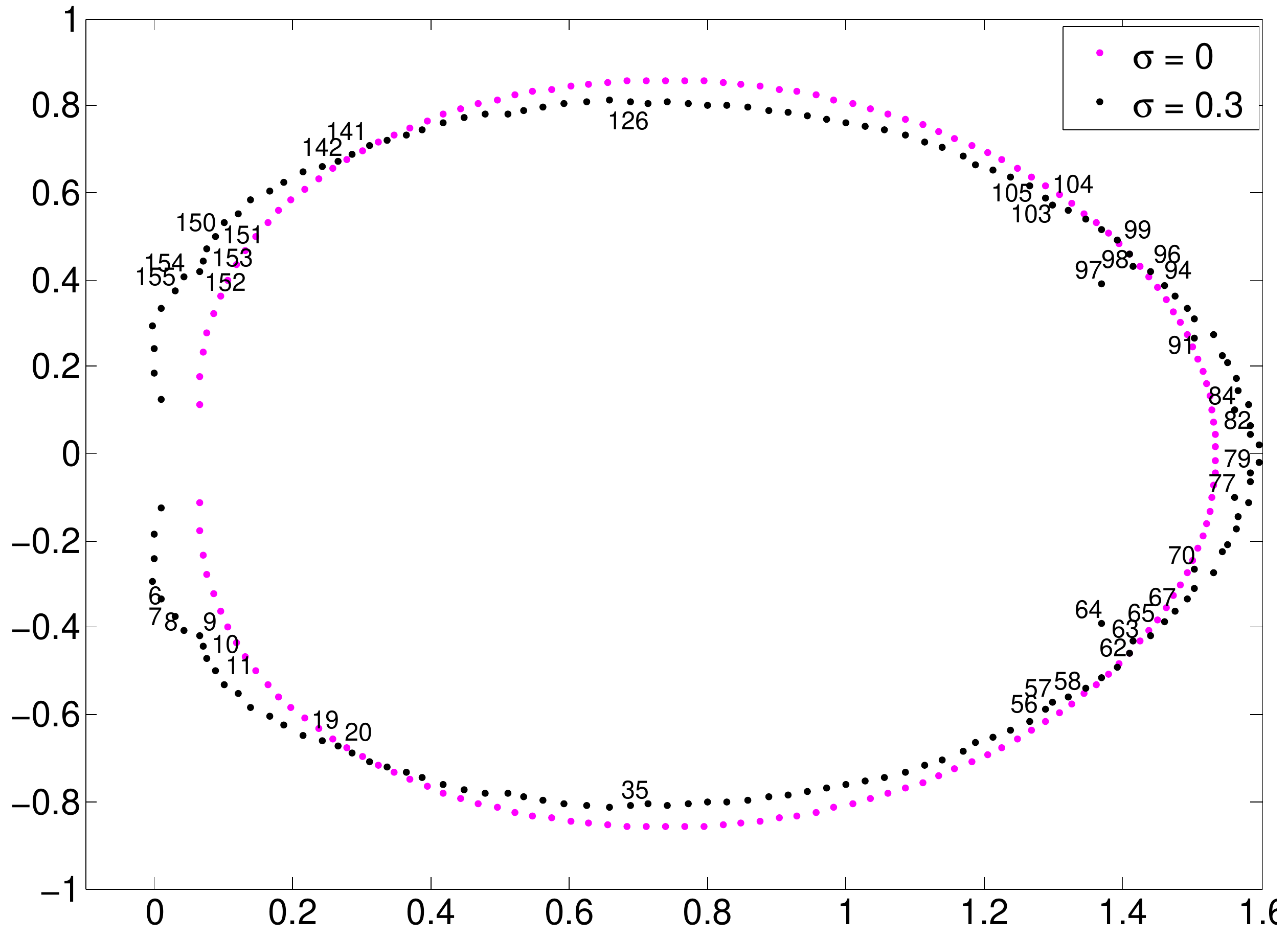}\includegraphics[scale=0.27]{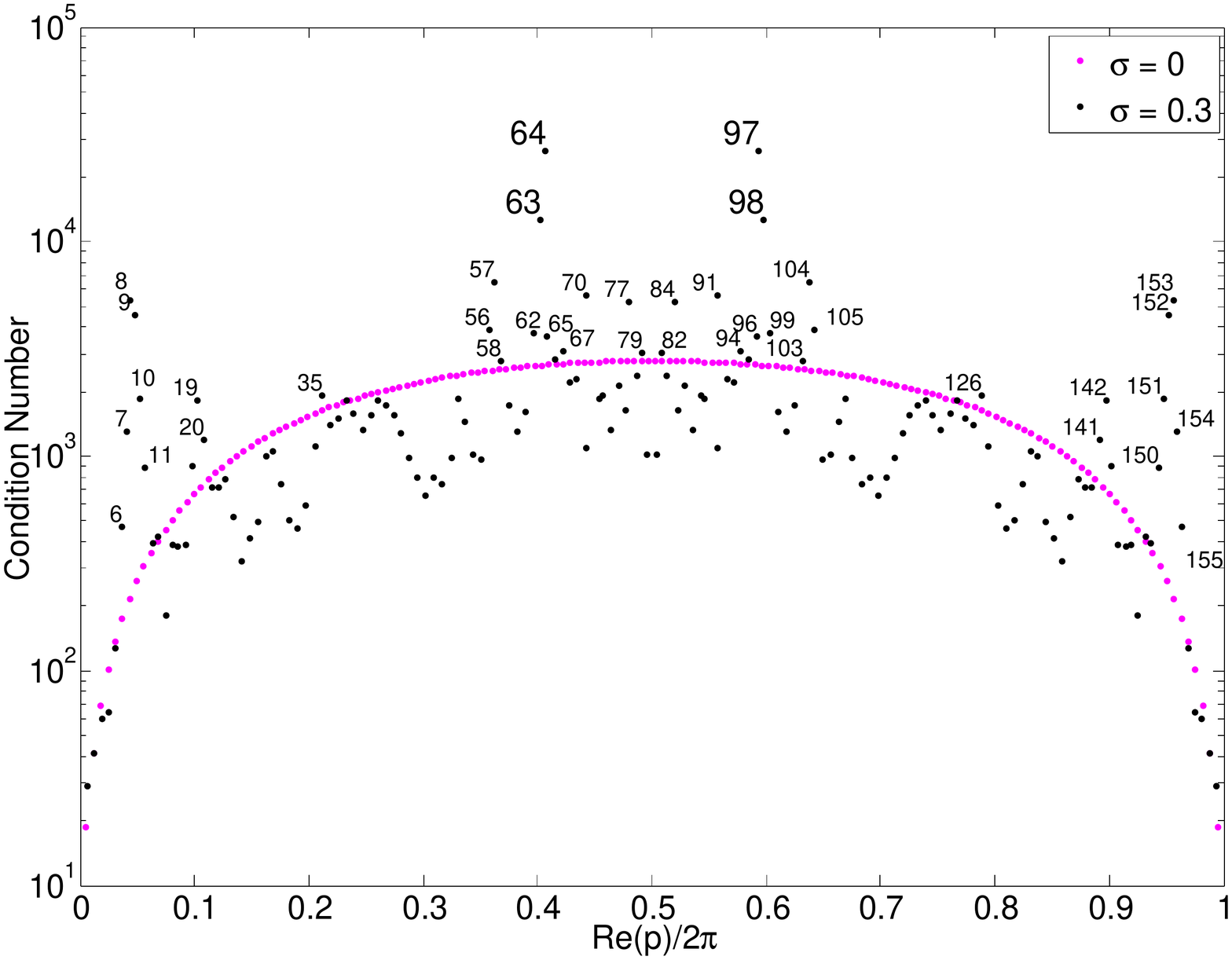}\caption{\label{fig:Runaways-type-II}Runaways type II-- Left: Labeled eigenvalues
in the complex plane. Right: the corresponding condition numbers (note that the vertical axis is in logarithmic scale).}
\end{figure}
 
\begin{figure}
\centering{}\includegraphics[scale=0.3]{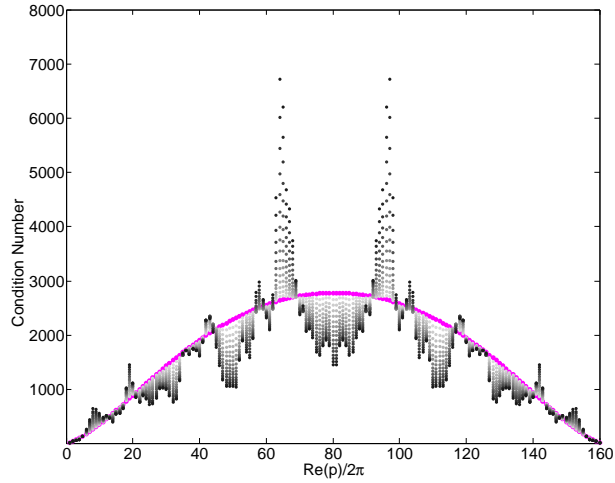}\caption{\label{fig:Condition-number-of}Condition number of the matrix ordered
by the real part of $p$ as a function of $\sigma=[0,0.01,\cdots,0.15]$.
The magenta corresponds to the condition numbers of the unperturbed
matrix and we use grey scale to show how the condition number changes
by increasing $\sigma$. The black dots show the condition numbers
with largest $\sigma$.}
\end{figure}

\begin{figure}
\includegraphics[scale=0.28]{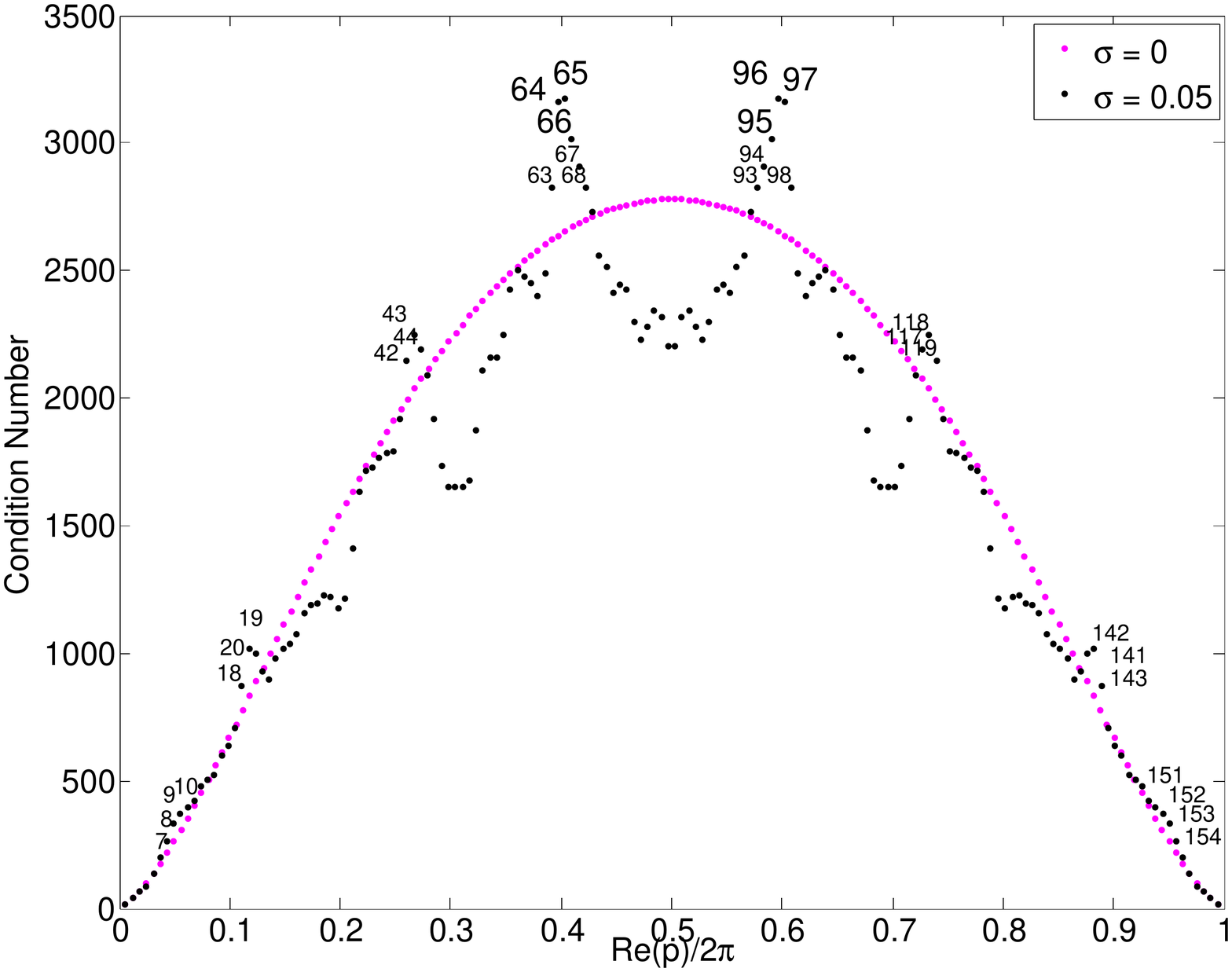}\includegraphics[scale=0.3]{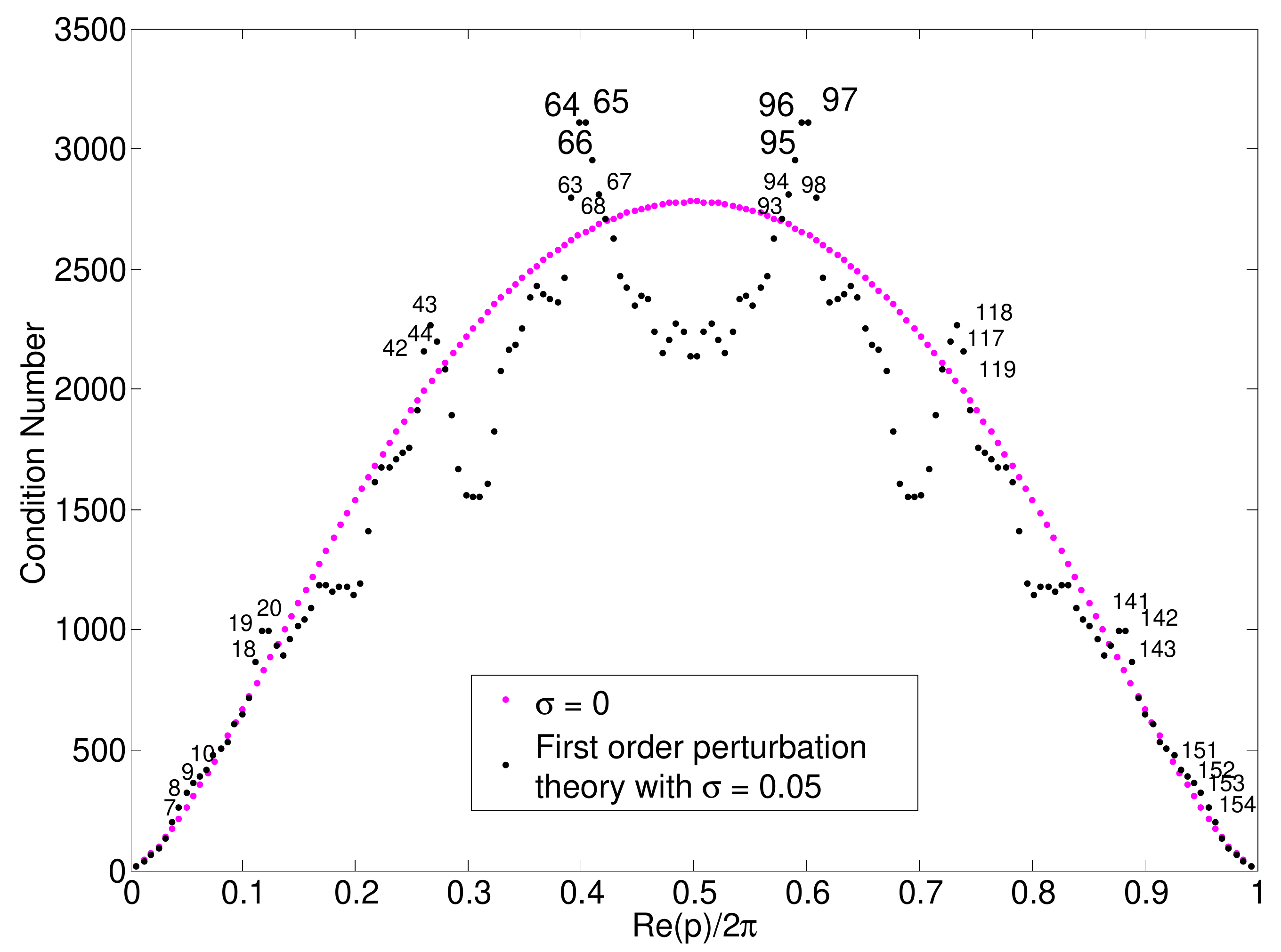}

\caption{\label{fig:RunawayII_1stOrder}Comparing the exact condition numbers
(left) and condition numbers calculated from perturbation theory (right).
In both cases we take $\sigma=0.05\ll1$ to show the predictive power
of first order perturbation theory for type II runaways before non-perturbative
behavior become visually evident. }
\end{figure}

The perturbation expansion of eigenpairs is given by Eqs. \eqref{eq:E_perturbExpand} and \eqref{eq:Evec_perturbExpand}. Multiplying $|\psi^\ell_\sigma\rangle$ in Eq. \eqref{eq:Evec_perturbExpand} on the left by the right eigenvector $\langle\tilde{\psi}_0^\ell |$ we get
\begin{eqnarray*}
E^{\ell}\left(\sigma\right)\langle\tilde{\psi}_{0}^{\ell}|\psi_{\sigma}^{\ell}\rangle & = & \langle\tilde{\psi}_{0}^{\ell}|T|\psi_{\sigma}^{\ell}\rangle+\sigma\langle\tilde{\psi}_{0}^{\ell}|V|\psi_{\sigma}^{\ell}\rangle=E_{0}^{\ell}\langle\tilde{\psi}_{0}^{\ell}|\psi_{\sigma}^{\ell}\rangle+\sigma\langle\tilde{\psi}_{0}^{\ell}|V|\psi_{\sigma}^{\ell}\rangle,
\end{eqnarray*}
since corrections to eigenvectors in Eq. \eqref{eq:Evec_perturbExpand} are all orthogonal
to $\langle\tilde{\psi}_{0}|$, we have $\langle\tilde{\psi}_{0}^{\ell}|\psi_{\sigma}^{\ell}\rangle=\langle\tilde{\psi}_{0}^{\ell}|\psi_{0}^{\ell}\rangle$.
Thus the perturbation $\sigma V$ displaces the eigenvalue $E^{\ell}\left(\sigma\right)$
by
\begin{eqnarray}
E^{\ell}\left(\sigma\right)-E_{0}^{\ell} & = & \frac{\sigma\langle\tilde{\psi}_{0}^{\ell}|V|\psi_{\sigma}^{\ell}\rangle}{\langle\tilde{\psi}_{0}^{\ell}|\psi_{0}^{\ell}\rangle}\le\frac{\left|\sigma\right|\left\Vert V\right\Vert \left\Vert \langle\tilde{\psi}_{0}^{\ell}|\right\Vert \mbox{ }\left\Vert |\psi_{\sigma}^{\ell}\rangle\right\Vert }{|\langle\tilde{\psi}_{0}^{\ell}|\psi_{0}^{\ell}\rangle|}\le\frac{\left|\sigma\right|\left\Vert V\right\Vert \left\Vert \langle\tilde{\psi}_{0}^{\ell}|\right\Vert }{|\langle\tilde{\psi}_{0}^{\ell}|\psi_{0}^{\ell}\rangle|}\equiv\left|\sigma\right|\left\Vert V\right\Vert \kappa\left(E_{0}^{\ell}\right),\label{eq:delE}
\end{eqnarray}
where $\kappa\left(E_{0}^{\ell}\right)$ is the condition number of $\left(E_{0}^{\ell}\right)$.

In order to theoretically predict the type II runaway eigenvalues,
we use $\sigma\ll1$ and first order perturbation theory on eigenstates
to calculate the condition number and compare it with the exact result.
Let the first order approximation to the state be $|\psi_{pert}^{\ell}\rangle\equiv|\psi_{0}^{\ell}\rangle+\sigma\mbox{ }|\psi_{1}^{\ell}\rangle$
(the subscript $pert$ denotes perturbation theory), which reads 
\begin{eqnarray*}
|\psi_{pert}^{\ell}\rangle & = & |\psi_{0}^{\ell}\rangle+\sigma\sum_{j\ne\ell}\frac{\langle\tilde{\psi}_{0}^{j}|V|\psi_{0}^{\ell}\rangle}{E_{0}^{\ell}-E_{0}^{j}}\mbox{ }|\psi_{0}^{j}\rangle+\mathcal{O}\left(\sigma^{2}\left\Vert V\right\Vert ^{2}\right)\\
\langle\tilde{\psi}^\ell_{pert}| & = & \langle\tilde{\psi}_{0}^{\ell}|+\sigma\sum_{j\ne\ell}\frac{\langle\tilde{\psi}_{0}^{\ell}|V|\psi_{0}^{j}\rangle}{E_{0}^{\ell}-E_{0}^{j}}\mbox{ }\langle\tilde{\psi}_{0}^{j}|+\mathcal{O}\left(\sigma^{2}\left\Vert V\right\Vert ^{2}\right)
\end{eqnarray*}
Using these we calculated
\begin{eqnarray*}
\frac{1}{|\cos\theta_{\sigma}^{\ell}|} & = & \left\Vert \langle\tilde{\psi}_{\sigma}^{\ell}|\right\Vert \qquad\mbox{ Exact}\\
\frac{1}{|\cos\theta_{pert}^{\ell}|} & = & \left\Vert \langle\tilde{\psi}_{pert}^{\ell}|\right\Vert \qquad\mbox{ Perturbation theory}
\end{eqnarray*}

In Fig. \ref{fig:RunawayII_1stOrder} one can see that first order
perturbation theory very accurately predicts runaway type II eigenvalues
(i.e., ill conditioned). See Fig. \ref{fig:Condition-number-of} to
see how the condition number changes with $\sigma$. 

Looking at Eq. \eqref{eq:attraction}, we see that the eigenvalues must
move into the bulk as the complex conjugates attract strongly. However,
here the strength of attraction is due to large $||\langle\tilde{\psi}_{pert}^{\ell}|\mbox{ }||$
appearing in the numerator. Therefore, ill-conditioning, combined
with the attraction result above predicts the runaways type II behavior. 
\begin{rem*}
The runaway type II are nonperturbative, i.e., cannot be captured
by perturbation theory as shown above. However, the onset of non-perturbative
behavior can be predicted using first order perturbation theory. We showed this by using very small $\sigma$ ($\sigma=0.05$ in Fig. \ref{fig:RunawayII_1stOrder}), and
predicted the exponential growth of the condition number for such
eigenvalues. 
\end{rem*}

\subsubsection{Eigenvalues from Free Probability Theory}

In this section we show that modern free probability theory is a successful tool in approximating the eigenvalue distribution
or density of states (DOS).  Free probability theory (FPT) is tailored for capturing the DOS of
the sum of matrices that are in \textit{generic} positions \cite{Speicher}.

In standard (i.e., classical) probability theory, DOS of the sum of random variables is the
convolution of their individually known distributions. The notion
is extended to commuting matrices, where there exists a basis that simultaneously diagonalizes the matrices. The joint density
is obtained by a convolution of individual densities. 

Suppose we are interested in the eigenvalue distribution of $A+B$
and that $A$ and $B$ are matrices with known eigenvalue distributions
$\rho_{A}$ and $\rho_{B}$. If the matrices commute $\left[A,B\right]=0$, then we can
work in a basis where both matrices are diagonal and 
\[
\rho_{A+B}=\rho_{A}\star\rho_{B}
\]
where we denote the convolution by $\star$. The requirement of simultaneously
diagonalizability is very stringent for matrices, especially when
they are random. Generic matrices are in a sense the extreme opposite
of commuting matrices. However, a modern notion of\textit{ free convolution}
has been developed that allows one to compute the DOS of the sum of
random matrices.  

FPT provides the exact distribution of the sum when the size of the matrices go to infinity and when they are fully generic.
That is, if we find a basis that diagonalizes $A$, then the eigenvectors
of $B$ in that basis have a Haar measure over the symmetric group
(see \cite{RamisMovassaghIE} for more details). The DOS of the sum
is given by their \textit{free} convolution \cite{Speicher}
\[
\rho_{A+B}=\rho_{A}\boxplus\rho_{B},
\]
where we denote the free convolution by $\boxplus$.

At the first sight, it may seem like the requirement of genericity
is also very stringent and that we are left with another very special
point like the commuting case where classical probability theory applies.
However, we have come to realize that the DOS of disordered systems
are often well captured by either FPT \cite{RamisMovassagh_Anderson}
or, in more complicated settings, by a one-parameter linear combination
of the classical and free probability theory \cite{RamisMovassaghIE,RamisMovassaghIE_PRL}.
This provides an exciting new opportunity for scientists to make quantitative
progress in understanding the DOS of interesting disordered physical
systems. 
\begin{figure}
\includegraphics[scale=0.3]{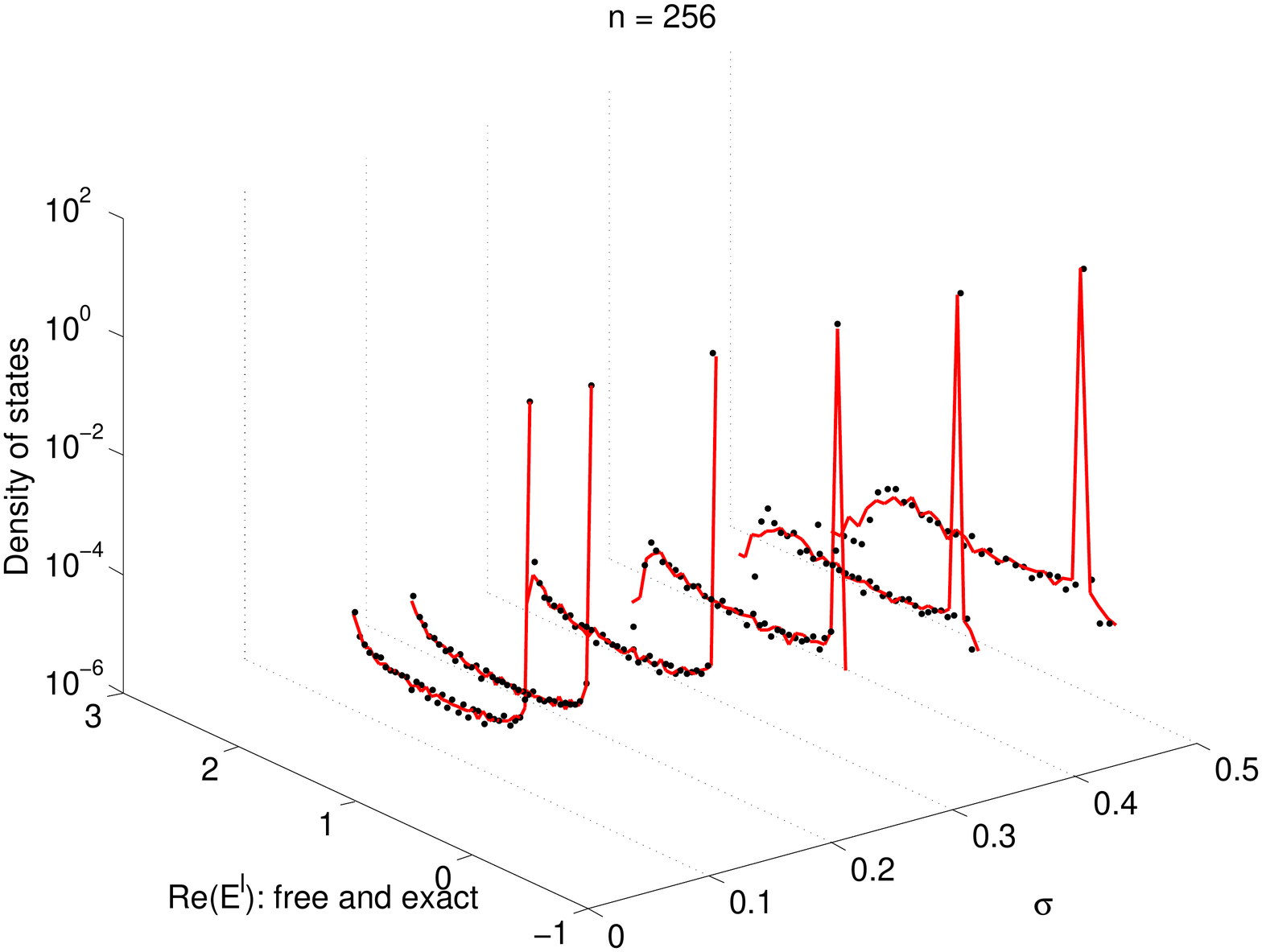}$\quad$\includegraphics[scale=0.3]{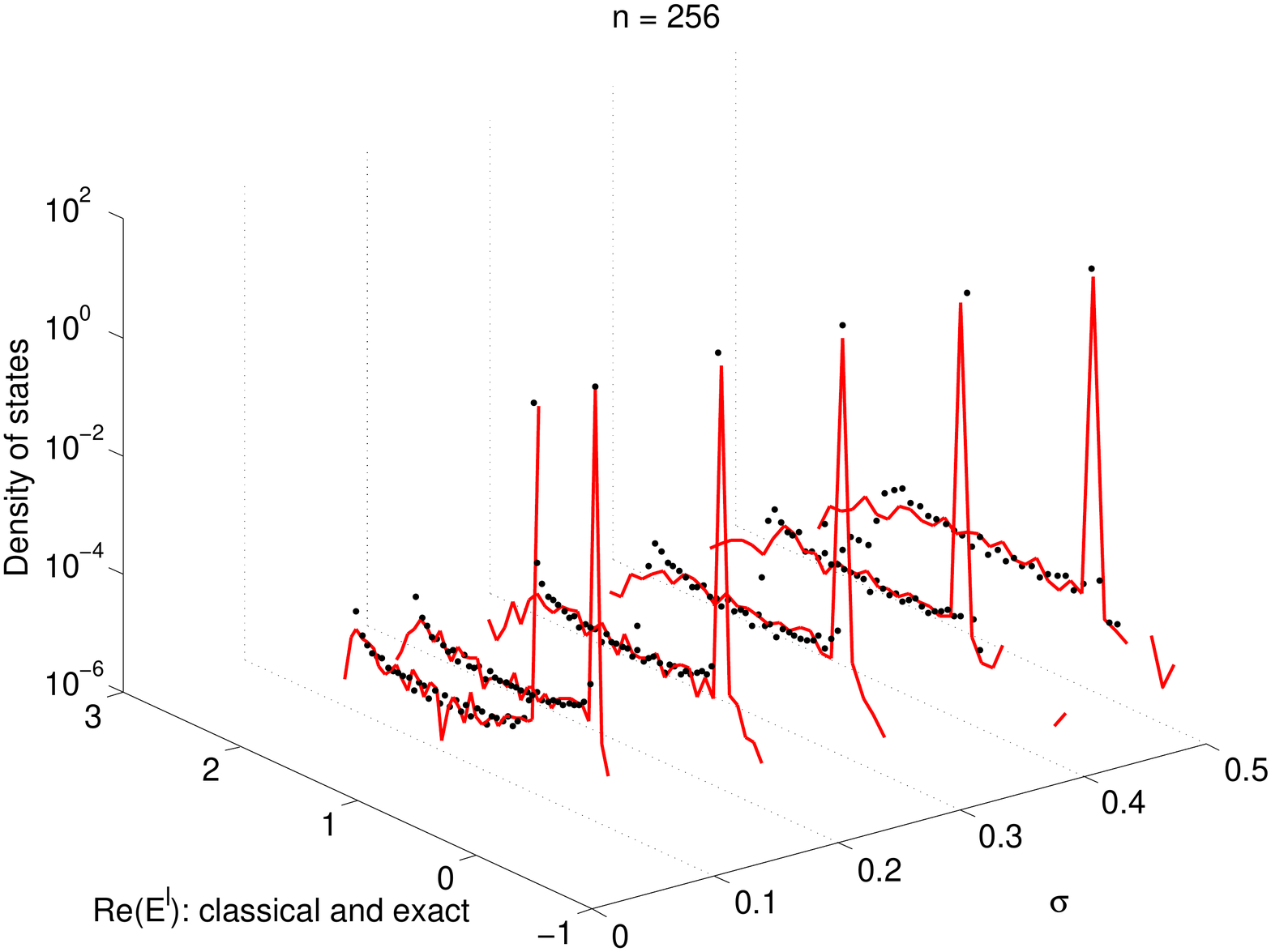}
\includegraphics[scale=0.3]{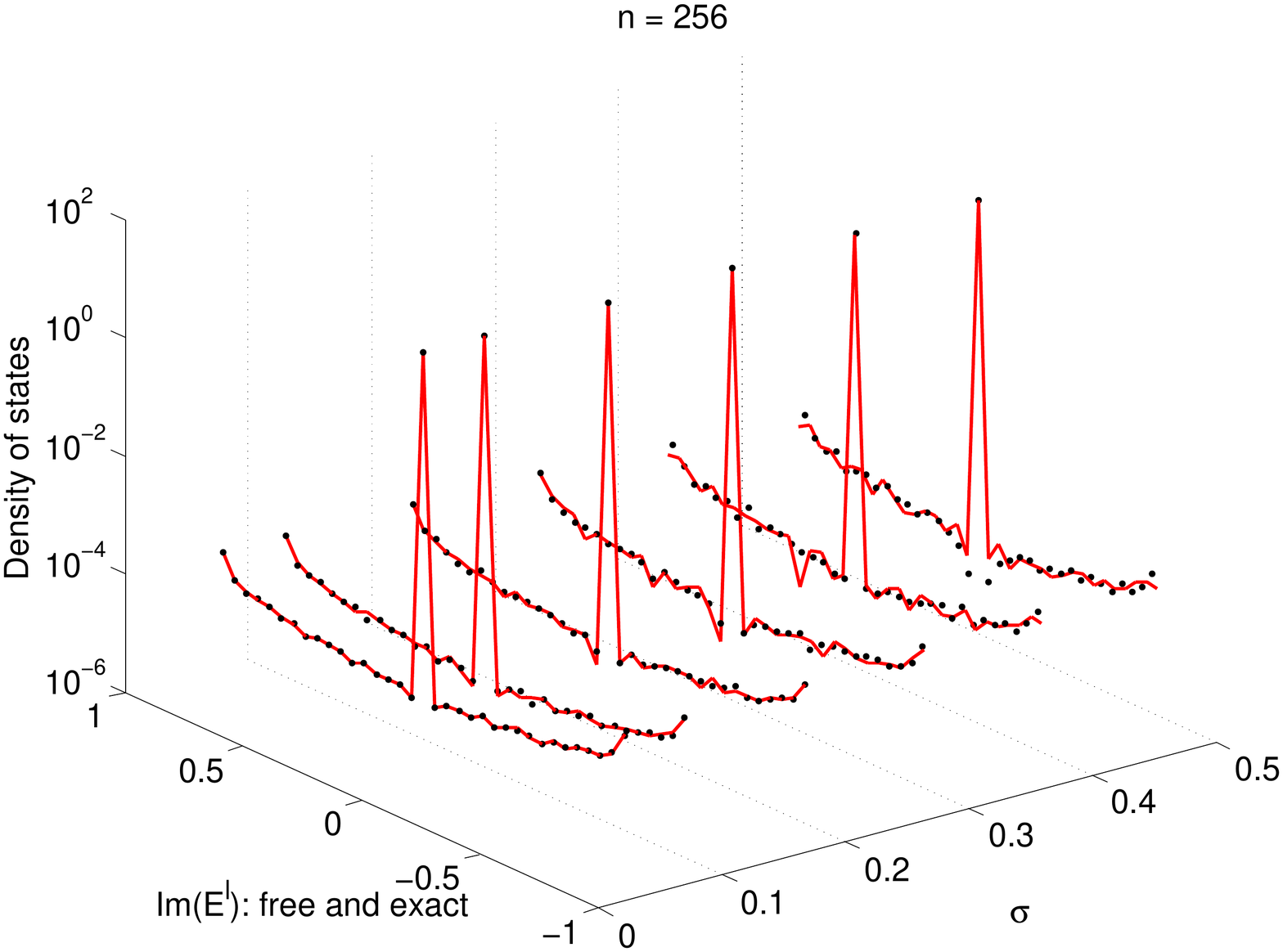}$\quad$\includegraphics[scale=0.3]{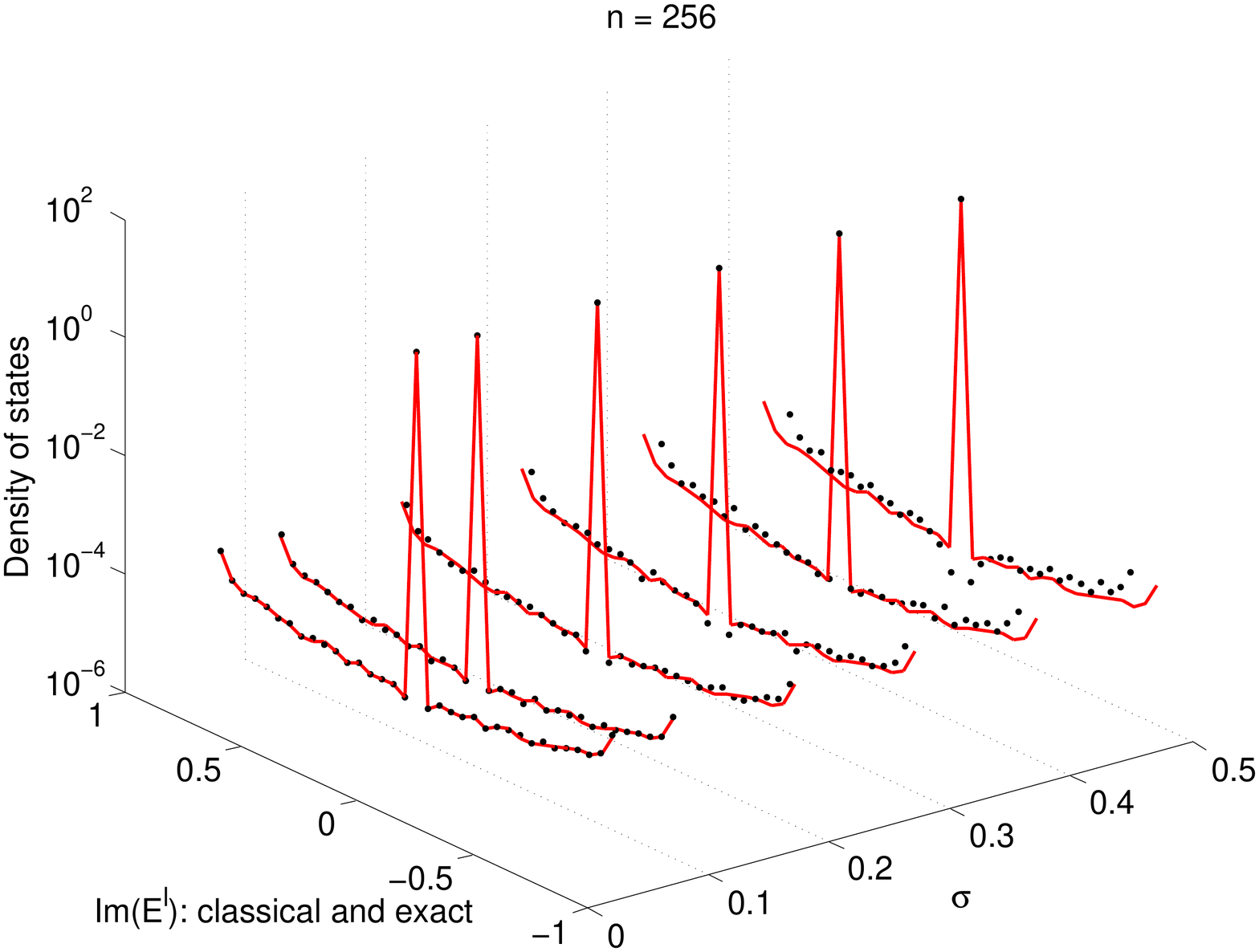}
\caption{\label{fig:FPT}The actual eigenvalue DOS (based on $E^\ell_\sigma$) are shown in black dots. The red curves are the free approximation $Q^{-1}\Lambda_{T}Q+\sigma V$, where $Q$ is a Haar matrix.
The horizontal axes are the eigenvalues.}
\end{figure}

Previously we established that the density of states of the Anderson
model \cite{Anderson58} is well described by FPT \cite{RamisMovassagh_Anderson}.
We could prove that if one writes the Hamiltonian as a sum of its
hopping part plus the diagonal random matrix with gaussian entries,
then the two are provably free up to their first $8$ moments. 

The Toeplitz problem is also translationally invariant yet is markedly
different as the matrix is not normal and perturbations can cause
drastic changes in the spectrum landscape. In this section we show
that free probability theory nevertheless captures the eigenvalue
distribution. We like to capture the eigenvalue distribution of $T+\sigma V$
to high accuracy from the knowledge of eigenvalue of $T$ and distribution of $V$ alone. Since
$T$ has $n$ simple eigenvalues, it is not defective and has an eigenvalue
decomposition
\begin{eqnarray*}
T & = & Q_{T}^{-1}\Lambda_{T}Q_{T}
\end{eqnarray*}
where $\Lambda_{T}$ is the diagonal matrix of the eigenvalues of
the Toeplitz matrix and $Q_{T}$ the matrix of its eigenvectors. The
exact problem, $T\left(\sigma\right)$, whose DOS we seek can be written
as
\[
T+\sigma V=Q_{T}^{-1}\Lambda_{T}Q_{T}+\sigma V .
\]

There are two noteworthy deformations of the exact problem, namely
the free and classical approximations 
\begin{eqnarray*}
Q_{T}^{-1}\Lambda_{T}Q_{T}+\sigma V & \qquad & \mbox{exact}\\
Q^{-1}\Lambda_{T}Q+\sigma V & \qquad & \mbox{free approximation}\\
\Pi^{-1}\Lambda_{T}\Pi+\sigma V & \qquad & \mbox{classical approximation}
\end{eqnarray*}
where $Q$ is an $n\times n$ random Haar orthogonal matrix, $\Pi$
is an $n\times n$ random permutation matrix. In other words, $V$
has eigenvectors that are equal to the standard basis, but $Q$ is
fully random with respect to this structure. For any realization of
$V$, in the free approximation, $T$ has equal probability of having
any set of eigenvectors, represented by a point, Haar distributed,
on the symmetric group. 

In Fig. \ref{fig:FPT}, we compare the results of exact diagonalization
of $T\left(\sigma\right)$ with classical and free approximation.
Note that the vertical axis is log-scaled. One can see that the real
part of the eigenvalues is much better captured by FPT, whereas the
classical fails starting from moderately small $\sigma$ and gets
worst with increasing $\sigma$. In particular, the classical approximation
of the $\Re\left(E^{\ell}\right)$ does not reach the top of the eigenvalue
atom at zero. 

Interestingly enough the imaginary part of the eigenvalues is well
captured by both methods for small $\sigma$; however as $\sigma$
increases the FPT captures the spectrum adequately but classical approximation
becomes inaccurate.

Comment: FPT does not require $V$ to be small; it is a \textit{non-perturbative
technique} for adding matrices. An important take away message is
that the relative structure of the eigenvectors of the two pieces,
i.e., $T$ and $V$, is unimportant. Namely, one can assume that one
has no particular structure relative to the other.

\section{Eigenvectors in presence of disorder $\sigma >0$}\label{sec:EigenvectorsDis}
\subsection{Localization: Entropies and Inverse Participation Ratios (IPR) of
the states}

Since Anderson's seminal work \cite{Anderson58}, the study of localization
of states of physical models such as metal insulator transitions \cite{Mott1965},
have been central in condensed matter theory. In the Hatano-Nelson
model, the eigenstates belonging to the ``wings'' of the spectrum
\cite[Section 31]{HatanoNelson1997,TrefethenEmbree2005} are known
to be localized. Here we quantify the localization of the states corresponding
to  the three class of eigenvalues discussed above, i.e., bulk, type
I and type II runaways and find some surprising new features. 

We use two methods of quantification of localization. First is \textit{entropy},
which is borrowed from information theory and the second is \textit{inverse
participation ratio (IPR)} which is a technique in  condensed matter
and statistical physics. 

Since (right) eigenvectors are all normalized we have $\sum_{j}|\psi_{j}^{\ell}|^{2}=1$ for all $\ell$. We can formally
consider $\left\{ |\psi_{1}^{\ell}|^{2},|\psi_{2}^{\ell}|^{2},\cdots,|\psi_{n}^{\ell}|^{2}\right\} $
as a discrete probability distribution of size $n$ where the probabilities
are $|\psi_{j}^{\ell}|^{2}$. A measure of uniformity versus locality
of the eigenstate is the Shannon entropy \cite{Cover1991} 
\begin{eqnarray}
H\{|\psi^{\ell}|^{2}\}=-\sum_{j=1}^{n}|\psi_{j}^{\ell}|^{2}\log_{2}|\psi_{j}^{\ell}|^{2}\qquad bits.\label{eq:entropy}
\end{eqnarray}

Comment: Entropy of any given wave-function is an increasing function
of delocalization-- it is maximum for most extended states and zero
for a delta function. 

Secondly, since $\sum_{j}|\psi_{j}^{\ell}|^{2}=1$, one can use the
fourth power to quantify localization. To this end we define the \textit{Inverse
Participation Ratio (IPR)} of an eigenstate by
\begin{figure}
\includegraphics[scale=0.3]{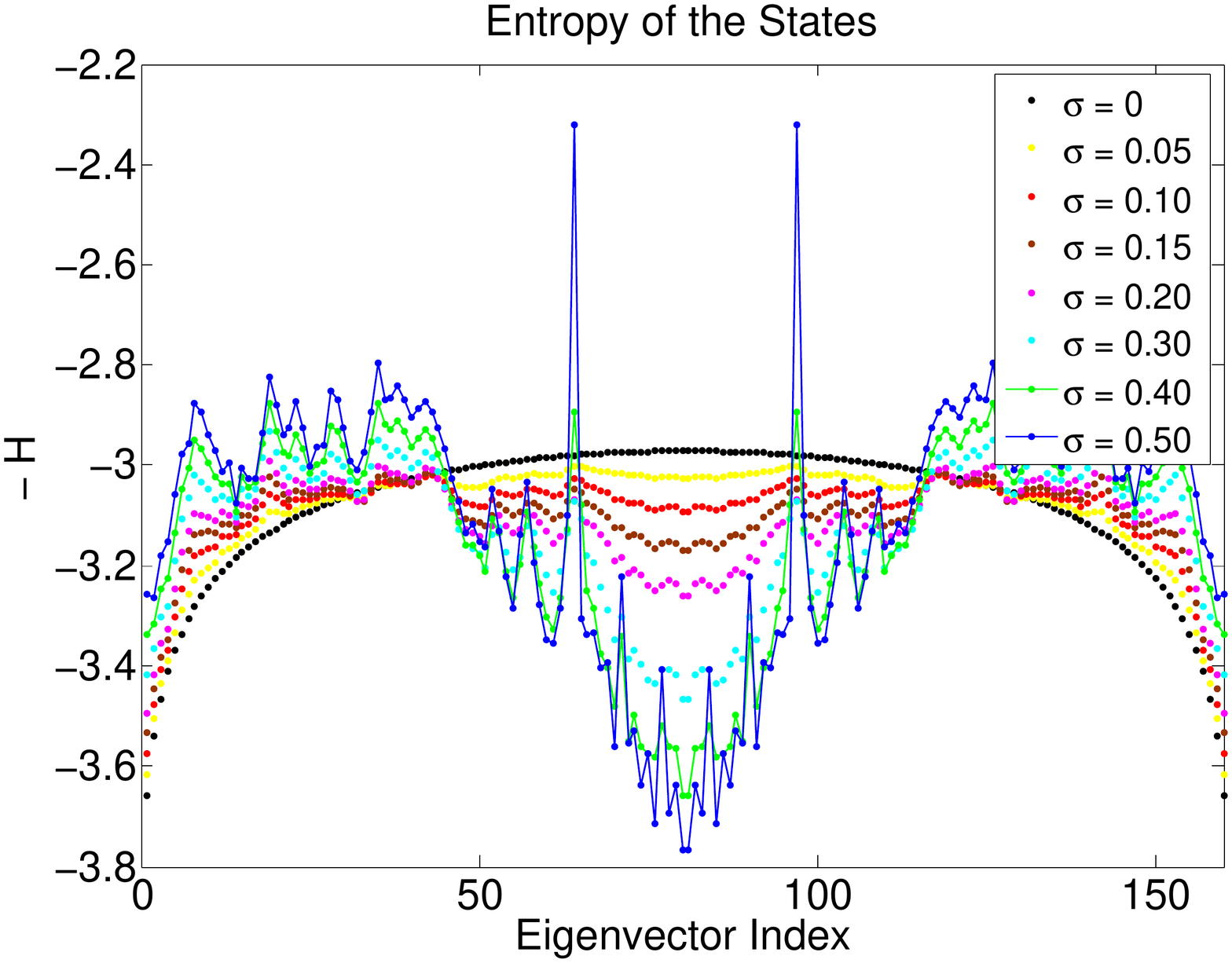}\includegraphics[scale=0.3]{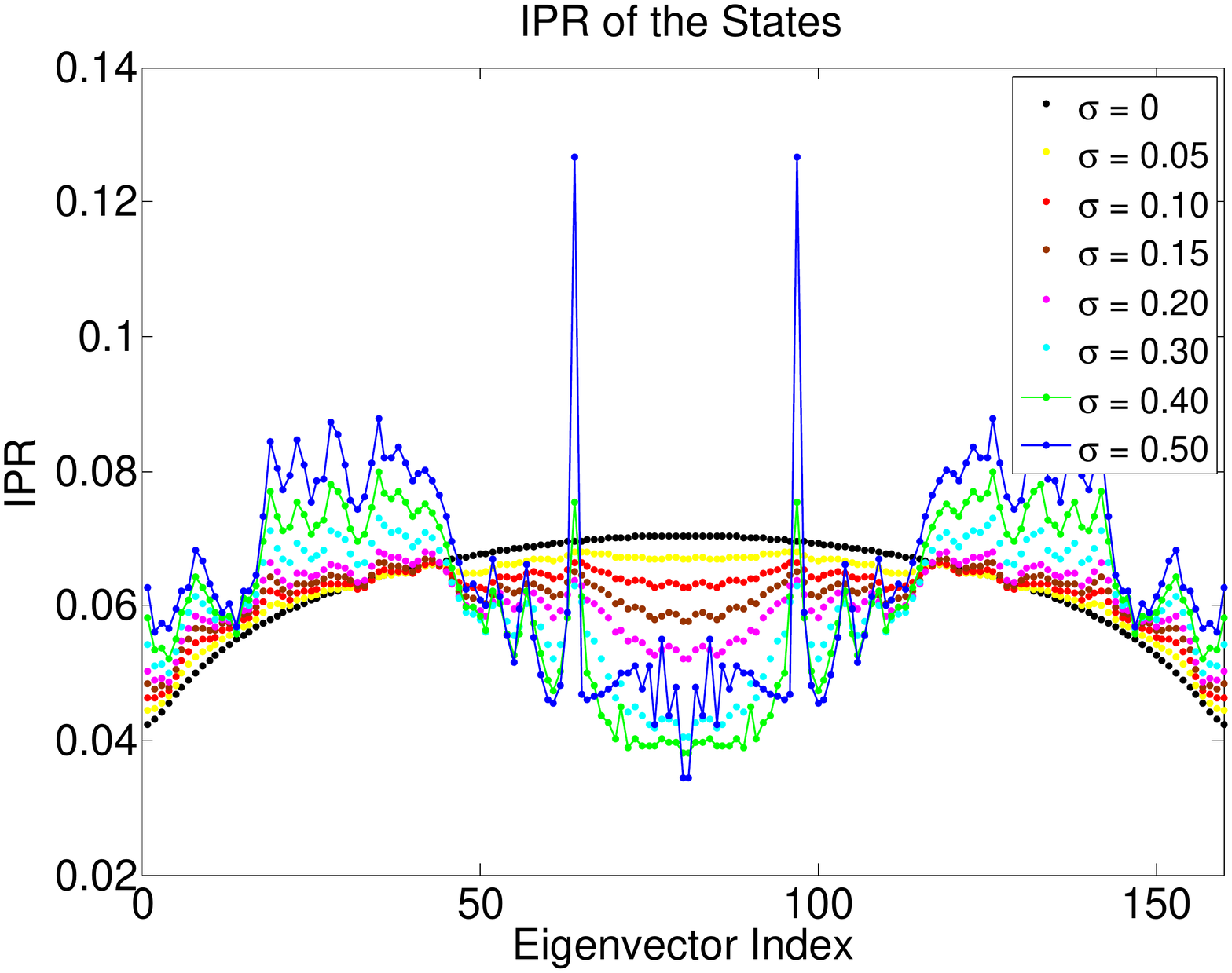}
\caption{\label{fig:The-localization-of}The localization of states: the horizontal
axis is in one-to-one correspondence with the labeling of the eigenvalues
according to $\mbox{Re}(p^{\ell})$. Left: Minus the entropy function  given by Eq. \eqref{eq:entropy}, i.e., $-H$. Right: IPR as given by Eq. \eqref{IPR}.}
\end{figure}
\begin{eqnarray}
IPR\left(\psi^{\ell}\right)=\sum_{j=1}^{n}|\psi_{j}^{\ell}|^{4}.\label{IPR}
\end{eqnarray}

Comment: In contrast with the entropy $IPR$ attains its maximum value
for the most localized and minimum for the least localized states.
To see this, suppose that a state is localized on the site $k$, then
$|\psi_{j}^{\ell}|^{2}=\delta_{jk}$. Whereas for a state that has
a uniform spread over all sites $|\psi_{j}^{\ell}|^{2}=\frac{1}{n}$
for all $j$. Therefore, 
\begin{eqnarray*}
IPR\left(\psi^{\ell}\right) & = & \sum_{j=1}^{n}\delta_{j,k}=1\qquad\mbox{most localized}\\
IPR\left(\psi^{\ell}\right) & = & \sum_{j=1}^{n}\frac{1}{n^{2}}=\frac{1}{n}\qquad\mbox{most delocalized},
\end{eqnarray*}
the latter in the thermodynamical limit, $n\rightarrow\infty$, vanishes.

In Fig. \ref{fig:The-localization-of} we plot the IPR and, to put
it on par with IPR, the negative of the entropy vs. the $\ell=1,\dots,n$.
\begin{figure}
\begin{centering}
\includegraphics[scale=0.3]{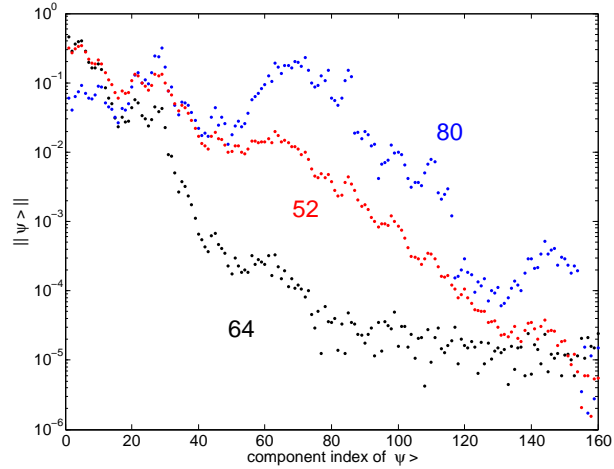}\caption{\label{fig:LocalizationColor}Localization behavior of eigenvectors:
Semi-log plot of the three different types of eigenvectors corresponding to eigenvalues shown in Fig. \ref{fig:Runaways-type-I:} and \ref{fig:Runaways-type-II}. Taking
archetypical examples to illustrate; 64 is a runaway type II, 52 is
a bulk and 80 a runaway type I states. }
\par\end{centering}
\end{figure}

\begin{figure}
\includegraphics[scale=0.22]{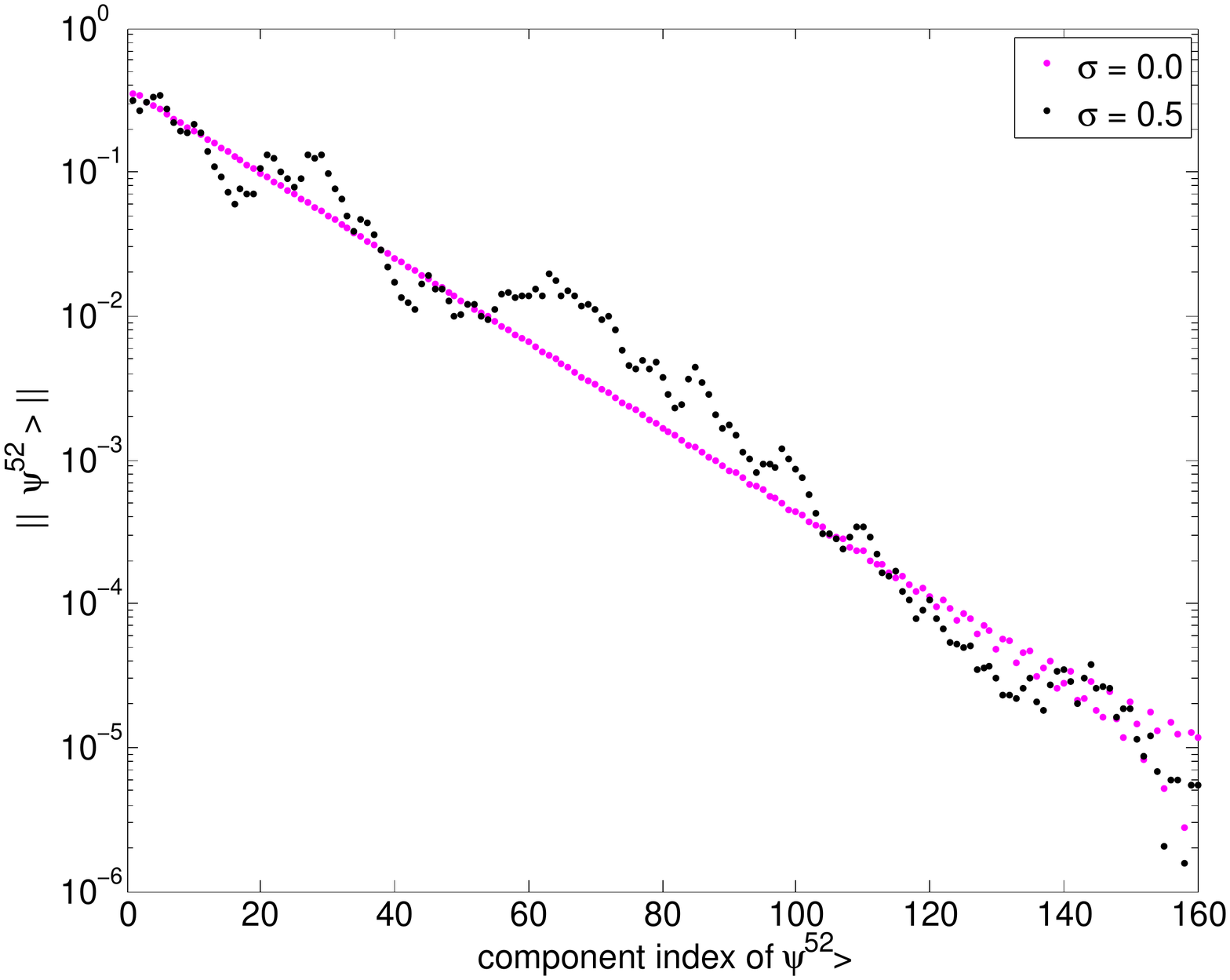}\includegraphics[scale=0.22]{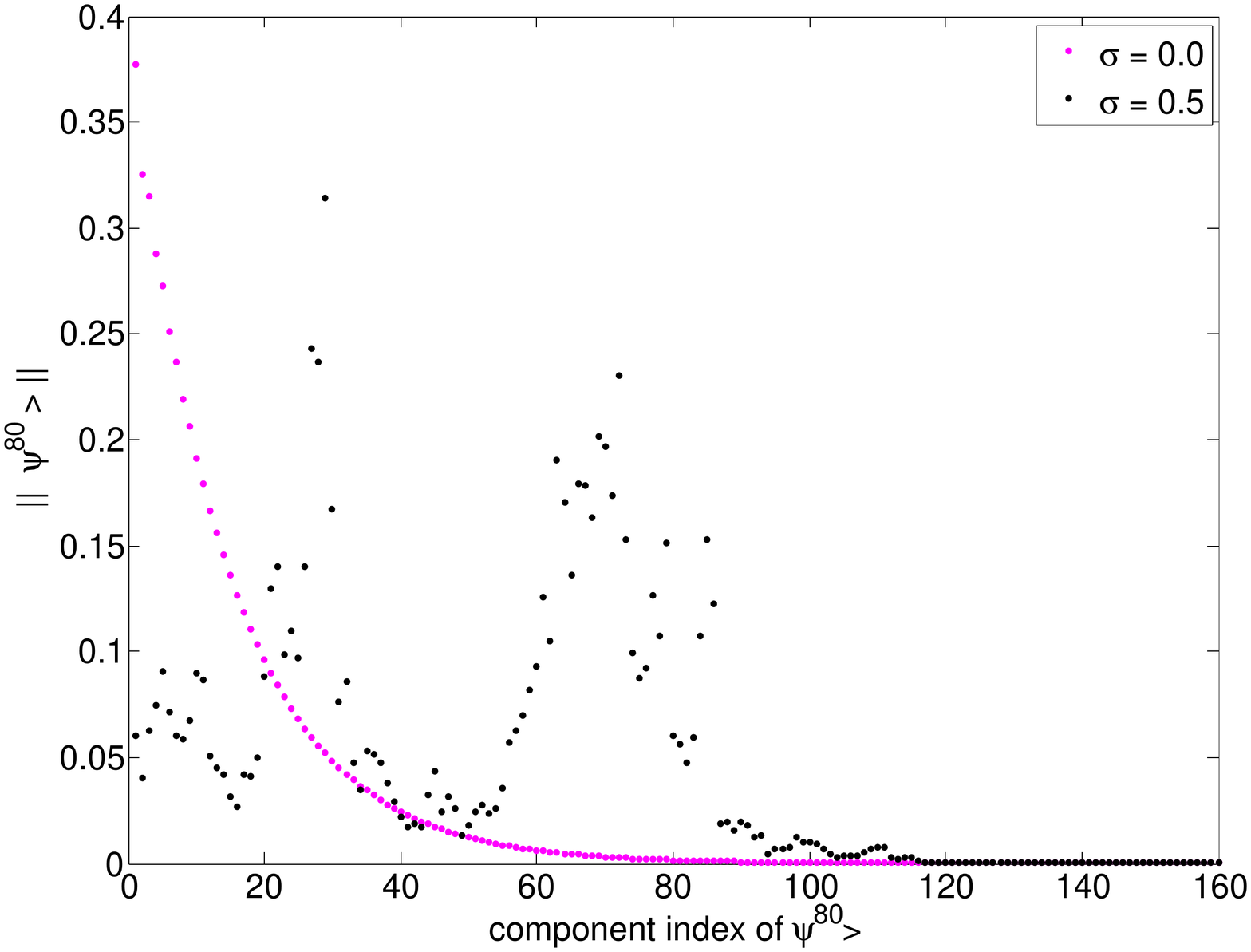}\includegraphics[scale=0.22]{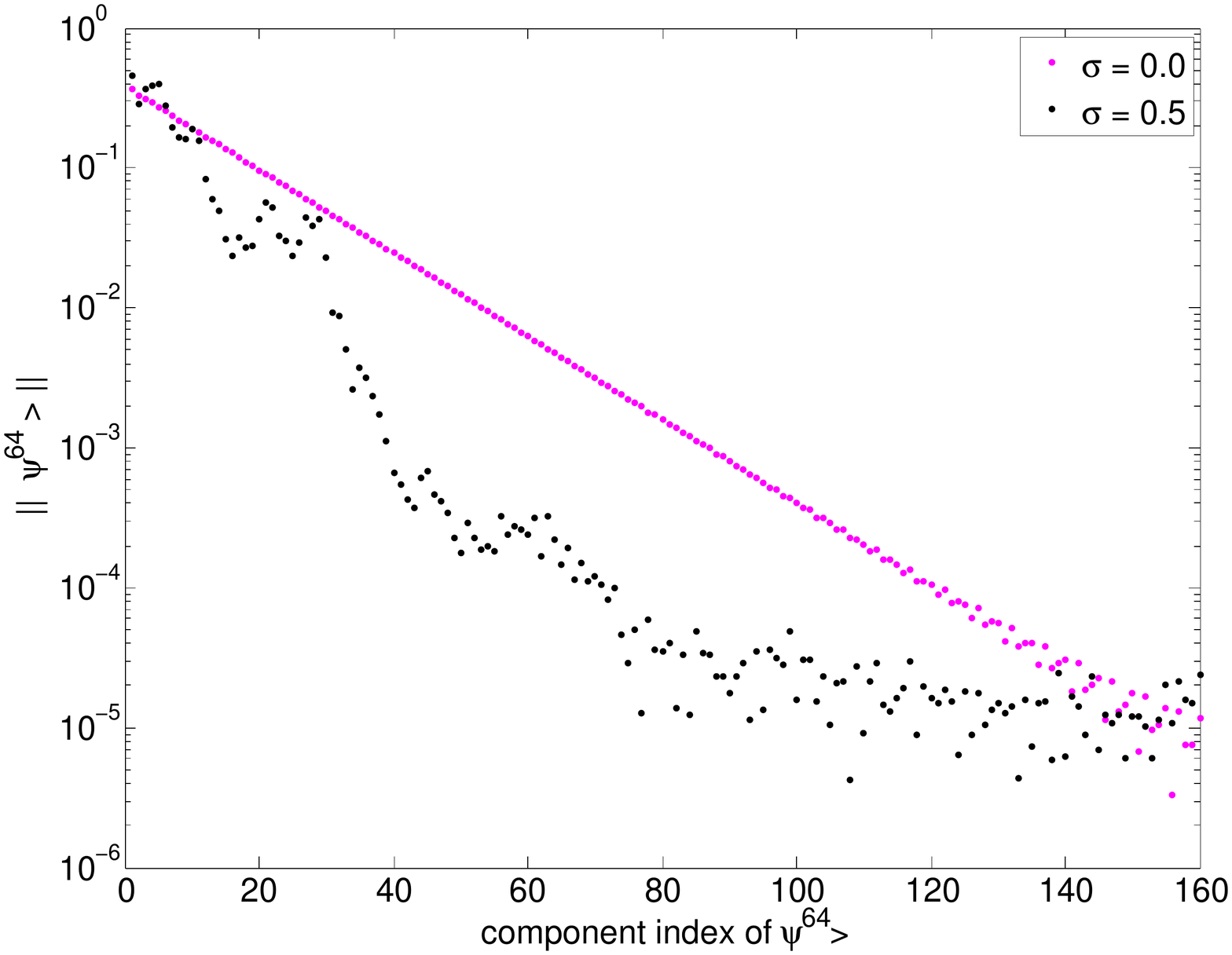}\caption{\label{fig:Localization-3Classes} Numerical observations. Left: Bulk
eigenvectors show exponential decay from the boundary (semi-log).
Middle: Type I runaways are most delocalized and have an algebraic
decay from an interior point (linear plot). Right: Type II runaways
are most localized (semi-log) and show an exponential decay
from the boundary (semi-log). }
\end{figure}

The eigenvectors are indexed by the real part of $p$ discussed above
and the numbering is in one to one correspondence with the eigenvalues
and condition numbers in the previous plots. As seen in Fig. \ref{fig:The-localization-of},
the unperturbed matrix $T$ has eigenvectors that are most localized
for eigenvalues near the real axis and on the right hand side of the spectrum (see Figs. \ref{fig:LocalizationColor}
and \ref{fig:TDisorder}), where $0\ll \ell \ll n-1$. 
On the other hand, the eigenvalues near the real axis but on the left correspond to values of $\ell$ near $0$ and $n-1$. We find that they are more delocalized in comparison. 

 Previously, it was shown that the entries of the eigenvectors in for $0\ll \ell \ll n-1$ have an exponentially decaying entries with the maximum at $i=0$  \cite{Dai2009}. The eigenvectors of $T(\sigma)=T+\sigma V$ are, as expected, simple deformations of the
unperturbed eigenvectors. They exponentially decay and show
a similar trend to $\sigma=0$ counterparts (Fig. \ref{fig:Localization-3Classes}, on the left).

For the perturbed matrix $T(\sigma)$ we find that what used to be
the most localized states become the most delocalized, i.e., the eigenvectors
corresponding to runaway type I eigenvalues have algebraic decays
(Fig. \ref{fig:Localization-3Classes}, in the middle), yet the unperturbed part
is more localized than the others.
Moreover, they have their maxima at an interior point, i.e., $\psi_i^\ell$ is maximum for an $i$ that is  $0 \ll i \ll n-1$.

This is surprising when one thinks of the Anderson model, where the
eigenvalues are all real and disorder inevitably causes localization.
The tendency for the disordered Fisher-Hartwig Toeplitz matrix is
reversed in that respect. Namely, the eigenvalues that are near the real line and end up real as a result of attraction, have eigenvectors that are {\it less} localized than the corresponding unperturbed eigenvectors.

These eigenmodes resemble those of \textit{twisted} Toeplitz matrices
\cite{TrefethenEmbree2005}, though the setting is different.

Most notable are runaway type II, which because of perturbation, show
significantly higher localization in comparison to the unperturbed
states (Figs. \ref{fig:LocalizationColor}, and \ref{fig:Localization-3Classes} on the right). 

Lastly the very localized eigenvectors near $\ell=0$ and $\ell=n$ 
remain localized and do not exhibit large variations in response to
perturbations; for these IPR and $-H$ show a small discrepancy in
quantifying localization. 

\subsection{Connections between eigenvalues and eigenvectors }
\begin{figure}
\begin{centering}
\includegraphics[scale=0.3]{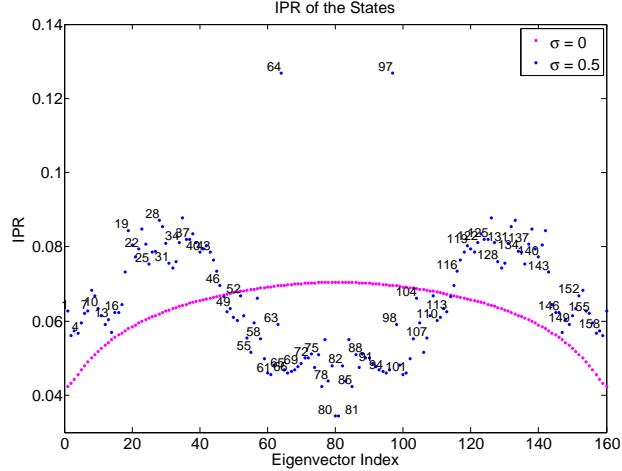}
\par\end{centering}
\caption{\label{fig:WholeStory}Here we show the close connection between localization
index and the condition number of $T\left(\sigma\right)$. The runaway types I and II correspond to the most delocalized and
localized states respectively (compare with the figure on the right in Fig. \ref{fig:Runaways-type-II}). }
\end{figure}

In the same vein as the eigenvalues, one wants to classify the eigenvectors of $T$ by their localization behavior in response to the diagonal perturbation.  We find that entropy and IPR are agreeable measures of localization of the eigenvalues of the disordered
matrix $T(\sigma)$. This is summarized in Fig. \ref{fig:WholeStory}. 

Interestingly, the condition number analysis not only predicts the various runaway behavior, it
is in one to one correspondence with the localization measures. The
ill conditioned eigenvalues end up having eigenvectors that are very
localized. Whereas, the runaway type I eigenvalues that tend to act
more like a normal eigenvalue (well conditioned) in response to perturbation become more
\textit{de}localized. Lastly, the remaining eigenvalues and eigenvectors
are well captured by second order perturbation theory.  This is summarized in the table appearing in Subsection \ref{sec:summary}.
\vspace{0.5cm}
\section{Summary and Future Work}

In this work we analyzed the eigenvalues and eigenvectors of $T+\sigma V$,
where $T$ is an $n\times n$ Toeplitz matrix generated by a Fisher-Hartwig
singular symbol, $V$ is a diagonal random perturbation and $\sigma$
is a real positive parameter quantifying the strength of the perturbation.
For $\sigma=0$, based on the Wiener-Hopf factorization technique
we showed that the eigenvalues, for sufficiently large $n$, are $E^{\ell}=a\left(e^{-ip^{\ell}}\right)$
where $p^{\ell}$'s are the complexed value ``momenta'' that we
analytically obtained in the asymptotic limit. The real part of $p^{\ell}$'s
are uniformly distributed in the interior of the spectrum and serve
as a good index set of the complex-valued eigenpairs. The right (left)
eigenvectors are exponentially decaying from the left (right) boundary.
In addition to solving for the left eigenvectors, we worked out the
trace, determinant and asymptotic form of the entries of $T$.

We have found a number of surprising features of the eigenpairs in
response to diagonal perturbations ($\sigma>0$). We find that there
are three classes: 1. The bulk eigenvalues and eigenvectors that are
well captured by second order perturbation theory of non-hermitian
matrices. The eigenvalues experience a compression (stretch) along
the imaginary (real) axis. The eigenvectors experience random deformations
but their localization behavior is  similar to the unperturbed ones. 2. The runaway type I,
which are the first class of non-perturbative eigenvalues. We proved
that they are caused by the attraction of complex conjugate pairs;
the eigenvalues close to the real axis attract most strongly till
they collide and become real. Surprisingly, the corresponding eigenvectors
become less localized and show algebraic decays with their maxima
in the interior, in contrast to the exponential decay from the boundary
of the unperturbed counterparts. 3. The runaway type II, which are
the second class of non-perturbative eigenvalues, move rapidly and
are predicted by their high condition numbers. The ill-conditioning
can be predicted from the first order perturbation theory for $\sigma\ll1$.
The corresponding eigenvectors show even stronger localization at
the boundary as compared to the unperturbed counterparts. The localization
was computed using both the inverse participation ratio and the entropy.
We found that the well- and ill-conditioning of the eigenvalues was
in a one to one correspondence with their less and more localization
of the corresponding eigenvectors respectively. Despite these new
findings, there is much left to be investigated. Open problems and
future work include:
\begin{enumerate}
\begin{singlespace}
\item We suspect that much of the work herein is directly applicable to
other Toeplitz matrices generated by other symbols.
\item Recall that the eigenvalues attract but they stay with the bulk till
they get close to the real line and then the complex conjugate pairs collide and
become real. A better understanding of the transition through the
degeneracy at the moment of collision is called for when before and
after the collision the spectrum is simple. 
\item Proof of the observed eigenvector localizations of type I and II Runaways.\end{singlespace}

\end{enumerate}
\begin{singlespace}
\vspace{0.5cm}

\end{singlespace}
\begin{acknowledgments}
Acknowledgements-- RM thanks Estelle Basor for discussions. The work was supported in part by the Chicago MRSEC grant, NSF grant number 0820054. RM thanks the James Franck Institute at  University of Chicago and the Perimeter Institute in Canada for their hospitality during the
summer of 2013. RM acknowledges the support of the NSF-DMS grant number
1312831, AMS-Simons travel grant and thanks the Goldstine Fellowship at IBM Research for the support and freedom. 

\vspace{0.5cm}

\end{acknowledgments}

\bibliographystyle{apsrev4-1}
\bibliography{mybib}

\end{document}